\documentclass[twocolumn]{article}
\usepackage[text={6.8in,8in},centering]{geometry}
\usepackage{xcolor}
\usepackage{amsmath}
\usepackage{placeins}
\usepackage[affil-it]{authblk}
\makeatletter
\renewcommand{\AB@affilsepx}{\qquad}
\makeatother
\usepackage{caption}
\usepackage{graphicx, yquant, tikz} 
\usepackage{pgf}
\usepackage{booktabs}
\usepackage{amsthm}
\usepackage{amsfonts}
\usepackage{amssymb}
\usepackage{mathtools}
\usepackage{array}
\usepackage{bm}
\usepackage{bbm}
\usepackage{multirow}
\usepackage[colorlinks=true, linkcolor=blue, citecolor=blue, urlcolor=blue]{hyperref}
\usepackage{cleveref}
\usepackage{qcircuit}
\usepackage{subcaption}

\DeclareCaptionLabelFormat{figlabel}{\textbf{Figure~#2}}
\newtheorem{theorem}{Theorem}
\newtheorem{lemma}[theorem]{Lemma}

\newtheorem{remark}[theorem]{Remark}

\newcommand{\E}{\mathbb E}

\newcommand{\F}{\mathbb{F}}

\newcommand{\bra}[1]{\langle #1|}
\newcommand{\braket}[1]{\langle #1\rangle}
\newcommand{\ket}[1]{|#1\rangle}

\definecolor{grey}{rgb}{0.5,0.5,0.5}
\yquantdefinebox{vdots}[inner sep=4pt]{$\vdots$}
\yquantdefinebox{blankS}[draw=none,fill=none,inner sep=0pt]{$\phantom{S}$}

\definecolor{ibmblue}{HTML}{002D9C}
\definecolor{ibmred}{HTML}{da1e28}
\definecolor{ibmpurple}{HTML}{8A3FFC}
\definecolor{ibmcyan}{HTML}{00539A}
\definecolor{ibmgreen}{HTML}{6fdc8c}
\definecolor{ibmteal}{HTML}{08bdba}

\DeclareCaptionLabelFormat{figlabel}{\textbf{Figure~#2}}
\title{Sampling hard circuits with verifiably high fidelity}

\author[1]{Simon Martiel\textsuperscript{*}}
\author[1]{Jay-U Chung\textsuperscript{*}}
\author[1]{Alireza Seif}
\author[2]{Soumik Ghosh}
\author[1]{Ian Hincks}
\author[1]{\authorcr Abhinav Deshpande}
\author[3]{Hidetaka Manabe}
\author[4]{Hanfeng Gu}
\author[3]{Feng Pan}
\author[2]{\authorcr Bill Fefferman}
\author[1]{Jay M. Gambetta}
\author[1]{Ali Javadi-Abhari\textsuperscript{\(\dagger\)}}
\affil[1]{IBM Research}
\affil[2]{University of Chicago\protect\\}
\affil[3]{Singapore University of Technology and Design}
\affil[4]{NVIDIA Corporation}

\date{} 

\begin{document}

\twocolumn[
\maketitle
\begin{center}
\begin{minipage}{0.9\textwidth}
\noindent
Sampling-based proposals are prominent candidates for demonstrating quantum computations beyond the reach of classical supercomputers. However, it has been difficult to combine their complexity-theoretic hardness with two capabilities needed for scalable quantum computing more generally: suppressing hardware errors, and verifying the quantum computation itself. Here we address both issues by introducing structured circuits, which, in addition to provable hardness guarantees, admit an encoding in a quantum code. This allows us to simultaneously reach high fidelities at high circuit depths, and to certify an experimental fidelity via the circuit structure and measurement of code syndromes. The resulting certificate is device dependent, but requires substantially weaker noise assumptions than existing fidelity proxy benchmarks. We demonstrate our proposal with a $64$-qubit, depth-$73$ Clifford circuit, doped with $314$ $T$ gates. We use a total of $76$ physical qubits to encode this computation in spacetime codes, effectively suppressing gate error rates by $10\times$ after syndrome post-selection, and yielding a state with a fidelity lower bound of $0.349$ with $95\%$ confidence. Our construction is a systematic method for promoting a stabilizer state to a magic state while keeping an error-detected fidelity certificate.
\end{minipage}
\end{center}
\vspace{.5em}
]

\begingroup
\renewcommand{\thefootnote}{\fnsymbol{footnote}}
\footnotetext[1]{These authors contributed equally.}
\footnotetext[2]{
\href{mailto:ali.javadi@ibm.com}{ali.javadi@ibm.com}}
\endgroup

\begin{figure*}[t]
\centering
\includegraphics[width=.99\textwidth]{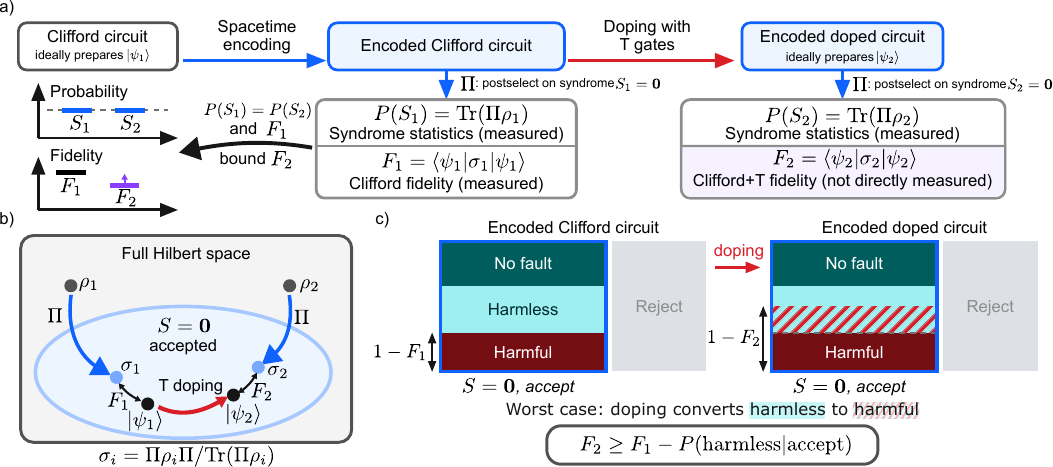} 
\caption{
{\bf  Bounding the fidelity of hard quantum states.} {\bf a.} An encoded Clifford circuit's state fidelity $F_1$ and its syndrome distributions $S_1$ are measured efficiently. The circuit is doped with $T$ gates that commute with code stabilizers, yielding statistically indistinguishable syndromes $S_2$. This leads to a lower bound on the fidelity $F_2$ of the (hard) doped circuit. {\bf b.} The code projects both states to the same subspace, and doping only rotates within this subspace. {\bf c.} A lower bound on $F_2$ can be efficiently estimated from the classification of Pauli faults affecting the Clifford circuit. The set of rejected faults and the no-fault probability remain invariant with doping. Any loss in fidelity is only possible due to the conversion of harmless faults, which are rare in random Cliffords.}\label{fig:fig1}
\end{figure*}

\begin{figure}[t]
\centering
\includegraphics[width=1\columnwidth]{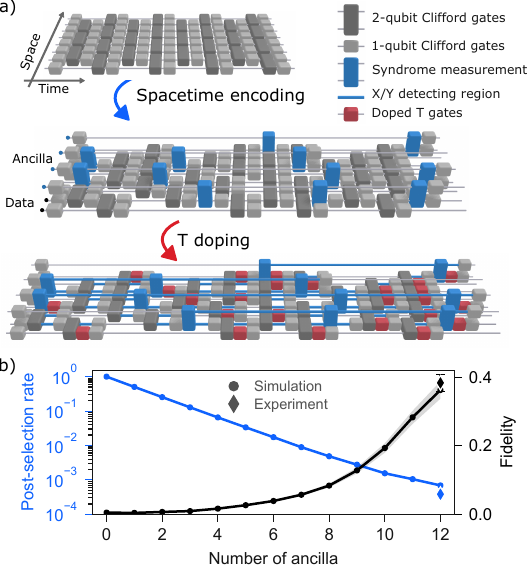} 
\caption{
{\bf  Constructing doped Clifford circuits within spacetime codes.} {\bf a.}  A highly entangling Clifford circuit is encoded in a spacetime code and doped with $T$ ($Z$-rotation) gates on wires not subject to $X$ or $Y$ error detection, producing an encoded state with high entanglement and magic. {\bf b.} Additional ancillas improve fidelity at the cost of lower post-selection rates. Pauli-noise simulations agree well with experiment, assuming $0.18\%$ $CZ$ error, idle process fidelity $e^{-t/\tau}$ with $\tau=170 \ \mu$s, and $0.1\%$ readout error.}\label{fig:fig2}
\end{figure}

\begin{figure*}[t]
\centering
\includegraphics[width=\textwidth]{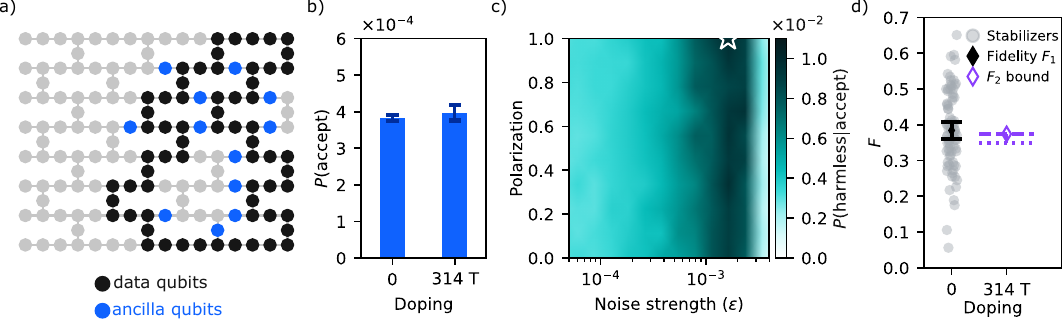} 
\caption{
{\bf  Quantum advantage experiment with a fidelity bound.} {\bf a.} Qubit layout on the IBM Boston superconducting processor, consisting of $64$ data and $12$ ancilla qubits. {\bf b.} Zero-syndrome probability $P(\mathrm{accept})$ is similar before and after doping, consistent with $T$ gates being noiseless.  {\bf c.} Estimated probability of harmless, accepted errors under various noise strengths and polarizations. The maximum possible fidelity drop that doping can cause due to conversion of harmless faults is $0.010(1)$. {\bf d.} Direct fidelity estimation of the Clifford state (0 doping) shows fidelity $F_1 = 0.38(1)$. We lower bound the fidelity of the hard doped state to $F_2 \geq 0.349$ (dotted line), and collect 1,389 post-selected samples in a total runtime of $14.5$ minutes. Error bars show normal $95\%$ confidence intervals.}\label{fig:fig3}
\end{figure*}

\begin{figure*}[t]
    \centering
\includegraphics[width=\linewidth]{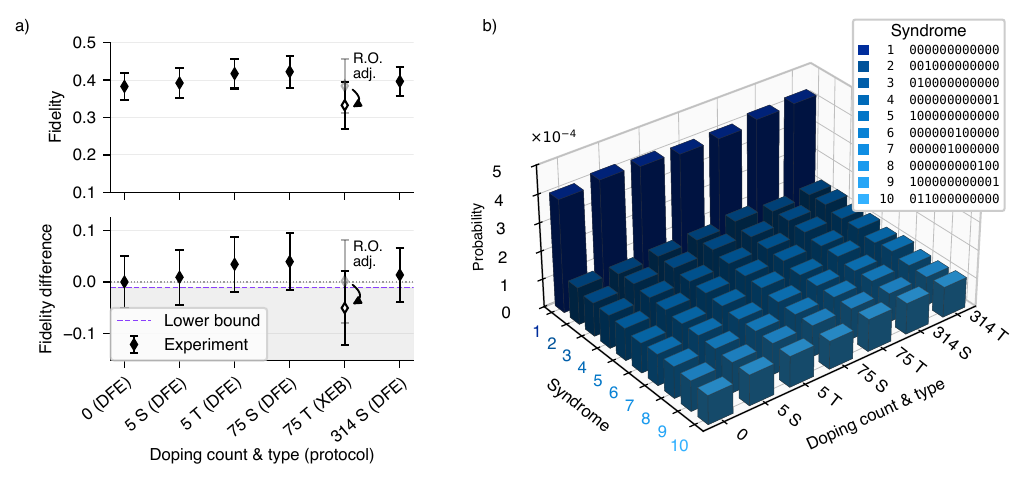}
\caption{{\bf Validation of fidelity lower bound.} {\bf a.} Various doping counts (5, 75, 314) and types ($T$ or $S$ gates) are chosen that permit efficient measurement of fidelity (DFE) or fidelity proxy (XEB). XEB is known to overestimate fidelity due to readout errors, so we also show a rescaled quantity that controls for this effect. All experiments show that our estimated lower bound accurately captures the maximum loss in fidelity due to doping. Error bars show $95\%$ confidence intervals. {\bf b.} Syndrome probability distributions show high consistency, regardless of doping strategy. This is consistent with a similar noise environment affecting all experiments, and doping gates being noiseless. }
\label{fig:fig4}
\end{figure*} 

\noindent Quantum processors have entered a distinct regime where their increasing size and precision now make it possible to perform computations that are difficult to reproduce classically. This provides an experimental probe of quantum complexity theory, testing the prediction that quantum computers achieve an exponential runtime separation over classical computation for certain tasks~\cite{bernstein1993quantum}. A prominent path to quantum advantage, supported by strong theoretical evidence for classical hardness, is to sample from the output distributions of certain circuit families~\cite{Bremner_2010,Aaronson2017,Movassagh2023}.

Beyond a complexity-theoretic foundation, however, a scalable demonstration of quantum advantage must satisfy two further criteria. First, increasing circuit size must be matched with a corresponding suppression of noise, such that the experiment evades classical simulation as it scales~\cite{Dalzell_2022,schuster2025polynomial}. Second, one must be able to trust that these complex calculations were performed with high fidelity, even in the presence of noise~\cite{harrow2017quantum,lanes2025framework}. Both requirements have been difficult to satisfy in existing sampling experiments: hardware-friendly proposals~\cite{Boixo2018} are poorly matched to low-overhead protection by encoding, whereas more structured proposals~\cite{Bremner_2010,jozsa2013classical} are difficult to map to near-term hardware in classically-hard regimes~\cite{hangleiter2404fault}. Verifying these experiments has its own scaling problem: resolving fidelity in broadly distributed outputs requires a prohibitive number of samples~\cite{Dalzell_2022, Deshpande_2022, hangleiter2019sample}, and more sample-efficient benchmarks rely on strong classical simulation and strong assumptions about hardware noise~\cite{Morvan2024, GaoKalinowski2024, fefferman2023effectnonunitalnoiserandom, ware2023sharp}.

Here we develop a protocol for quantum advantage based on {\em doped Clifford sampling (DCS)}, and prove its asymptotic hardness. We show that the circuit structure enables simultaneously larger circuit volumes, higher fidelities, and an experimentally accessible fidelity certificate. Our construction is aided by the use of spacetime codes for Clifford circuits~\cite{bacon2017sparse, delfosse2023spacetime} and a method for injecting non-Cliffordness that preserves the code. The spacetime code detects errors during computation, increasing fidelity at the cost of a sampling overhead. Leveraging the fact that the fidelity of Clifford states is efficiently measurable, and by doping such that code syndromes are preserved, we are able to sample from hard quantum states and lower bound their fidelity, even when classical simulation is infeasible.

We experimentally demonstrate our protocol on a superconducting quantum processor using a $64$-qubit, depth-$73$ Clifford circuit doped with $314$ T gates, encoded in a total of $76$ physical qubits. Syndrome post-selection produces a Clifford reference state with fidelity $0.38(1)$, from which we derive a $95\%$-confidence lower bound of $0.349$ for the fidelity of the classically hard doped state. We validate the fidelity lower bound in independently measurable regimes at low and intermediate $T$ counts. Our experiment generates roughly 1,400 samples in 14 minutes, which we estimate to be infeasible for current tensor-network and stabilizer-based algorithms. It is reasonable to expect that the time to simulate this particular instance will change with better simulation methods and improved classical hardware. Nevertheless, our primary goal is to demonstrate that algorithm-agnostic hardness guarantees can coexist with trusted execution of deep, error-detected circuits on current hardware.

Figure~\ref{fig:fig1} summarizes how the fidelity certificate is obtained in a DCS workflow. We begin with a random Clifford circuit, whose output state fidelity can be measured and verified efficiently using direct fidelity estimation (DFE)~\cite{direct_fid_est,dfedasilve}. Notably, DFE makes no assumptions about hardware noise. We then encode this circuit in a spacetime code, where post-selection on the zero syndrome projects the state into a low-error subspace~\cite{error_detection}.  Lastly, we increase the amount of non-Cliffordness (``magic'') by carefully doping with $T$ gates while preserving the code structure. This ensures that the two circuits (i.e., before and after doping) detect the same set of faults. We experimentally confirm that the syndrome distributions remain nearly identical. Together, these properties let us obtain a lower bound for the doped circuit's fidelity, even though it is not directly measurable or simulable. We do so by observing that quantum states before and after doping are both projected into the same subspace, and that the only loss in fidelity can come from a conversion of previously-harmless faults (i.e., faults that stabilize the Clifford circuit) to harmful faults. We show that the probability of harmless faults in random Clifford circuits is asymptotically negligible, and can be estimated efficiently for a given circuit instance (Supplementary Information \ref{sec:sfe}).

This method of lower bounding fidelity is in contrast to existing approaches which largely verify such experiments by a fidelity proxy~\cite{hangleiter2026has}: computing statistical properties of the observed samples that are expected to track fidelity under particular noise assumptions. For example, cross-entropy benchmarking (XEB) requires sufficiently weak, spatially homogeneous, and temporally independent noise~\cite{GaoKalinowski2024,ware2023sharp,dalzell2024random}, while mirror benchmarking requires Markovian, gate-independent noise that is uniform on each layer~\cite{mayer2021theory}. By contrast, we permit strong, non-uniform, and spatiotemporally correlated noise. We also do not impose anticoncentration or expensive classical compute requirements. We assume a general Pauli noise model, which is attainable by twirling the experiment's noisy gates~\cite{wallman2016randomized}.

Figure~\ref{fig:fig2} describes the circuit and code construction (also see Supplementary Information \ref{sec:circuit_implementation}). We begin with a brickwork Clifford skeleton of similar width and depth, constructed on qubits connected in a cycle. It has been shown that a linear-depth circuit is sufficient to create an arbitrary stabilizer state on a cycle~\cite{goubaultdebrugiere2026toappear}, and that this connectivity quickly frustrates tensor-network simulation methods~\cite{manabe2026classical}. At the same time, this Clifford skeleton has inherent Pauli symmetries which can be measured with spacetime Pauli checks using additional ancilla qubits~\cite{error_detection}. Any Pauli operator with arbitrary support in space and time that stabilizes the output of the circuit is a valid check, and its measurement will detect faults that anticommute with that Pauli operator. 

Discovering good symmetries, however, is a difficult task for a large circuit: some are easier to measure owing to their Pauli weight and the connectivity of ancilla to data qubits, while the fraction of errors they detect can also vary considerably. We use fast decoder-based heuristics to search in the Pauli group generated by valid checks, yielding low-overhead syndrome measurements with large detecting regions. In the $64$-qubit Clifford circuit used in this work, the construction approaches distance-2 error-detection behavior using only $12$ ancillas, detecting $93\%$ of all first-order errors. We measure a $131\times$ increase in state fidelity compared to the unencoded circuit, at the cost of a $990\times$ decrease in effective sampling rate. This slowdown is tolerable for superconducting platforms and still leaves them comparatively faster than trapped-ion and neutral-atom systems~\cite{ransford202698, bluvstein2023logical}.

While the encoded Clifford circuit provides high entanglement and fidelity, non-Clifford resources are required for classical hardness. Existing approaches initialize or measure the circuit in non-Clifford bases and achieve theoretical hardness~\cite{jozsa2013classical,bouland2017complexity,yoganathan2019quantum,ghosh2023complexity}, but require a large problem size for practical quantum advantage since they only introduce an amount of magic linear in qubit count. We can improve this to quadratic without compromising the structure of the code, and its error detection capability, by using the spacetime picture.

To do this, we compute the detecting region of each check by back-propagating its ancilla Pauli measurement through the Clifford circuit and recording its path. This yields a Pauli operator supported on all $O(n^2)$ spacetime locations of the circuit, which determines the faults detectable by that ancilla measurement. If a local Pauli error \(P\in\{X,Y,Z\}\) remains undetected by all ancillas, then we can just as well intentionally inject a $P(\theta)$ rotation at that location, as it preserves all measured stabilizers by commutation. This doping dramatically increases the simulation hardness of the circuit, yet retains the same fidelity gains provided by the code.

While the above construction permits any non-Clifford rotation as long as it commutes with checks, $Z$ rotations (including $T$ gates) are special in our hardware: they are implemented by virtual frame tracking and do not add extra noise~\cite{virtualz, ball2016role}.  This means that $T$ doping does not introduce new errors and only changes how existing errors propagate. This fact, together with the fact that $T$ gates are placed in stabilizer-commuting locations, controls the magnitude of undetected errors, leading to the advertised fidelity bound.

We realize the DCS protocol in a $76$-qubit experiment designed to certify a high-fidelity quantum computation in a classically challenging regime. We take several steps to reduce hardware noise, leading to better fidelity and post-selection rates (see Supplementary Information \ref{sec:experiments} and \ref{sec:hardware}). We choose a subset of qubits with low gate errors, arranged as a cycle of data qubits with dangling ancillas. The spacetime Pauli checks are optimized by a randomized search to minimize undetected error channels. Lastly, we calibrate gates and readouts for optimal performance in our specific circuit, discard shots where non-Markovian errors are detected~\cite{olepaper}, and perform gate Pauli twirling to tailor the noise toward a stochastic Pauli channel.

Our experiment is described in Figure~\ref{fig:fig3}. The full physical circuit contains 2,644 $CZ$ gates: 2,336 belong to the depth-73 DCS computation and the remaining 308 couple the ancillas for syndrome extraction. We measure a fidelity of $0.38(1)$ for the stabilizer state, which we obtain by DFE using $80$ randomly chosen stabilizers and 250,000 shots per stabilizer, at a post-selection rate of $3.83(4)\times10^{-4}$. Using a readout mitigation strategy~\cite{PhysRevA.105.032620}, we estimate the true state fidelity to be $0.55(2)$. We inject 314 T gates, the maximum permitted by the code, and measure a statistically indistinguishable probability of acceptance between the Clifford and the $T$-doped circuits. This lets us lower bound the fidelity of the $T$-doped state by estimating the conditional probability of harmless faults among accepted faults in the Clifford circuit. To obtain this estimate, we perform a Monte Carlo simulation sweeping various Pauli noise strengths and polarizations, classifying each fault configuration by whether it is detected by the code and whether it stabilizes the state. We numerically estimate a maximum possible loss in fidelity of $0.010(1)$ due to doping.  Subtracting the upper confidence limit on this loss from the lower confidence limit on the DFE fidelity yields the conservative bound $F_2 \geq 0.349$ for the $T$-doped state at the $95\%$ confidence level.

\begin{figure}[!t]
\centering
\includegraphics[width=\columnwidth]{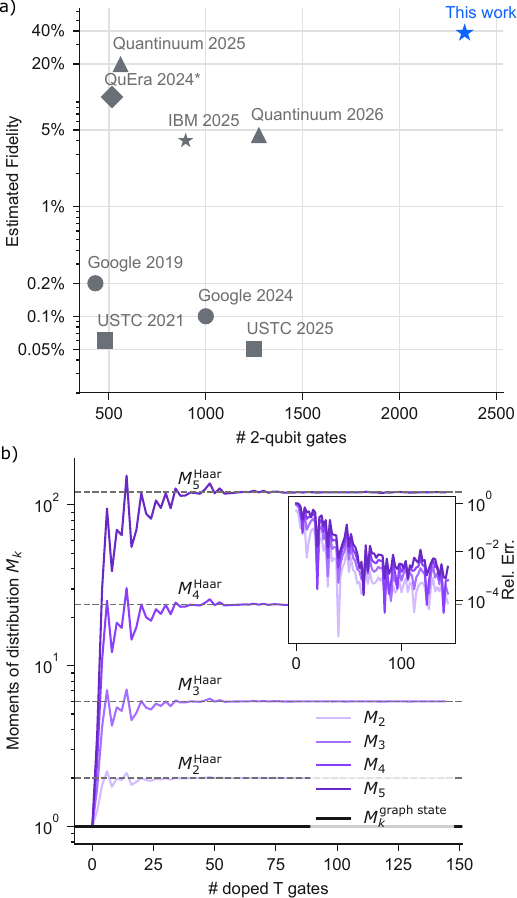} 
\caption{
{\bf DCS in the landscape of quantum sampling experiments.} {\bf a.} Spacetime encoding increases circuit size and fidelity, compared to prior work (asterisk indicates counting compiled $CCZ$ gates). {\bf b.} Similar to prior proposals, DCS distributions converge in moments to the Haar distribution. Simulations shown for a $24\times24$ doped Clifford circuit encoded using $12$ ancillas.
}\label{fig:fig5}
\end{figure}

We can independently validate our estimated lower bound on fidelity, since it holds after doping with any $Z$-rotation and any amount of doping. Figure~\ref{fig:fig4} summarizes these validation experiments. In the low-magic regime, DFE remains tractable since the state retains a sparse Pauli description (see Supplementary Information~\ref{sec:dfe}). Added $T$ gates can spread observable weight among an exponentially larger set of Paulis, making it harder to both classically importance-sample the Paulis and to experimentally resolve their expectation values. Nevertheless, we perform DFE after doping with $5$ $T$ gates and show agreement with our bound. As a separate check, we use $S$ gates for doping in more locations and measure DFE on these new Clifford circuits, again supporting the validity of the fidelity bound. 

At medium amounts of doping with $O(n)$ $T$ gates, a different validation strategy becomes available. In this regime, the output distribution begins to anticoncentrate, and its second moment approaches that of Haar-random states. Consequently, XEB becomes a calibrated proxy for fidelity (see Supplementary Information~\ref{sec:xeb}). At the same time, state-of-the-art Clifford + $T$ simulators are still capable of strong simulation with around 75 T gates~\cite{kissinger2022simulating}. These two facts together enable an independent validation via XEB, albeit with substantial computational effort and stronger assumptions on noise. We observe a similar validation of our fidelity bound in this regime. Lastly, we observe excellent agreement among the full distribution of syndromes obtained from all circuits, strongly supporting the assumption that doping gates are noiseless and the noise environment affecting all circuits is the same. Our experimental methodology has facilitated this by interleaving all shots over the span of a $5$-hour experiment, which distributes drifting noise similarly among shots.

Figure~\ref{fig:fig5} places our work in the context of previous sampling experiments. Random circuit sampling (RCS) is readily implemented using hardware-native gates~\cite{Arute2019, Wu2021, Morvan2024, DeCross2025, gao2025establishing, merkel_clifford_benchmarks, ransford202698}, but does not possess enough structure to easily admit software-level protections. This means that RCS has primarily been used as a benchmark of physical qubits and gates. More structured proposals such as IQP sampling can be encoded in error-detecting codes but remain difficult to scale~\cite{bluvstein2023logical,hangleiter2404fault}, and their algebraic structure can expose specialized classical attacks~\cite{codsi2023classically,maslov2024fast,nelson2026polynomial}. Our proposal uses a low-overhead encoding to move the accessible regime toward larger circuits and higher fidelities. Importantly, the circuit structure does not compromise anticoncentration, which is a desirable property of hard sampling proposals~\cite{hangleiter2018anticoncentration}: as the number of $T$ gates increases, we observe that low-order moments of the output distribution approach their Haar values~\cite{leone2026non}. DCS therefore combines the broad output statistics of hard sampling problems with circuit structure that can be exploited for error detection and fidelity certification.

In terms of classical hardness, our proposal is a compelling one for quantum advantage. Specifically, we show that our circuit ensemble is universal in the worst-case and that, assuming two appropriate complexity-theoretic conjectures---similar to the ones used in the theory of other sampling-based advantage proposals~\cite{aaronson2010computational, bremner2017achieving, bouland2019complexity}---no classical sampler exists for the output distribution in the average case (see Supplementary Information \ref{sec:asymptotic_hardness}). To quantify the hardness of simulation at the finite size accessed experimentally, we numerically investigated many leading algorithms and found them intractable (see Supplementary Information \ref{sec:simulation_hardness}). 

To see why, we note that our construction makes the classical cost due to entanglement and magic separately explicit: the former due to the high-depth Clifford backbone and the latter due to high $T$ count. In the worst case, Rank-width based tensor-network contraction of graph-like circuit representations scales exponentially as $O(2^{2r})$ with the width $r$ of a given rank decomposition~\cite{kuyanov2026efficientclassicalsimulationlowrankwidth}. While finding an optimal rank decomposition is hard~\cite{oum2017rank}, an extensive search over balanced cuts of the underlying graph state finds no bipartite entanglement entropy below $r=27$. Optimized contraction heuristics~\cite{Gray_2021} likewise struggle with the circuit's cylinder-like structure and large minimal separators---reducing the resulting width by slicing incurs a steep multiplicative cost because the bottlenecks are largely index-disjoint. Approximate MPS methods can avoid such costs, especially in 1D Haar-random circuits, by truncating rapidly decaying Schmidt spectra~\cite{PhysRevX.10.041038}. But our doped-Clifford circuits exhibit substantially flatter spectra, causing rapid loss of fidelity with truncation over the deep circuit. Stabilizer-based methods face a complementary obstruction: approximate simulation of $t$ $T$ gates is conjectured to require $\Omega\!\left(2^{0.228t}\right)$ stabilizer terms~\cite{bravyi2016improved}, which seems impractical for $t=314$. Hybrid methods can combine these approaches~\cite{ybnf-rjw8,clifft}, but are limited for highly entangled magic resources.

Noise can open additional routes for classical simulation, since sufficiently noisy circuits may be approximated by truncating many correlations or paths in the evolution~\cite{Aharonov_2023, Noh_2020}. This is precisely where circuit encoding can change the simulable regime. Simulations show that our logical circuit has been run at an effective $CZ$ error rate of $2\times 10^{-4}$---a $10\times$ improvement compared to physical $CZ$ rates on the device and far below state-of-the-art two-qubit error rates---likely making noisy approximations impractical in the regime accessed here. Finally, one might ask whether restricting $T$ gates to code-preserving locations inadvertently suppresses magic. Numerically we find no meaningful reduction in stabilizer entropy~\cite{leone2022stabilizer} relative to unrestricted doping, ruling out such simple shortcuts to simulation. We leave open the possibility of faster simulations with improved algorithms.

\paragraph{Outlook} Scalable quantum advantage hinges on the ability to remove errors from the computation as system sizes grow, and the ability to trust the computational output. Our work makes strides along this path. Two important open problems remain. On the scalability front, our approach uses error detection to push the computational boundary, achieving larger circuit volumes at higher fidelities. But error detection alone is not scalable indefinitely because of post-selection overhead. Yet with improved hardware, the same spacetime codes or other codes could be used to demonstrate fault-tolerant quantum advantage~\cite{delfosse2023spacetime, hangleiter2404fault}. On the verification front, our proposal relaxes noise assumptions by using circuit symmetries and syndromes to verify a quantum state~\cite{ringbauer2025verifiable,xiao2026situ}. However, a fully device-independent verification makes no quantum assumptions and only verifies the problem's output, proving that a quantum computer was used in the first place. This task is not efficient for random circuits with a broad output distribution, but recent progress in this direction has relied on planting secrets in quantum circuits that a classical verifier can check~\cite{aaronson2024verifiable, gharibyan2025heuristic,deshpande2025peaked}.

\paragraph{Data availability} All experimental data and the circuits used are available through Zenodo at \href{https://zenodo.org/records/21633064}{10.5281/zenodo.21633064}.

\paragraph{Note added} An earlier version of this manuscript demonstrated the DCS protocol on a one-dimensional circuit with open boundaries. Subsequent work exploited this geometry to perform tensor-network simulation using approximately 159 H100-GPU-hours, and independently verified the estimated quantum fidelity~\cite{manabe2026classical} (see Supplementary Information~\ref{sec:independent_validation}). The present manuscript instead uses periodic boundary conditions, eliminating the low-width contraction path exploited before. The theoretical framework is unchanged.

\paragraph{Acknowledgments} We thank Oliver Reardon-Smith, John van de Wetering, Matthew Sutcliffe, Bryan Clark, Zejun Liu, Xavier Waintal, Tristan Cam, Kevin Smith and Minh Tran for their valuable expertise on classical simulation algorithms. SM would like to thank Olivier Ezratty for stimulating discussions on quantum advantage criteria. We acknowledge helpful conversations with Abhinav Kandala, Samantha Barron, Ewout van den Berg, Sergey Bravyi, Dominik Hangleiter, Jonas Helsen, Liang Jiang, Blake Johnson, David McKay, Ilan Rosen, Liran Shirizly, Maika Takita, Kristan Temme and Nobuyuki Yoshioka. We thank the entire IBM Quantum team whose work enabled this study.

\FloatBarrier
\bibliographystyle{unsrt}
\bibliography{biblio}

\onecolumn
\clearpage

\setcounter{page}{1} 
\renewcommand{\thepage}{S\arabic{page}} 

\begin{center}
    {\LARGE\bfseries Supplementary Information}
\end{center}

\vspace{1em}

\tableofcontents
\clearpage

\setcounter{section}{0}
\setcounter{subsection}{0}
\setcounter{figure}{0}
\setcounter{table}{0}
\setcounter{equation}{0}
\setcounter{theorem}{0}

\renewcommand{\thesection}{S\arabic{section}}
\renewcommand{\thesubsection}{S\arabic{section}.\arabic{subsection}}
\renewcommand{\thefigure}{S\arabic{figure}}
\renewcommand{\thetable}{S\arabic{table}}
\renewcommand{\theequation}{S\arabic{equation}}
\renewcommand{\thetheorem}{S\arabic{theorem}}

\section{Constructing $T$-doped Clifford circuits within spacetime codes}\label{sec:circuit_implementation}

In this section we provide a detailed explanation of our circuit construction  methods.

\subsection{High-entanglement graph state generation}

To prepare highly entangled random graph states, we use a circuit ansatz shown in Figure~\ref{figure:base_circ}. This circuit begins with a layer of $-H-$ on every qubit, followed by a brickwork pattern alternating odd and even layers of $-CZ-$ gates on a loop of qubits, where each $-CZ-$ layer is followed by a layer of random single-qubit $-\sqrt{X}-$ or $-S-\sqrt{X}-$ gates. A final layer of single-qubit Cliffords is appended to the end to rotate the stabilizer state into a graph state.

We use qubits connected in a cycle graph, i.e., a one-dimensional chain with periodic boundary conditions. Such a connectivity admits a large number of ancilla, while frustrating tensor network contraction methods much more effectively than a chain with open boundary conditions~\cite{manabe2026classical}.

In the graph state picture, our circuit starts with the empty graph on $n$ qubits, repeatedly toggles the edges of the cycle graph, then applies local complementations or togglings of self-loops.

The circuits we use in this work have a $-CZ-$ depth slightly larger than $n$. It has been proven that depth $n$ is sufficient to prepare arbitrary graph states on a cycle architecture~\cite{goubaultdebrugiere2026toappear}, which gives us confidence that the circuit is expressive enough to generate graph states with maximum entanglement.
Indeed, we empirically find that we obtain random graph states that almost saturate the entanglement upper bound. For example, when computing the bipartite adjacency rank over $10^8$ random bipartitions in our $64$-qubit graph state, we obtain a minimum rank/bipartite entanglement entropy of $27$, which is close to the maximum possible rank of $32$.

While a cycle layout is used for $n$ logical qubits, an additional number of ancilla qubits are required for error detection. On the IBM Heron processor family, a typical 1-D chain alternates between degree-2 and degree-3 nodes in the qubit connectivity graph. Every degree-3 node allows us to attach an ancilla qubit and perform error detection on the logical circuit.

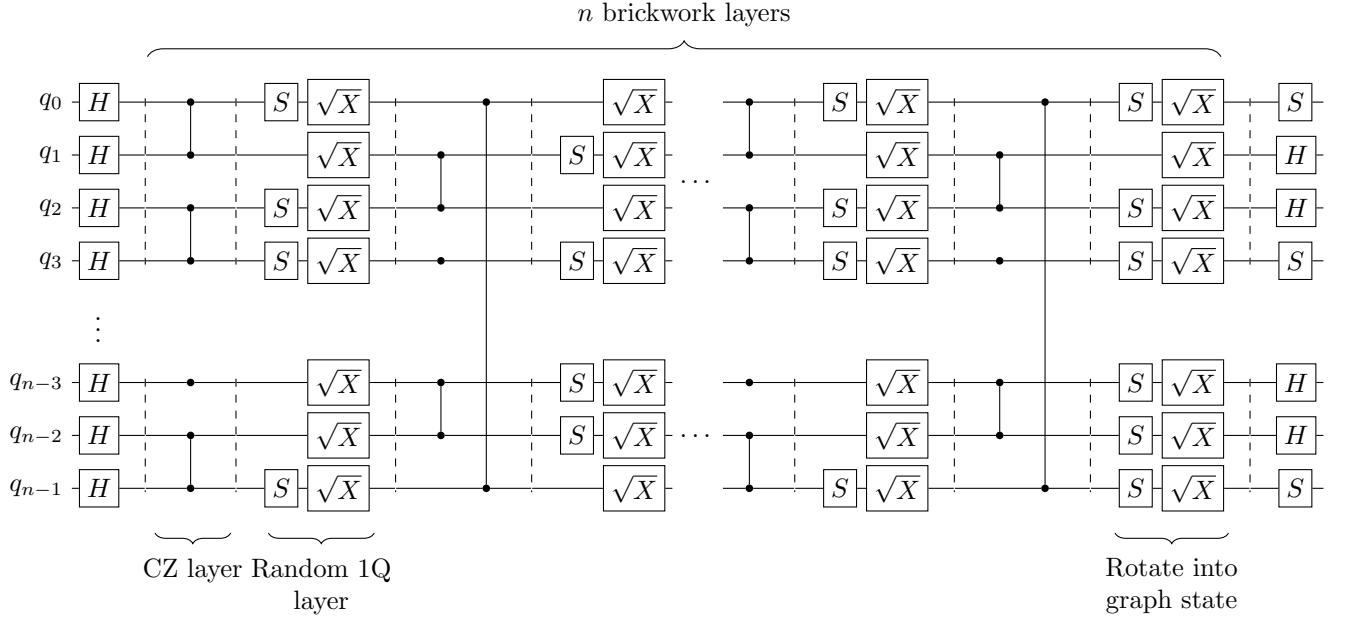
\begin{figure}[h]
\center

\begin{tikzpicture}
    \usetikzlibrary{decorations.pathreplacing}

    \begin{yquant*}
        qubit {$q_0$} q0;
        qubit {$q_1$} q1;
        qubit {$q_2$} q2;
        qubit {$q_3$} q3;
        nobit ellipsis;
        vdots ellipsis;
        qubit {$q_{n-3}$} qnm3;
        qubit {$q_{n-2}$} qnm2;
        qubit {$q_{n-1}$} qnm1;

        h q0;
        h q1;
        h q2;
        h q3;
        h qnm3;
        h qnm2;
        h qnm1;
        barrier (q0,q1,q2,q3);
        barrier (qnm3,qnm2,qnm1);

        zz (q1, q0);
        zz (q3, q2);
        zz (qnm3);
        zz (qnm1, qnm2);
        barrier (q0,q1,q2,q3);
        barrier (qnm3,qnm2,qnm1);

        box {$S$} q0;
        box {$\sqrt{X}$} q0;

        blankS q1;
        box {$\sqrt{X}$} q1;

        box {$S$} q2;
        box {$\sqrt{X}$} q2;

        box {$S$} q3;
        box {$\sqrt{X}$} q3;

        blankS qnm3;
        box {$\sqrt{X}$} qnm3;

        blankS qnm2;
        box {$\sqrt{X}$} qnm2;

        box {$S$} qnm1;
        box {$\sqrt{X}$} qnm1;
        barrier (q0,q1,q2,q3);
        barrier (qnm3,qnm2,qnm1);

        zz (q2, q1);
        zz (q3);        
        zz (qnm3,qnm2);
        zz (qnm1,q0);
        barrier (q0,q1,q2,q3);
        barrier (qnm3,qnm2,qnm1);

        blankS q0;
        box {$\sqrt{X}$} q0;

        box {$S$} q1;
        box {$\sqrt{X}$} q1;

        blankS q2;
        box {$\sqrt{X}$} q2;

        box {$S$} q3;
        box {$\sqrt{X}$} q3;

        box {$S$} qnm3;
        box {$\sqrt{X}$} qnm3;
        
        box {$S$} qnm2;
        box {$\sqrt{X}$} qnm2;

        blankS qnm1;
        box {$\sqrt{X}$} qnm1;

        text {\dots} (q0,q1,q2,q3);
        text {\dots} (qnm3,qnm2,qnm1);

        zz (q1, q0);
        zz (q3, q2);
        zz (qnm3);
        zz (qnm1, qnm2);
        barrier (q0,q1,q2,q3);
        barrier (qnm3,qnm2,qnm1);

        box {$S$} q0;
        box {$\sqrt{X}$} q0;

        blankS q1;
        box {$\sqrt{X}$} q1;

        box {$S$} q2;
        box {$\sqrt{X}$} q2;

        box {$S$} q3;
        box {$\sqrt{X}$} q3;

        blankS qnm3;
        box {$\sqrt{X}$} qnm3;

        blankS qnm2;
        box {$\sqrt{X}$} qnm2;

        box {$S$} qnm1;
        box {$\sqrt{X}$} qnm1;
        barrier (q0,q1,q2,q3);
        barrier (qnm3,qnm2,qnm1);

        zz (q2, q1);
        zz (q3);        
        zz (qnm3,qnm2);
        zz (qnm1,q0);
        barrier (q0,q1,q2,q3);
        barrier (qnm3,qnm2,qnm1);

        box {$S$} q0;
        box {$\sqrt{X}$} q0;

        blankS q1;
        box {$\sqrt{X}$} q1;

        box {$S$} q2;
        box {$\sqrt{X}$} q2;

        box {$S$} q3;
        box {$\sqrt{X}$} q3;

        box {$S$} qnm3;
        box {$\sqrt{X}$} qnm3;

        box {$S$} qnm2;
        box {$\sqrt{X}$} qnm2;

        box {$S$} qnm1;
        box {$\sqrt{X}$} qnm1;
        barrier (q0,q1,q2,q3);
        barrier (qnm3,qnm2,qnm1);

        box {$S$} q0;
        h q1;
        h q2;
        box {$S$} q3;
        h qnm3;
        h qnm2;
        box {$S$} qnm1;
    \end{yquant*}

    \draw[decorate,decoration={brace,mirror,amplitude=4pt}]
      (1.1,-6) -- (2.1,-6)
      node[midway,below=6pt,align=center] {CZ layer};

    \draw[decorate,decoration={brace,mirror,amplitude=4pt}]
      (2.60,-6) -- (4,-6)
      node[midway,below=6pt,align=center] {Random 1Q\\layer};
    \draw[decorate,decoration={brace,mirror,amplitude=4pt}]
      (13.8,- 
      6) -- (15.3,-6)
      node[midway,below=6pt,align=center] {Rotate into\\graph state};
    \draw[decorate,decoration={brace,amplitude=6pt}]
      (1,0.3) -- (15.20,0.3)
      node[midway,above=8pt] {$n$ brickwork layers};
\end{tikzpicture}
\caption{{\bf Circuit preparing high-entanglement random graph states on $n$ qubits connected on a ring}. After an initial layer of Hadamard gates, alternating layers of brickwork CZ gates and single-qubit Clifford gates are applied until reaching a CZ depth of $n$. A final layer of phase or Hadamard gates rotates the stabilizer state into a graph state.}
\label{figure:base_circ}
\end{figure}

\subsection{Spacetime codes} \label{sec:check_and_doping}


\textit{Pauli checks} are syndrome detection gadgets wrapping a target payload circuit. They can be seen as artificially engineered stabilizers implemented via ancilla qubits and gates. As such, they can be used to detect errors during circuit execution. Unlike the fully error-corrected regime, the number of Pauli checks can be kept small---they are often used to reject shots where a non-zero syndrome is measured.
The performance of Pauli checks is directly tied to the fraction of errors detected and the number of extra gates added to the initial circuit: a good check will detect a large fraction of errors while introducing few additional gates.

Based on the formalism of spacetime codes, it is possible to find stabilizers distributed in both space and time of a circuit that satisfy the criteria for good checks~\cite{error_detection}.
More precisely, one can identify stabilizers located on the internal wires of a circuit, restricted to a single data qubit. This property ensures that the corresponding syndrome detection gadget can be efficiently implemented using a number of entangling gates exactly matching the weight of the stabilizer, even when the data qubit is only connected to a single ancilla. Figure~\ref{fig:pauli_check} depicts such a space-local stabilizer and the corresponding check implementation. We will now introduce the tools necessary to characterize and assess the quality of spacetime Pauli checks.

\begin{figure}[h!]
    \centering
    \hspace{-2.5cm}\scalebox{0.9}{
        \begin{tikzpicture}
    \newsavebox{\czgate}
    \savebox{\czgate}{
        \draw (0, 0) node[ibmcyan, fill=ibmcyan, circle, inner sep=2pt] {};
        \draw (0, -1) node[ibmcyan, fill=ibmcyan, circle, inner sep=2pt] {};

        \draw[very thick, ibmcyan] (0, 0) -- (0, -1);
    }
    \newsavebox{\hgate}
    \savebox{\hgate}{
        \draw[ibmcyan, fill=ibmcyan, rounded corners] (-0.3, 0.3) rectangle node[white]{$H$} (0.3, -0.3);
    }

    \newsavebox{\sqrtzgate}
    \savebox{\sqrtzgate}{
        \draw[ibmcyan, fill=ibmcyan, rounded corners] (-0.3, 0.3) rectangle node[white]{$S$} (0.3, -0.3);
    }

    \newsavebox{\wirewithdot}
    \savebox{\wirewithdot}{
        \draw (0.5, 0) node[ibmpurple, fill=ibmpurple, rectangle, inner sep=2pt] {};
        \draw[ibmpurple, very thick] (0, 0) -- (1, 0);
    }
    \newsavebox{\wirewithdotacz}
    \savebox{\wirewithdotacz}{
        \draw (0.4, 0) node[ibmpurple, fill=ibmpurple, rectangle, inner sep=2pt] {};
        \draw[ibmpurple, very thick] (0, 0) -- (1, 0);
    }
    \newsavebox{\wirewithdotbcz}
    \savebox{\wirewithdotbcz}{
        \draw (0.6, 0) node[ibmpurple, fill=ibmpurple, rectangle, inner sep=2pt] {};
        \draw[ibmpurple, very thick] (0, 0) -- (1, 0);
    }
    \newsavebox{\longwire}
    \savebox{\longwire}{
        \draw (1, 0) node[ibmpurple, fill=ibmpurple, rectangle, inner sep=2pt] {};
        \draw[ibmpurple, very thick] (0, 0) -- (2, 0);
    }
    \newsavebox{\inputgate}
    \savebox{\inputgate}{
        \draw (0, 0) node[ibmred, fill=ibmred, circle, inner sep=3.5pt] {\textcolor{white}{\tiny in}};
    }
    \newsavebox{\outputgate}
    \savebox{\outputgate}{
        \draw (0, 0) node[ibmred, fill=ibmred, circle, inner sep=2pt] {\textcolor{white}{\tiny out}};
    }
    
    \newcommand{\wirelabel}[3]{
        \draw(#1, #2) node[yshift=0.5cm]{\textcolor{ibmpurple}{$w_#3$}};
    }
    
    \newsavebox{\fullcircuita}
    \savebox{\fullcircuita}{
        \draw (-1, 0) node{\usebox{\wirewithdot}};
        \wirelabel{-0.5}{0}{1};
        \draw (0, 0) node{\usebox{\wirewithdotbcz}};
        \wirelabel{0.6}{0}{2};
        \draw (1, 0) node{\usebox{\longwire}};
        \wirelabel{2}{0}{3};
        \draw (3, 0) node{\usebox{\wirewithdotbcz}};
        \wirelabel{3.6}{0}{4};
        \draw (4, 0) node{\usebox{\wirewithdotacz}};
        \wirelabel{4.5}{0}{5};
        \draw (-1, -1) node{\usebox{\longwire}};
        \wirelabel{0}{-1.1}{6};
        \draw (1, -1) node{\usebox{\wirewithdotacz}};
        \wirelabel{1.5}{-1.1}{7};
        \draw (2, -1) node{\usebox{\longwire}};
        \wirelabel{3}{-1.1}{8};
        \draw (4, -1) node{\usebox{\wirewithdotacz}};
        \wirelabel{4.5}{-1.1}{9};
        \draw (-1, 0) node{\usebox{\inputgate}};
        \draw (-1, -1) node{\usebox{\inputgate}};
        \draw (0, 0) node{\usebox{\hgate}};
        \draw (1, 0) node{\usebox{\czgate}};
        \draw (2, -1) node{\usebox{\hgate}};
        \draw (4, 0) node{\usebox{\czgate}};
        \draw (3, 0) node{\usebox{\sqrtzgate}};
    
        \draw (5, 0) node{\usebox{\outputgate}};
        \draw (5, -1) node{\usebox{\outputgate}};
    }

    \draw (0,0) node{a)};
    \draw (2.7, 1-1.5)  node {$ \Qcircuit @C=1.0em @R=1.2em @!R { 
        &\gate{H} & \ctrl{1} & \qw       & \gate{S}  & \ctrl{1} &\qw\\
        &\qw      & \ctrl{0} & \gate{H}  &  \qw & \ctrl{0} &\qw\\
     }$
    };
    
    \draw (7-0.8,0) node{b)};
    \draw(8, 0) node{\usebox{\fullcircuita}};
    \draw (0,0-2) node{c)};
    \draw (2.7, 1-1.5-2)  node {$ \Qcircuit @C=1.0em @R=1.2em @!R { 
       &\gate{X} &\gate{H} & \ctrl{1} & \qw       & \gate{S}  & \ctrl{1} &\qw\\
        &\gate{X} &\qw      & \ctrl{0} & \gate{H}  &  \gate{Z} & \ctrl{0} &\qw\\
     }$
    };
    \draw (2.7, 1-1.5-2-1)  node {$X_{w_1}X_{w_6}Z_{w_8}$};
    \draw (7,0-2) node{d)};
    \draw (7+2.7, 1-1.5-2)  node {$ \Qcircuit @C=1.0em @R=1.2em @!R { 
       &\gate{H} & \ctrl{1} & \qw       & \gate{S}  & \ctrl{1} &\qw&\qw\\
        &\gate{Y}      & \ctrl{0} & \gate{H}  &  \qw & \ctrl{0} &\gate{Y}&\qw\\
     }$
    };
    \draw (7+2.7, 1-1.5-2-1)  node {$Y_{w_6}Y_{w_9}$};
    \draw (2,0-2-2.5) node{e)};
    \draw (2.7+3.5, 1-1.5-2-2.5)  node {$ \Qcircuit @C=1.0em @R=1.2em @!R { 
       &\qw&\gate{H} & \ctrl{1} & \qw       & \gate{S}  & \ctrl{1} &\qw&\qw&\qw&\qw\\
        &\qw&\gate{Y}      & \ctrl{0} & \gate{H}  &  \qw & \ctrl{0} &\gate{Y}&\qw&\qw&\qw\\
        \lstick{\ket{0}}&\gate{H}& \ctrl{-1}     &\qw & \qw  &  \qw & \qw &\ctrl{-1}&\qw&\gate{H}&\qw&\rstick{\ket{0}}
     }$
    };

\end{tikzpicture}
    }
    \caption{{\bf Spacetime check construction.} {\bf a.} A $2$-qubit Clifford quantum circuit. {\bf b.} Its corresponding set of wires. {\bf c.} An example of a spacetime stabilizer. Notice how this stabilizer is spread in both space and time: it involves both qubits (space spread) and is nontrivial at different circuit depths, possibly requiring swap operations to implement. {\bf d.} Another spacetime stabilizer, this time localized in space: it only involves the second qubit at different circuit depths. {\bf e.} Implementation of a Pauli check using the spacetime stabilizer of (d).}
    \label{fig:pauli_check}
\end{figure}
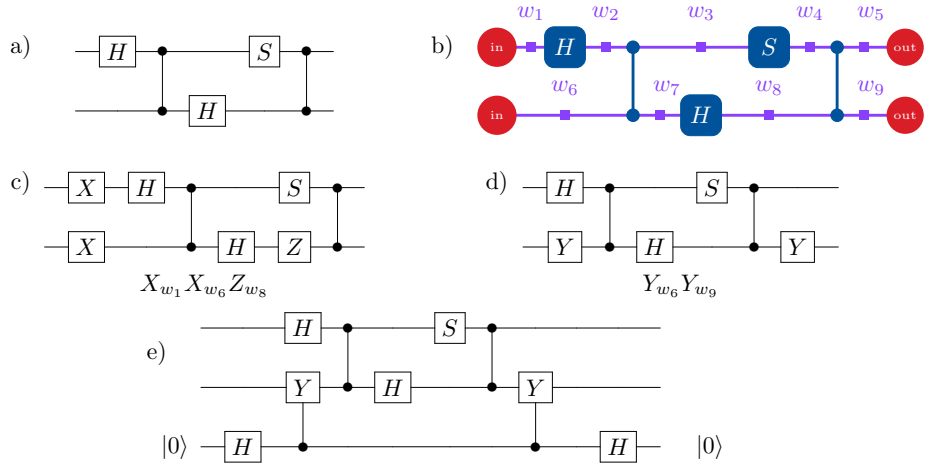

\paragraph{Circuit wires as spacetime coordinates.} Consider a Clifford circuit $C$. This circuit can be understood as a directed acyclic graph $\mathcal{G}(C)$, or \emph{DAG}, where vertices are gates and edges are input/output dependencies between gates. Since we are interested in adding extra Pauli gates between the existing gates of the circuit, we can consider the \emph{line graph} $\mathcal{G}(C)^\star$ whose vertices represent the input/output relations between the gates. Two connections will be related in this graph if and only if they correspond to an incoming/outgoing pair of connections connected to the same gate. Figure \ref{fig:dag} depicts this construction.
We denote by $\mathcal{W}(C)$ the set of vertices of $\mathcal{G}(C)^\star$. 
Notice that $\mathcal{G}(C)^\star$ is itself acyclic and thus induces a partial order $\leq_C$ on $\mathcal{W}(C)$. This partial order exactly matches the notion of time-dependence between wires in a circuit. For instance, given a single wire $w\in \mathcal{W}(C)$, one can isolate the ``first-half'' of the circuit surrounding all the wires $w' \leq_C w$. This effectively gives access to the reverse light-cone of $w$, denoted $A_w$. Similarly, one can also define the forward light-cone of $w$, $B_w$ as the subcircuit composed of all the gates surrounding wires $w'$ such that $w\leq_C w'$.

\begin{figure}[t]
    \centering
    \hspace{-4.5cm}\scalebox{0.6}{
        \input{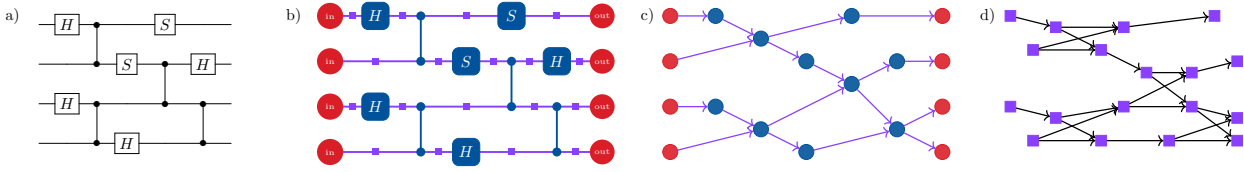}
    }
    \caption{\textbf{Wires as spacetime coordinates of a Clifford circuit.} {\bf a.} A quantum circuit $C$. {\bf b.} This same circuit with $8$ additional dummy gates representing inputs and outputs (in red). Wires are highlighted in purple. {\bf c.} The underlying directed acyclic graph (DAG) $\mathcal{G}(C)$. Its set of vertices is composed of all the gates, including the dummy gates. Its edges are exactly the input/output relations between the circuit's gates. {\bf d.} The line graph $\mathcal{G}(C)^\star$ of the underlying DAG. This graph is another DAG. Its vertices are the edges of $\mathcal{G}(C)$. Two vertices are connected if and only if the corresponding edges form a directed path of length 2 in $\mathcal{G}(C)$.}
    \label{fig:dag}
\end{figure}

Each wire in a given circuit carries the information of a single qubit's state. As such, one can forget that those qubits are logically related through the circuit's gates and consider a global spacetime Hilbert space $\mathcal{H}(C)$ composed of qubit spaces attached to each wire in the circuit:
\begin{align}
    \mathcal{H}(C) = \bigotimes_{w\in\mathcal{W}(C)} \mathbb{C}^2
\end{align}

We can now define a Pauli operator over $\mathcal{H}(C)$. This operator will have one Pauli component per wire in the circuit.  Let $P$ be some single-qubit Pauli operator and $w\in \mathcal{W}(C)$. We denote by $P_w$ the Pauli operator $P\otimes \left(\bigotimes_{x\in\mathcal{W}, x\neq w} I\right)$, or, in other words, the Pauli operator acting as $P$ on qubit $w$ and as the identity on every other qubit. To simplify notation in some settings, we might interchangeably use $(P, w)$ or $P_w$ to refer to this same operator.
We are now interested in characterizing the subgroup of Pauli operators such that adding one of those operators to the circuit does not change the implemented Clifford operator.

\paragraph{Valid check condition.} Given a single wire $w\in \mathcal{W}(C)$ and a Pauli $P_w$ located on that wire, one can back-propagate $P_w$ to the input wires of the circuit by conjugating it via the reverse light-cone $A_w$. We denote this operator by $B(P_w)$. This operator is supported only on the input wires of the circuit and is such that $C\cdot B(P_w)$ is equivalent to $C$ with an extra single Pauli operator $P$ on wire $w$.

Given a collection ${P_1}_{w_1}, ..., {P_k}_{w_k}$, we want to express that adding those extra Paulis in $C$ doesn't modify $C$. As per the statement above, this is equivalent to: $$C\cdot \prod_{1\leq i\leq k} B({P_i}_{w_i}) = C \Leftrightarrow \prod_{1\leq i\leq k} B({P_i}_{w_i}) =I$$

In other words, one can ``sprinkle'' Pauli operators throughout the wires of a Clifford circuit $C$ as long as their back-propagators amount to the identity. Those arrangements of Paulis are exactly the \emph{spacetime stabilizers} of $C$. 
To ease notation, we will denote by $B(E) = \prod_{P_w \in E} B(P_w)$ the back-propagator of a collection of Pauli operators spread throughout $C$. Notice that one could have obtained the same condition by propagating each Pauli through its forward light-cone $B_w$ instead. We denote the resulting forward-propagator $F(E)= \prod_{P_w \in E} F(P_w)$.

Notice that the spacetime stabilizers of $C$ form a subgroup of the Pauli group over $\mathcal{H}(C)$. As such, one can use the Symplectic encoding of Pauli operators \cite{calderbank1998quantum,gottesman1997stabilizer} to map spacetime Pauli operators to vectors $v\in \mathbb{F}_2^{2|\mathcal{W}(C)|}$. The condition of being a spacetime stabilizer then naturally boils down to satisfying a set of linear equations of the form:
\begin{align}\label{eq:kernel}
    B\cdot v &= 0
\end{align}
where the columns of $B$ are exactly the symplectic encodings of the back propagators of all possible single Pauli $X/Y/Z$ on all possible wires.

Equation \ref{eq:kernel} characterizes the spacetime stabilizers of the circuit, without any assumption with respect to the input state of the circuit. In our setting, we are interested in running the target circuit on the trivial stabilizer state $\ket{0^n}$. Under this additional relaxation we change Equation  \ref{eq:kernel} into:
\begin{align}\label{eq:kernel_stab}
    N\cdot B\cdot v &= 0
\end{align}
where $N$ is the nullspace of the stabilizer group of the input state. When this state is $\ket{0^n}$, this boils down to forgetting about the last $n$ rows of $B$ (i.e. the $Z$ components of the back-propagators).

\paragraph{Implementing checks.} Once a valid spacetime stabilizer $({P_i}_{w_i})_{1\leq i \leq k}$ is picked, it can be turned into a Pauli check by allocating an extra ancilla $a$ in state $\ket{+}$ and adding $k$ extra controlled Pauli gates of the form $\Lambda_a P_i$ from $a$ to the appropriate wire (see Figure \ref{fig:pauli_check}). 
This construction assumes that the ancilla $a$ can be coupled with the qubits corresponding to those wires. Notice that the number of extra (entangling) gates added to the circuit corresponds exactly to the weight of the check.

\paragraph{Spatially constrained checks.} In practice, we want to find spacetime Pauli checks
that can be natively implemented given a target architecture. In order to negate any unexpected overhead due to transpilation, we proceed as follows.
We start by mapping the target Clifford circuit to the architecture layout. In our experiment, we generated the target circuit using a brickwork construction which can be easily mapped to any connected path of qubits. Then, to each qubit $a$ sitting next to the chosen path of qubits, we associate its neighboring data qubit $v_a$. Once $v_a$ is isolated, we can use the framework described above to characterize the spacetime stabilizer subgroup supported only on the circuit wires corresponding to qubit $v_a$. This ensures that, independently from the chosen stabilizer, the corresponding check can be implemented trivially.

\begin{figure}[h]
    \centering
    \scalebox{0.75}{
        \begin{tikzpicture}
    \foreach \n in {0,...,11}{
        \draw (\n, 0) node[circle, draw, ibmblue, fill=ibmblue](n\n){};
    }
    \draw[ibmblue, very thick](0,0) -- (11, 0);
    \draw (5, -1) node[circle, draw, ibmpurple, fill=ibmpurple](a1){};
    \draw[ibmpurple, very thick](5, -1) --(5, 0);
    \draw (5, 0) node[circle, draw, ibmgreen, fill=ibmgreen](a1){};

    \draw (7, -1) node[circle, draw, ibmpurple, fill=ibmpurple](a1){};
    \draw[ibmpurple, very thick](7, -1) --(7, 0);
    \draw (7, 0) node[circle, draw, ibmgreen, fill=ibmgreen](a1){};

    \draw (9, -1) node[circle, draw, ibmpurple, fill=ibmpurple](a1){};
    \draw[ibmpurple, very thick](9, -1) --(9, 0);
    \draw (9, 0) node[circle, draw, ibmgreen, fill=ibmgreen](a1){};
    \draw[ibmblue, very thick, dashed] (0, 0) .. controls (1, 1) and (8, 2) .. (11, 0);
\end{tikzpicture}
    }
    \caption{{\bf Check picking layout.} The data qubits (in blue and green) are arranged in a loop. Any adjacent ancilla qubit (purple) can be used to implement a spatially local check supported on a single data qubit (green). This method is easily generalizable to any connectivity of data and ancilla qubits.}
    \label{fig:check_layout}
\end{figure}
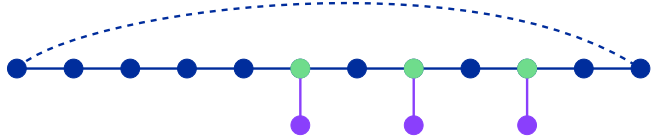

\paragraph{Finding low-weight checks.} Since the implementation cost of checks is directly tied to the weight of the underlying spacetime stabilizer, we are interested in finding low-weight spacetime stabilizers. This problem, however, is notably hard in the general case~\cite{vardy1997intractability}. In practice, we rely on greedy heuristics. We start from a random low-weight Pauli and greedily add extra Paulis by trying to reduce the weight of the corresponding back-propagator. When successful, this approach provides us with a valid stabilizer that is then used to build a Pauli check (see below).

\paragraph{Scoring checks.} Once a set of low-weight checks is computed, we need a tool to differentiate their performance. This performance assessment relies on computing another Pauli operator, called the \textit{back-cumulant}, acting on $\mathcal{H}(C_i)$ where $C_i$ is the Clifford circuit obtained by implementing the $i$th check candidate. The \textit{back-cumulant} of a check qubit $q$, denoted $\overleftarrow{B}(C, q)$,  is obtained by back-propagating the final $Z_q$ measurement of the check qubit while tracking the support of that Pauli on each intermediate wire. Figure \ref{fig:full_example_cumulant} depicts such a back-cumulant. Once this operator is computed, and given a Pauli noise model, one can sample many faults and see if they anticommute with the check's back-cumulant, giving us a direct detection criterion~\cite{delfosse2023simulation}. 

\begin{figure}[h]
    \centering
    \includegraphics[scale=0.85]{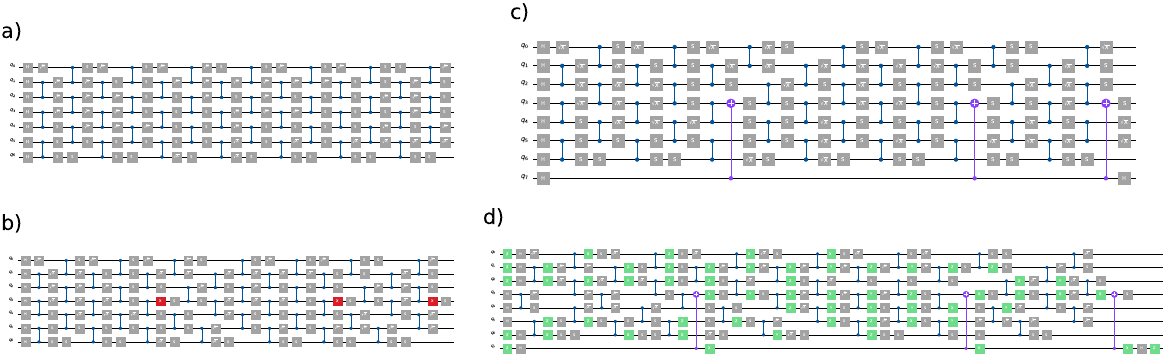}
    \caption{\textbf{Large circuit coverage with a low-overhead check.} {\bf a.} An example circuit with entangling gates highlighted. {\bf b.} Circuit with a spacetime Pauli operator inside it highlighted in red, forming a valid check of the circuit. {\bf c.} Check circuitry added: an ancilla in \( \ket{+} \) state is entangled to the circuit according to the valid check, and is measured in the \( X \) basis. {\bf d.} Check's back-cumulant highlighted in green. The check will detect any fault that anticommutes with its back-cumulant.}
    \label{fig:full_example_cumulant}
\end{figure}

In practice, we model hardware noise using a Lindblad local Pauli noise where each $2-$qubit gate is followed by $15$ individual channels of shape $\mathcal{E}(\rho) = w_k \rho + (1-w_k) P_k \rho P_k^\dagger$, for  each non-trivial Pauli $P_k$ and where $w_k = \frac{1 + e^{-2\lambda_k}}{2}$. Given such a noise model, we use the quantity $\Gamma = e^{2\sum_k\lambda_k}$ as proxy to the infidelity of the resulting circuit.

When scoring a new check, we first prune the model by setting $\lambda_k=0$ for any $P_k$ that anticommutes with at least one of the check back-propagators (including the previously picked checks). Once the noise model is pruned, we compute the resulting $\Gamma$ and use this as the cost for the new check. This amounts to considering a first order approximation of the residual noise.

\paragraph{Check picking heuristic and implementation.} We use the following heuristic to enumerate a set of low-weight spatially local stabilizers. We first isolate the set of wires $\mathcal{W}_q$ corresponding to the target data qubit $q$. We then pick a fraction $f$ of those wires, obtaining $\mathcal{W}'\subset \mathcal{W}_q$, and pick two random wires $w_1, w_2\in \mathcal{W}'$. We then enumerate the $9$ possible placements of Paulis on those two wires and run a greedy decoding algorithm to modify the Pauli operator into a valid spacetime stabilizer. In practice, we use $f=0.5$ and we iterate the procedure 1000 times, generating up to $9000$ stabilizers. A Rust implementation of this check picking algorithm, together with the noise modeling tools used in this work, is freely available~\cite{qiskit-paulice}.

\subsection{Increasing magic by $T$-doping within the code}
\label{sec:t-doping}

It is possible to extend the spacetime stabilizer formalism to describe stabilizers in the presence of non-Clifford rotations, such as $T$ gates, in the main Clifford skeleton. In such cases, Equations \ref{eq:kernel} and \ref{eq:kernel_stab} become even more constrained, and stabilizer encodings will satisfy:

\begin{align}
     \begin{bmatrix}N\cdot B\\F\end{bmatrix} \cdot x &= 0 
\end{align}

where $F$ is a matrix obtained by computing the commutation relations between single Pauli propagators and $T$ gates. That is $F_{i,j} =1$ if and only if the $i$th $T$ gate commutes with the Propagator of the $j$th single Pauli considered when building $B$.
As such, the dimension of the stabilizer group will decrease linearly with the number of such new constraints (i.e. the number of $T$ gates).

To avoid this caveat, we use a $T$-doping framework: we first choose a random circuit instance $C$, pick a collection of spacetime Pauli checks according to the available ancilla, producing another Clifford circuit $C'$. We then compute the set $\mathcal{T}$ of internal wires $w$ such that $Z_w$ commutes with all the checks' back-cumulants. By construction, any $T$ gate added on $w\in \mathcal{T}$ will commute with the added checks and preserve the code.

In practice, we compute all the valid $T$ locations that directly follow a $CZ$ gate and prune this set of locations using three criteria. We first remove any rotation that commutes with all previous rotations and propagates into a diagonal rotation on the circuit input. Similarly, we remove any rotation that commutes with all following rotations and propagates into a diagonal rotation in the circuit output. Indeed, the first set of rotations would act trivially since they belong to the spacetime stabilizer group of the doped circuit while the latter set would commute with final $Z-$measurements. Lastly, if two rotations are equivalent (i.e. the first forward propagates to the second) and commute with all rotations lying in between them, we remove one of the two at random.
This process is iterated until no such pattern can be found.

\section{Bounding state fidelity after $T$-doping} \label{sec:sfe}

In this section we prove the claimed lower bound on the fidelity of a non-stabilizer state prepared by a doped Clifford circuit, given (efficient) access to the fidelity of the Clifford state, and the syndromes of its spacetime code.

\begin{theorem}\label{thm:fidelity}

Let \(C_1\) be an \(n\)-qubit random Clifford circuit of depth \(n\), preparing
\(\ket{\psi_1}=C_1\ket{0^n}\), equipped with spacetime Pauli checks implemented on qubits $(q_i)_{1\leq i \leq m}$.  Let \(C_2\) be obtained by doping \(C_1\) with
check-commuting \(T\) gates at spacetime wires \(\{w_a\}_{a=1}^t\), satisfying
\[
[Z_{w_a},\overleftarrow{B}(C_1, q_i)]=0
\qquad
\forall a=1,\dots,t,\quad \forall i=1,\dots,m .
\]

Define accepted and harmless spacetime faults in the Clifford circuit 

\[
\mathcal{A}
:=\left\{
E  \,\middle|\,
[E,\overleftarrow{B}(C_1,q_i)]=0
\quad \forall i\in[m]
\right\},
\qquad
\mathcal{H}
:=\left\{
E\in\mathcal{A}\setminus\{I\}
\,\middle|\,
F(E)\in\operatorname{Stab}(\lvert\psi_1\rangle)
\right\}.
\]

Then the post-selected fidelity of the doped state
\(\ket{\psi_2}=C_2\ket{0^n}\) 
is lower bounded by
\begin{align}
F_2
&\ge F_1-\Pr(E\in\mathcal{H}\mid E\in\mathcal{A}),
\label{eq:fidelity-bound}
\end{align}

\noindent where, under a circuit-independent Pauli noise model

\begin{equation}
\Pr(E\in\mathcal H\mid E\in\mathcal A)
=
O(n^{-c})
\label{eq:harmless-probability}
\end{equation}

\noindent for some constant $c > 0$.

\end{theorem}

\begin{proof}

The proof follows 3 steps:

\begin{enumerate}
    \item We show that the equality of post-selection rates implies the noiseless implementation of T gates.
    
    \item Assuming identical noise affecting $C_1$ and $C_2$, we show that the fidelity drop due to doping is capped by the probability of harmless faults in $C_1$, i.e. those accepted non-trivial faults that stabilize $\ket{\psi_1}$.

    \item We show that the probability of harmless-non-identity errors in random linear-depth Clifford circuits is inverse polynomial in the number of qubits.
\end{enumerate}

\subsection{Bounding the noise on $T$ gates} As a first step, assume that each $T$ gate is affected by a small amount of noise taking the form of a single-qubit uniform depolarizing channel $\mathcal{E}$ right after the $T$ gate:
$$ \mathcal{E}(\rho) = (1 - 3\varepsilon)\rho + \sum_{P\in\{X, Y, Z\} }\varepsilon P \rho P . $$
Let us further assume that among the $t$ $T$ gates used for doping, $k$ of them are covered by some Pauli check. This entails that, for those gates, $\frac{2}{3}$ of the errors of the corresponding channel will be detected by at least one check.
Denote the acceptance probability of outcomes in the doped circuit by $p_2(\varepsilon)$. At first order, we have that:
\begin{align}\label{eq:low_tnoise_influence}
    p_2(\varepsilon) &= p_2(0) - 2\varepsilon k + O(\varepsilon^2)
\end{align}

We can show that the acceptance probability (i.e. post-selection rate) of the undoped experiment, $p_1$, and of the doped-with-perfect-T-gates experiment, $p_2(0)$ are identical.
This comes directly from the fact that $T$ gates commute with the Pauli checks by construction. Indeed, for any doped Clifford circuit containing Pauli noise, we can propagate faults to the end of the circuit by:
\begin{itemize}
    \item commuting the fault through Clifford gates by conjugating it,
    \item commuting the fault through $T$ gates by spawning a $S^\dagger$ gate after the $T$ gate
\end{itemize}
Overall, the end faulty circuit will have shape:

$$ P \prod R_Q(\pi/2) C_2 $$

\noindent where $C_2$ is the ideal doped circuit, each $Q$ is a multi-qubit operator obtained by propagating a $Z$ from each $T$ gate that anticommutes with the fault and $P$ is the final propagated fault. Since the $T$ gates commute with the checks, we have that the check measurements will trivially commute with the $R_Q(\pi/2)$ rotations. Hence, the post-selection rate of the doped circuit with perfect $T$ gates matches exactly that of the noisy Clifford circuit.

Using this result together with Equation \ref{eq:low_tnoise_influence} gives us:
\begin{align}
    p_2(\varepsilon) &= p_1 - 2\varepsilon k + O(\varepsilon^2)
\end{align}

Since we do not observe any post-selection rate drop between the doped and undoped experiments and since $k$ is large, we conclude that $\varepsilon$ is close to $0$ and that $T$ gates can be considered noiseless.

\subsection{Lower bounding the doped state's fidelity}

In addition to the accepted ($\mathcal{A}$) and harmless ($\mathcal{H}$) fault paths defined above, let us define the harmful fault paths:

\[
\mathcal{L}:=\{E\in\mathcal{A}:E\neq I,\;
F(E)\notin\operatorname{Stab}(\ket{\psi_1})\}.
\]

Note that harmless and harmful sets are defined with respect to the Clifford circuit. While they can change after doping, the accepted fault set remains invariant. This is due to the fact that T gates commute with checks, therefore preserving their back-cumulants:

\[
[E,\overleftarrow{B}(C_2, q_i)]=[E,\overleftarrow{B}(C_1, q_i)]=0 .
\]

For an accepted Pauli fault path \(E\), define its impact on the output fidelity as
\begin{equation*}
\begin{aligned}
f_1(E)&=|\langle\psi_1|F(E)|\psi_1\rangle|^2,
\qquad
f_2(E)&=|\langle\psi_2|\widetilde{F}(E)|\psi_2\rangle|^2,
\end{aligned}
\end{equation*}

\noindent where the second equation contains forward-propagation $\widetilde{F}$ through a non-Clifford circuit. Let $p(E)$ denote the probability of fault path $E$. Then, the post-selected fidelities satisfy,
\[F_2 - F_1 = \frac{\sum_{E \in \mathcal{A}}{p(E) (f_2(E) - f_1(E))}}{p_{\mathrm{acc}}}
\]

Given that the T gates are noiseless, the spacetime error distribution $p(E)$ is the same in both circuits. Furthermore, trivial (identity) errors contribute equally to both fidelities. Therefore we can write the fidelity difference between the doped and Clifford circuits as

\begin{equation*}
F_2 - F_1 = \frac{\sum_{E \in \mathcal{H}}{p(E) (f_2(E) - f_1(E))} + \sum_{E \in \mathcal{L}}{p(E) (f_2(E) - f_1(E))} }{p_{\mathrm{acc}}}
\end{equation*}

Since $\ket{\psi_1}$ is a stabilizer state, 
\[
f_1(E)=
\begin{cases}
1 & E\in \mathcal{H}\\
0 & E\in \mathcal{L}
\end{cases}
.\]

Therefore,

\begin{equation*}
 F_2 - F_1   = - \frac{\sum_{E \in \mathcal{H}}{p(E) (1 - f_2(E))} }{p_{\mathrm{acc}}} + \frac{\sum_{E \in \mathcal{L}}{p(E) f_2(E)} }{p_{\mathrm{acc}}} .
\end{equation*}

Note that $0 \leq f_2(E) \leq 1$. Therefore the first term lower bounds the fidelity difference and the second term upper bounds it. We are only interested in the lower bound --- any fidelity improvement due to doping only helps our verification process. We therefore get the following lower bound on the fidelity of the doped state with respect to the Clifford state:

\begin{equation}
F_2 \geq F_1 - \frac{\Pr(E \in \mathcal{H})}{\Pr(E \in \mathcal{A})} = F_1 - \Pr(E \in \mathcal{H} \mid E \in \mathcal{A})
\end{equation}

\subsection{Bounding the probability of harmless errors}

It remains to prove that $\Pr(E \in \mathcal{H} \mid E \in \mathcal{A})$ is small under reasonable assumptions.
Our two assumptions are the following (informal):
\begin{enumerate}
    \item \textbf{Circuit independence.} The noise model only depends on the size of the circuit and not on the gates within it. In our case, this is a very loose assumption: it entails that for fixed $n$, two different circuits drawn from our random circuit model will be subject to the same noise models.
    \item \textbf{Space-time Pauli noise.} We further assume that the noise acting on the wires of the target circuit is Pauli.
\end{enumerate}

More formally, fixing $n$, consider a circuit $C$ drawn from our family of Clifford circuits. Let $\mathcal{W}(C)$ be the entire set of wires and subset $\tilde{\mathcal{W}}(C)\subseteq \mathcal{W}(C)$ be the set of wires that immediately follow $CZ$ or $\sqrt{X}$ gates. Denote the corresponding Hilbert space as $\tilde{\mathcal{H}}(C)$. Notice that for any circuits in our family, the set of wires $\tilde{\mathcal{W}}(C)$ always has the same size. In particular, given two circuits $C_1$ and $C_2$ in our family, we can always build a bijection $\sigma:  \tilde{\mathcal{W}}(C_1) \rightarrow \tilde{\mathcal{W}}(C_2)$ that is monotonic w.r.t $\leq_{C_1}$ and $\leq_{C_2}$, that is, $\forall w, w'\in\tilde{\mathcal{W}}(C_1)$, $w\leq_{C_1} w' \implies \sigma(w)\leq_{C_2} \sigma(w')$.

Thus, there is a natural morphism between Pauli operators acting on $\tilde{\mathcal{H}}(C_1)$ and $\tilde{\mathcal{H}}(C_2)$.
This entails that one can consider just a single reference circuit $C_{\textrm{ref}}$, and any Pauli operators on $\tilde{\mathcal{H}}(C_{\textrm{ref}})$  can be mapped to a Pauli operator acting on wires within any instance $C_i$ of the same family. For simplicity, we take the reference circuit to be the Clifford circuit obtained when no $S$ gates are inserted.

Define a Pauli noise model supported on $\tilde{\mathcal{H}}(C_{\textrm{ref}})$ as stochastic circuit independent if it is invariant under our circuit switching morphisms, i.e. if the noise does not depend on the presence of $S$ gates in the circuit.
Notice that this class of noise models can contain heavily correlated errors, both spatially and temporally.

Consider Pauli error $E \in \tilde{\mathcal{H}}(C_{\textrm{ref}})$ induced by such a noise model. $E$ is deemed harmless if and only if it back-propagates to some diagonal Pauli in $C_{\textrm{ref}}$. We claim that any harmless error occurring in $C_{ref}$ is unlikely to stay harmless for a random circuit from our family. This is because any random circuit in our family can be constructed by inserting $S$ gates with probability $\frac{1}{2}$ after each $CZ$ gate in $C_{\textrm{ref}}$.  So, on average, a back-propagator will cross paths with $O(n^2)$ such $S$ gates and pick up an additional $Z$ term at that spacetime location - which, when fully back-propagated, will turn into a random global Pauli.
As claimed, then, a typical harmless error for $C_{\textrm{ref}}$ is unlikely to be harmless for any other circuit. 

This argument generalizes to any pair of circuits in our family, suppressing the possibility that an adversarial model contains arbitrary heavy-weight spacetime stabilizers. Overall, the only contribution will come from errors that do not cross many $S$ locations in their back-propagator. These errors, by definition, are low weight and positioned at the boundary of the circuit.


Now we provide a counting argument to bound the probability of harmless faults. For each elementary (i.e., weight-1) spacetime Pauli fault \(P_j\), write
\[
B(P_j)=X^{x_j}Z^{z_j},
\qquad
x_j,z_j\in\mathbb{F}_2^n.
\]

For a fault path \(E\), define its total pulled-back \(X\)-component as
\[
x(E):=\sum_{j\in\operatorname{supp}(E)}x_j.
\]
Hence, we can redefine the set of harmless faults as
\[
\mathcal{H}
=
\left\{
E\in\mathcal{A}\setminus\{I\}
\,\middle|\,
x(E)=0
\right\}.
\]

Let \(W\) denote the fault weight, and define
\[
a_w(n)
:=
\Pr\!\left(
x(E)=0
\,\middle|\,
E\in\mathcal{A},\ W=w
\right).
\]
The harmless-fault probability can then be decomposed as
\[
\Pr(E\in\mathcal{H}\mid E\in\mathcal{A})
=
\sum_{w\geq 1}
\Pr(W=w\mid E\in\mathcal{A})\,a_w(n).
\]

Let \(N=\Theta(n^2)\) denote the number of noisy spacetime locations in the
\(n\)-qubit, depth-\(n\) brickwork circuit. For \(w=1\), harmlessness requires
\[
x_j=0.
\]
The systematic solutions are boundary faults near the input, of which there are
\(O(n)\), compared with \(\Theta(N)=\Theta(n^2)\) possible single faults.
Under the assumptions of the theorem, this gives
\[
a_1(n)=O(n^{-1}).
\]

For any fixed \(w\geq 2\), harmlessness requires
\[
x_{j_1}+x_{j_2}+\cdots+x_{j_w}=0.
\]
After fixing \(w-1\) faults, the final fault must realize a prescribed
pulled-back \(X\)-pattern. For a random Clifford circuit, we assume that such
collisions occur with bounded multiplicity~\cite{webb2015clifford}. Therefore,
under the assumptions of the theorem, the probability of a harmless
weight-\(w\) path is suppressed by a factor \(O(1/N)\), yielding
\[
a_w(n)
=
O(N^{-1})
=
O(n^{-2}),
\qquad
w\geq 2,
\]
for fixed \(w\).

At high fault weights, the accumulated pulled-back \(X\)-component is expected
to mix approximately uniformly over \(\mathbb{F}_2^n\)
~\cite{brown2012scrambling}. Since harmlessness requires
this vector to equal the unique zero vector,
\[
a_w(n)=O(2^{-n})
\]
in the high-weight mixing regime.

Therefore, when the accepted fault distribution is dominated by low weights,
the harmless-fault probability is inverse polynomial in \(n\), whereas when it
is dominated by the high-weight mixing regime, it is exponentially small.

\end{proof}

Theorem~\ref{thm:fidelity} states that the fidelity loss due to doping is upper-bounded by an asymptotically small quantity. This quantity depends on both the number of qubits $n$ and the underlying physical noise distribution. For a local depolarizing noise model with strength $p$, the fault weights follow a binomial distribution $W \sim Binomial(N, p)$, where the scaling can fall in different regimes depending on whether the dominant fault weights $Np$ are small or large. We numerically confirm the scaling predicted by Theorem~\ref{thm:fidelity} in Figure~\ref{fig:harmless}. 

\begin{figure*}[h]
\centering
\includegraphics[width=\textwidth]{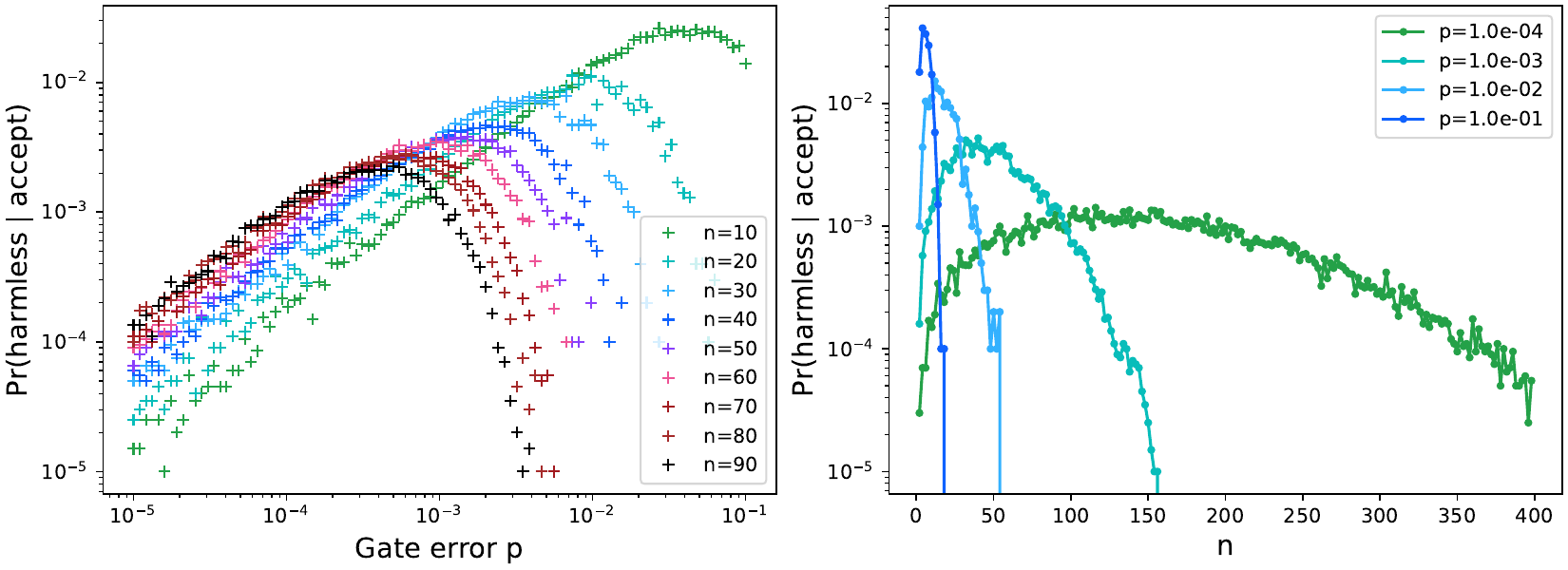} 
\caption{
{\bf  Bounding the probability of harmless faults in $n \times n$ Clifford circuits.} We sweep gate error rates and number of qubits. The probability of harmless and accepted faults is very small. Importantly, it gets exponentially small with $n$ and $p$ beyond a critical point.} \label{fig:harmless}
\end{figure*}

While Theorem~\ref{thm:fidelity} asymptotically bounds the drop in fidelity, we can also estimate this bound for the particular size of experiment that we run.  Since the bound only depends on the Clifford circuit (its fidelity and probabilities of harmless, accepted faults), we can do so efficiently. We use Monte Carlo simulation of faults drawn according to an underlying Pauli noise model (sweeping different noise strengths and polarizations). For the particular $64 \times 73$ circuit used in experiments, we estimate a maximum possible fidelity loss of around $1\%$.

\subsection{Generalization to coherent noise}
\label{subsec:coherent-noise}

We now generalize the fidelity relation to coherent
unitary noise.  The essential difference from stochastic Pauli noise
is that distinct Pauli fault paths occur with complex amplitudes and
may therefore interfere.

As Figure~\ref{fig:coherent-sim} numerically shows, the post-doping fidelity loss due to coherent errors stays small, which we attribute to the random circuit itself twirling the noise. Further Pauli twirling of gates completely removes the effect of coherent errors and recovers the stochastic Pauli noise model. Next, we formalize these observations.

\begin{figure*}[h]
\centering
\includegraphics[width=\textwidth]{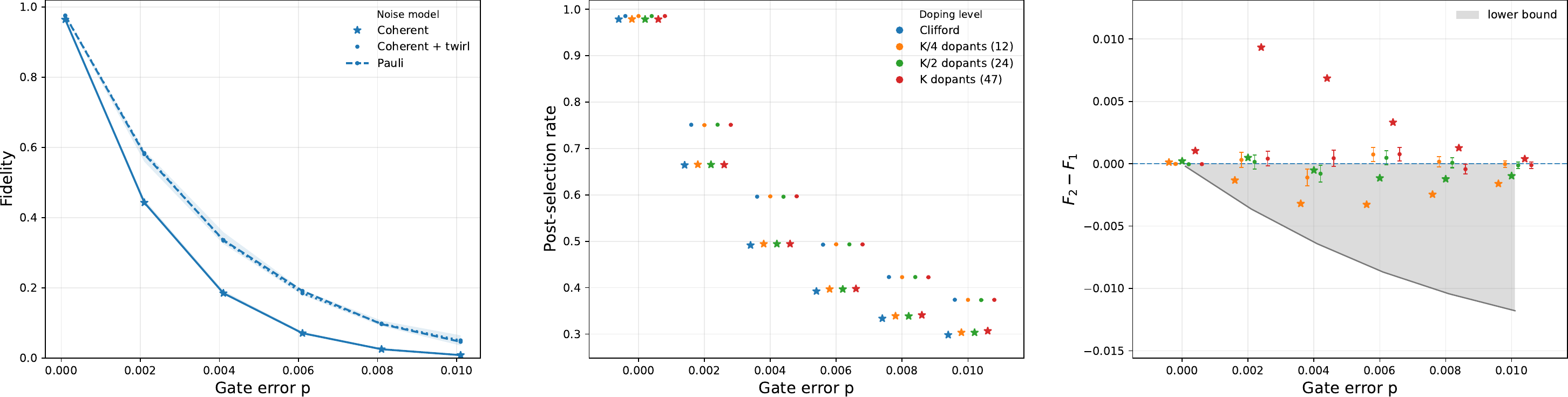} 
\caption{
{\bf Numerical estimation of fidelity change due to doping under Pauli and coherent errors.} The circuit is a 16-qubit Clifford with CZ-depth 16 and encoded with 2 spacetime checks. Left: post-selected fidelity of the Clifford, comparing stochastic Pauli noise, coherent noise, and locally twirled coherent noise, confirming that twirling converts coherent to Pauli noise. Middle: post-selection acceptance rate stays constant with doping, regardless of noise model. Right: Fidelity change from doping. The shaded gray region shows the estimated lower bound $-\Pr(harmless \mid accept)$, confirming that both noise models are above this bound.}\label{fig:coherent-sim}
\end{figure*}

Expanding the coherent errors at all noisy locations in the Pauli basis gives the spacetime fault-path expansion
\begin{equation}
    U_{\mathrm{err}}
       =
       \sum_E c_E E,
    \label{eq:coherent-fault-expansion}
\end{equation}
\noindent where $c_E\in\mathbb C$. We assume the same errors at locations common to $C_1$ and $C_2$, so the coefficients $c_E$ agree. Because the inserted $T$ gates are ideal and commute with the measured checks, each fault path produces the same syndrome in both circuits, and the accepted set $A$ is common.

For $j\in\{1,2\}$, let $V_j(E)$ denote the effective output operator obtained by propagating $E$ through $C_j$. The unnormalized post-selected state is
\begin{equation}
    \ket{\widetilde{\phi}_j}
       :=
       \sum_{E\in A}
          c_E V_j(E)\ket{\psi_j}.
    \label{eq:coherent-postselected-state}
\end{equation}
Its post-selection probability is
\begin{equation}
    p_j
       :=
       \braket{\widetilde{\phi}_j|
               \widetilde{\phi}_j},
    \label{eq:coherent-acceptance-probability}
\end{equation}
and, for $p_j>0$, its fidelity with the ideal target is
\begin{equation}
    \rho_j^{\mathrm{postselected}}
       :=
       \frac{
         \ket{\widetilde{\phi}_j}
         \bra{\widetilde{\phi}_j}
       }{p_j},
    \qquad
    F_j
       :=
       \bra{\psi_j}\rho_j^{\mathrm{postselected}}\ket{\psi_j}
       =
       \frac{
         \left|
           \braket{\psi_j|\widetilde{\phi}_j}
         \right|^2
       }{p_j}.
    \label{eq:coherent-postselected-fidelity}
\end{equation}

Define the single-path target overlap
\begin{equation}
    \alpha_j(E)
       :=
       \bra{\psi_j}V_j(E)\ket{\psi_j}.
    \label{eq:coherent-path-overlap}
\end{equation}
The unnormalized accepted target overlap is
\begin{align}
    G_j
       &:=
       p_jF_j                                                   \\
       &=
       \left|
          \sum_{E\in A}c_E\alpha_j(E)
       \right|^2
       =
       D_j+\Xi_j,
    \label{eq:coherent-target-overlap}
\end{align}
where
\begin{equation}
    D_j
       :=
       \sum_{E\in A}
          |c_E|^2|\alpha_j(E)|^2
    \label{eq:coherent-diagonal-part}
\end{equation}
and
\begin{equation}
    \Xi_j
       :=
       2\operatorname{Re}
       \sum_{\substack{E,E'\in A\\E<E'}}
       c_E^*c_{E'}
       \alpha_j(E)^*\alpha_j(E').
    \label{eq:coherent-interference-part}
\end{equation}

\begin{theorem}[Fidelity relation under coherent unitary noise]
\label{thm:coherent-fidelity}
Suppose that $p_1,p_2>0$, and define
\begin{equation}
    q_H^{\mathrm{diag}}
       :=
       \sum_{E\in H}|c_E|^2.
    \label{eq:diagonal-harmless-weight}
\end{equation}
Then
\begin{equation}
    F_2
       \geq
       \frac{p_1}{p_2}F_1
       -
       \frac{
          q_H^{\mathrm{diag}}
          +|\Xi_2-\Xi_1|
       }{p_2}.
    \label{eq:general-coherent-fidelity-bound}
\end{equation}
If $p_1=p_2=p_{\mathrm{acc}}>0$, defining
\begin{equation}
    \Gamma_{\mathrm{coh}}
       :=
       \frac{|\Xi_2-\Xi_1|}{p_{\mathrm{acc}}},
    \label{eq:coherent-interference-penalty}
\end{equation}
gives
\begin{equation}
    F_2
       \geq
       F_1
       -
       \frac{q_H^{\mathrm{diag}}}{p_{\mathrm{acc}}}
       -
       \Gamma_{\mathrm{coh}}.
    \label{eq:equal-acceptance-coherent-bound}
\end{equation}
\end{theorem}

\begin{proof}
From Eq.~\eqref{eq:coherent-target-overlap},
\begin{equation}
    p_2F_2-p_1F_1
       =
       (D_2-D_1)+(\Xi_2-\Xi_1).
    \label{eq:coherent-fidelity-difference}
\end{equation}
Because $C_1$ is Clifford and $\ket{\psi_1}$ is a stabilizer state, an accepted fault has $|\alpha_1(E)|^2=1$ for $E\in H\cup\{I\}$ and zero otherwise. Writing
\[
    L=A\setminus\bigl(H\cup\{I\}\bigr)
\]
and using $0\leq|\alpha_2(E)|^2\leq1$ gives
\begin{align}
    D_2-D_1
       &=
       \sum_{E\in H}
          |c_E|^2\bigl(|\alpha_2(E)|^2-1\bigr)
       +
       \sum_{E\in L}
          |c_E|^2|\alpha_2(E)|^2                              \\
       &\geq
       -\sum_{E\in H}|c_E|^2                                  \\
       &=
       -q_H^{\mathrm{diag}}.
    \label{eq:coherent-diagonal-bound}
\end{align}
Since $\Xi_2-\Xi_1\geq-|\Xi_2-\Xi_1|$,
\[
    p_2F_2-p_1F_1
       \geq
       -q_H^{\mathrm{diag}}
       -|\Xi_2-\Xi_1|.
\]
Dividing by $p_2$ proves Eq.~\eqref{eq:general-coherent-fidelity-bound}; the equal-acceptance case follows immediately.
\end{proof}

\subsubsection{Twirling coherent noise}

We now show that exact Pauli twirling removes the interference term and recovers the stochastic Pauli relation. Let
\[
    \overline{\mathcal P}
    :=
    \mathcal P/\{\pm 1,\pm i\}
\]
and, for $R,E\in\overline{\mathcal P}$, define the Pauli character
\begin{equation}
    \eta_R(E)
    =
    \begin{cases}
        +1, & [R,E]=0,\\
        -1, & \{R,E\}=0.
    \end{cases}
\end{equation}
For independently sampled twirling operators $R$, each fault-path amplitude transforms as
\begin{equation}
    c_E
    \longmapsto
    c_E\eta_R(E),
    \qquad
    |\eta_R(E)|=1.
\end{equation}
Character orthogonality gives
\begin{equation}
    \mathbb{E}_R
    \left[
        \eta_R(E)^*\eta_R(E')
    \right]
    =
    \delta_{E,E'}.
    \label{eq:pauli-character-orthogonality}
\end{equation}
Averaging therefore removes all terms with $E\neq E'$, so $\Xi_j^{\mathrm{tw}}=0$ and the twirled channel is
\begin{equation}
    \mathcal{N}_{\mathrm{tw}}(\rho)
    =
    \sum_E q_E E\rho E^\dagger,
    \qquad
    q_E:=|c_E|^2.
    \label{eq:twirled-pauli-channel}
\end{equation}
Because every Pauli fault path is either accepted or rejected deterministically, both circuits have the twirl-averaged acceptance probability
\begin{equation}
    p_{\mathrm{acc}}^{\mathrm{tw}}
    =
    \sum_{E\in A}q_E.
    \label{eq:twirled-acceptance}
\end{equation}
Applying the diagonal argument above yields
\begin{align}
    F_2^{\mathrm{tw}}
    &\geq
    F_1^{\mathrm{tw}}
    -
    \frac{
        \sum_{E\in H}q_E
    }{
        \sum_{E\in A}q_E
    }
    \nonumber\\
    &=
    F_1^{\mathrm{tw}}
    -
    \Pr_{\mathrm{tw}}
    \left(
        E\in H
        \,\middle|\,
        E\in A
    \right).
    \label{eq:twirled-pauli-bound}
\end{align}
Thus exact Pauli twirling eliminates the coherent-interference penalty and recovers the stochastic Pauli fidelity relation. The quantity $p_{\mathrm{acc}}^{\mathrm{tw}}$ is the acceptance probability of the twirl-averaged channel and need not equal that of the original coherent channel.

\subsection{Additional observations}

For completeness, we also consider several adversarial scenarios for noise and how they could invalidate our verification protocol. These scenarios do not seem to be physically realistic, and could not have passed our validation experiments where we measured fidelity at low amounts of doping and cross-entropy at medium amounts of doping.

\paragraph{Z-only error on $T$ gates.} If there is a pure $Z$ error on $T$ gates, those will be invisible to checks, by construction. Therefore, the syndromes will remain unchanged, but the fidelity will drop. 

\paragraph{Noise that stabilizes the circuit.} If noise is concentrated onto the spacetime stabilizer group of the Clifford circuit, we could observe a large fidelity in the Clifford circuit and a low fidelity in the doped circuit. Moreover, by definition, harmless errors cannot be detected via post-selection. However, this requires the hardware noise to be correlated with the global properties of the circuit that runs on it, and not just the gates or their relative arrangements in space and time.

\subsection{Independent validation of original open-boundary experiment}
\label{sec:independent_validation} 

\begin{figure}[t] \centering \includegraphics[width=.85\linewidth]{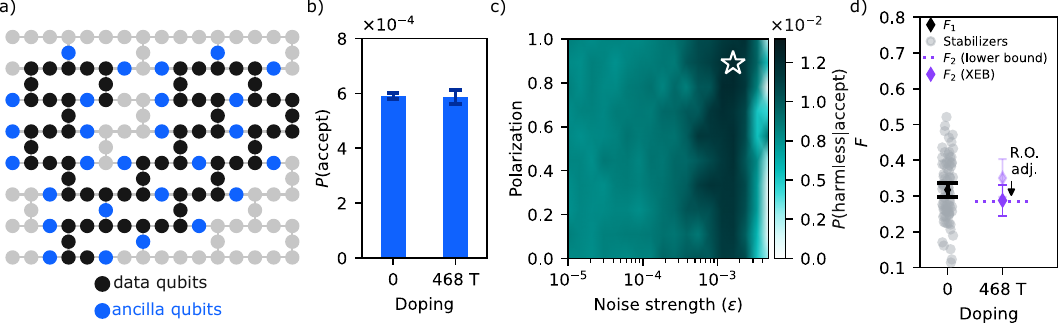} \vspace{0.5em} \includegraphics[width=.85\linewidth]{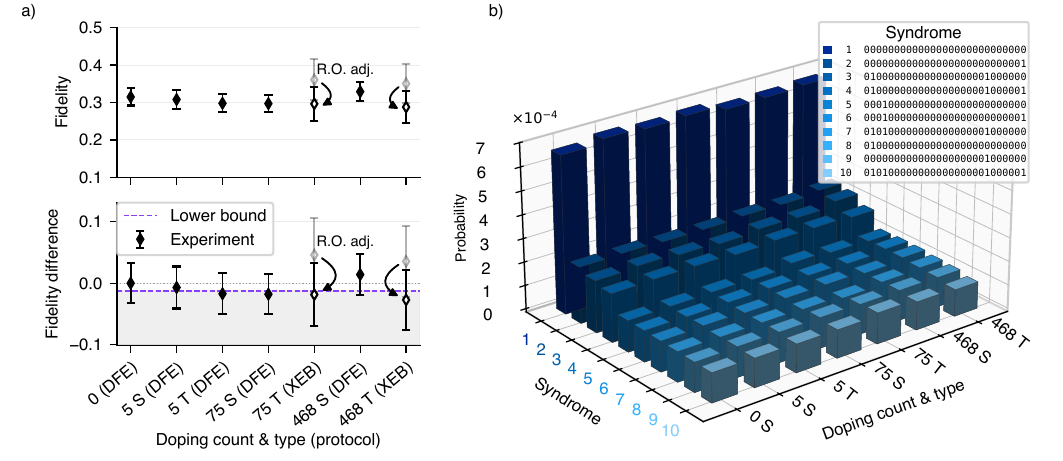} \caption{\textbf{Independent validation of the original open-boundary experiment.} Top: experimental setup and fidelity certification for the original 70-qubit, 468-$T$ experiment. Bottom: validation of the fidelity bound, augmented with the independent log-XEB estimate obtained from exact classical simulation of the 2,051 experimental outputs~\cite{manabe2026classical}. The independent calculation is statistically consistent with the experimentally estimated fidelity and supports the certified lower bound. Panels correspond to Figures~\ref{fig:fig3} and ~\ref{fig:fig4} of the main text but on a different linear chain that appeared in an earlier version of this work, respectively.} \label{fig:independent_validation} \end{figure}

The first version of this work reported a DCS instance with an open-boundary, 70-qubit, depth-70 circuit and 468 $T$ gates. Subsequent work exploited this geometry and computed the exact ideal probabilities of all 2,051 reported output samples using a tensor-network simulation requiring 256 H100 GPUs for 37.3 minutes, or approximately 159 H100-GPU-hours~\cite{manabe2026classical}. The resulting log-XEB fidelity proxy was \[ F_{\log\mathrm{-XEB}}=0.35034, \qquad 95\%~\mathrm{CI}=[0.29763,0.40305]. \] Under the Porter--Thomas and scrambled-noise assumptions required to interpret XEB as a fidelity proxy, this result is consistent with our DFE estimate $F_1=0.32(1)$ and independently corroborates the certified lower bound $F_2\geq0.284$,  see Figure~\ref{fig:independent_validation}.

\section{Fidelity estimation for stabilizer and low-magic states}
\label{sec:dfe}

We estimate target-state fidelities using direct fidelity estimation
(DFE)~\cite{direct_fid_est,dfedasilve}. We first review the formalism following Ref.~\cite{direct_fid_est}. For stabilizer target states we use the
standard DFE procedure, i.e., we sample random stabilizer observables from the target
state and measure their expectation values experimentally. For target states
containing a small number of non-Clifford gates, we extend this procedure by
sampling Paulis exactly from the DFE relevance distribution using a classical
marginal sampler described below.

\subsection{Direct Fidelity Estimation}
\label{app:dfe}

Let \(\rho=\ket{\psi}\!\bra{\psi}\) be the pure target state and let
\(\sigma\) denote the experimentally prepared state. We use the \(n\)-qubit
Pauli basis \(\mathcal{P}_n\), normalized such that
\(\operatorname{Tr}(P_kP_\ell)=d\delta_{k\ell}\), with \(d=2^n\). For any state
\(\tau\), define
\[
    \chi_\tau(P)=\frac{1}{\sqrt{d}}\operatorname{Tr}(P\tau).
\]
Since \(\rho\) is pure, \(\sum_{P\in\mathcal{P}_n}\chi_\rho(P)^2=1\). The DFE
relevance distribution is therefore \(p(P)=\chi_\rho(P)^2\), supported on Paulis
with \(\chi_\rho(P)\neq0\), and the fidelity can be written as
\begin{equation}
    F(\rho,\sigma)
    =
    \mathbb{E}_{P\sim p}
    \left[
        \frac{\chi_\sigma(P)}{\chi_\rho(P)}
    \right].
    \label{eq:dfe_expectation}
\end{equation}

Experimentally, for each sampled Pauli \(P_j\), we estimate
\(\langle P_j\rangle_\sigma=\operatorname{Tr}(P_j\sigma)\) from \(m_j\)
repeated measurements. If \(a_{j,r}\in\{\pm1\}\) is the outcome of shot \(r\),
then
\[
    \widehat{\langle P_j\rangle}
    =
    \frac{1}{m_j}\sum_{r=1}^{m_j}a_{j,r},
    \qquad
    \widehat{\chi}_\sigma(P_j)
    =
    \frac{1}{\sqrt{d}}\widehat{\langle P_j\rangle}.
\]
After sampling \(L\) Paulis \(P_j\sim p(P)\), the DFE estimator is
\begin{equation}
    \widehat{F}
    =
    \frac{1}{L}
    \sum_{j=1}^{L}
    \frac{\widehat{\chi}_\sigma(P_j)}
         {\chi_\rho(P_j)} .
    \label{eq:dfe_estimator}
\end{equation}
In the absence of experimental or calibration bias,
\(\mathbb{E}[\widehat{F}]=F(\rho,\sigma)\).

\subsection{Readout-error mitigation}
\label{app:dfe_readout_mitigation}

For each measured Pauli \(P\), we mitigate readout error by rescaling the
measured expectation value using an independently calibrated response factor.
After the basis-change circuit used to measure \(P\), the Pauli measurement is
reduced to a computational-basis parity measurement. If \(S\) is the set of
qubits on which \(P\) acts nontrivially, we define
\(Z_S=\prod_{i\in S} Z_i\). Using twirled computational-basis calibration data,
we estimate the corresponding readout fidelity
\(\lambda_P\equiv\lambda_{Z_S}\).

Under the twirled Pauli readout-noise model, the measured expectation value
satisfies \(\langle P\rangle_{\mathrm{meas}}=\lambda_P
\langle P\rangle_\sigma\). We therefore use
\(\widehat{\langle P\rangle}_{\mathrm{mit}}=
\widehat{\langle P\rangle}_{\mathrm{meas}}/\widehat{\lambda}_P\) and substitute
this quantity into Eq.~\eqref{eq:dfe_estimator}. Equivalently, the
readout-mitigated DFE estimator is
\begin{equation}
    \widehat{F}_{\mathrm{mit}}
    =
    \frac{1}{L}
    \sum_{j=1}^{L}
    \frac{
        \widehat{\langle P_j\rangle}_{\mathrm{meas}}
    }{
        \sqrt{d}\,
        \widehat{\lambda}_{P_j}\,
        \chi_\rho(P_j)
    } .
    \label{eq:dfe_mitigated_estimator}
\end{equation}
This procedure removes the multiplicative readout bias under the calibrated
Pauli readout-noise model. The variance of a mitigated Pauli estimate is
increased by a factor \(1/\lambda_P^2\), so Paulis with small readout fidelity
contribute larger statistical uncertainty.

\subsection{Sampling complexity}
\label{app:dfe_sample_complexity}

The DFE estimator has two statistical contributions: Monte Carlo error from
sampling Pauli settings and finite-shot error from estimating each sampled Pauli
expectation value.

For the idealized estimator with exact Pauli expectation values, let
\(X_P=\chi_\sigma(P)/\chi_\rho(P)\), where
\(P\sim\chi_\rho(P)^2\). Since \(|\operatorname{Tr}(P\sigma)|\leq1\), this
random variable is bounded by \(B_\rho\), where
\[
    B_\rho
    :=
    \max_{P:\chi_\rho(P)\neq 0}
    \frac{1}{\sqrt{d}\,|\chi_\rho(P)|}.
\]
Hoeffding's inequality gives
\[
    L =
    O\!\left(
        \frac{B_\rho^2}{\epsilon^2}
        \log\frac{1}{\delta}
    \right)
\]
sampled Pauli settings to estimate the fidelity to additive error \(\epsilon\)
with failure probability \(\delta\), ignoring finite-shot error.

For finite-shot measurements, the single-setting estimator for a sampled Pauli
\(P_j\) is
\[
    \widehat{X}_j
    =
    \frac{1}{m_j}
    \sum_{r=1}^{m_j}
    \frac{a_{j,r}}
         {\sqrt{d}\,\chi_\rho(P_j)} ,
    \qquad
    a_{j,r}\in\{\pm1\}.
\]
Thus the number of repetitions required for a fixed sampled Pauli scales as
\[
    m_j
    =
    O\!\left(
        \frac{1}
             {d\,\chi_\rho(P_j)^2\epsilon^2}
        \log\frac{L}{\delta}
    \right),
\]
up to constants from the chosen error allocation.

\subsection{Sampling Paulis for stabilizer target states}
\label{app:dfe_stabilizer_sampling}
For stabilizer target states, the DFE distribution is uniform over the
stabilizer group. Let
\[
    \mathcal{S}=\langle g_1,\ldots,g_n\rangle
\]
be the stabilizer group of the target state. We sample a uniformly random bit
string \(u\in\{0,1\}^n\) and measure
\[
    P(u)=\prod_{i=1}^n g_i^{u_i}.
\]
The stabilizer sign is tracked classically and included in \(\chi_\rho(P)\).
Because \(|\chi_\rho(P)|=1/\sqrt{d}\) for every sampled stabilizer Pauli, the
DFE distribution is uniform over the \(d=2^n\) stabilizer Paulis and zero
elsewhere. Thus \(B_\rho=1\), giving the standard dimension-independent scaling
\(L=O(\epsilon^{-2}\log(1/\delta))\) for DFE of stabilizer target states.
\subsection{Exact Pauli sampling for states with few \(T\) gates}
\label{app:dfe_few_t_sampling}

For non-stabilizer target states, the DFE relevance distribution
\(p(P)=\chi_\rho(P)^2\) is generally not uniform over a stabilizer group. For the low-magic target states used in this work, which are prepared by
Clifford+\(T\) circuits with a small number of \(T\) gates, we sample exactly
from this distribution without enumerating all \(4^n\) Paulis. Instead, we use
the chain rule together with exact evaluation of the required
marginals~\cite{hinscheEfficientDistributedInnerProduct2025}. Let \(\rho=U\ket{0^n}\!\bra{0^n}U^\dagger\) be the target state. We generate
the Pauli string sequentially. Suppose the single-qubit Paulis on the first
\(\ell\) qubits have been fixed, forming the partial Pauli string
\(Q_\ell=P_1\otimes\cdots\otimes P_\ell\). We define its marginal DFE weight as
\begin{equation}
    M(Q_\ell)
    :=
    \sum_{P_{\ell+1},\ldots,P_n}
    \chi_\rho
    \left(
        Q_\ell\otimes
        P_{\ell+1}\otimes\cdots\otimes P_n
    \right)^2 .
    \label{eq:dfe_partial_marginal}
\end{equation}
The conditional probability of extending \(Q_\ell\) by
\(A\in\{I,X,Y,Z\}\) is
\[
    \Pr(P_{\ell+1}=A\mid Q_\ell)
    =
    \frac{M(Q_\ell\otimes A)}
         {\sum_{B\in\{I,X,Y,Z\}}M(Q_\ell\otimes B)}.
\]

We compute the marginal weight of each partial Pauli string as a two-copy overlap. Let
\(R_\ell=\{\ell+1,\ldots,n\}\) denote the set of unassigned qubits, and let
\(\operatorname{SWAP}_{R_\ell}\) swap the two copies on those qubits~\cite{hinscheEfficientDistributedInnerProduct2025}. Define
\[
    O_{Q_\ell}
    :=
    \left(Q_\ell\otimes Q_\ell\right)
    \otimes
    \operatorname{SWAP}_{R_\ell}.
\]
Using Pauli orthogonality and the identity
\(\operatorname{SWAP}=2^{-m}\sum_{R\in\mathcal{P}_m}R\otimes R\) on an
\(m\)-qubit subsystem, Eq.~\eqref{eq:dfe_partial_marginal} becomes
\begin{equation}
    M(Q_\ell)
    =
    \frac{1}{2^\ell}
    \operatorname{Tr}
    \left[
        (\rho\otimes\rho)O_{Q_\ell}
    \right]
    =
    \frac{1}{2^\ell}
    \bra{0^{2n}}
    (U^\dagger\otimes U^\dagger)
    O_{Q_\ell}
    (U\otimes U)
    \ket{0^{2n}} .
    \label{eq:dfe_marginal_overlap}
\end{equation}

Thus the marginal weight of each partial Pauli string can be obtained from the
all-zero amplitude of the doubled circuit
\[
    C_{Q_\ell}
    :=
    (U^\dagger\otimes U^\dagger)
    O_{Q_\ell}
    (U\otimes U).
\]
In particular,
\(\bra{0^{2n}}C_{Q_\ell}\ket{0^{2n}}=2^\ell M(Q_\ell)\), which we calculate using a ZX-calculus-based circuit simulator~\cite{kissinger2022simulating}. Because
\(M(Q_\ell)\geq0\), the corresponding all-zero probability satisfies
\(p_0(Q_\ell)=2^{2\ell}M(Q_\ell)^2\), and hence
\(M(Q_\ell)=2^{-\ell}\sqrt{p_0(Q_\ell)}\).

At each step, the sampler requires only ratios of marginal weights for the four
possible extensions of the same partial Pauli string. The common factor \(2^{-(\ell+1)}\) therefore cancels, giving the
conditional sampling rule
\begin{equation}
    \Pr(P_{\ell+1}=A\mid Q_\ell)
    =
    \frac{\sqrt{p_0(Q_\ell\otimes A)}}
         {\sum_{B\in\{I,X,Y,Z\}}
          \sqrt{p_0(Q_\ell\otimes B)}} .
    \label{eq:dfe_probability_from_zero_prob}
\end{equation}

The sampler used in the experiment is therefore:

\begin{enumerate}
    \item Initialize the empty partial Pauli string \(Q_0=\emptyset\).

    \item For each \(\ell=0,\ldots,n-1\), evaluate the four possible
    extensions \(Q_\ell\otimes A\), where \(A\in\{I,X,Y,Z\}\).

    \item For each extension, construct the doubled circuit
    \[
        C_{Q_\ell\otimes A}
        =
        (U^\dagger\otimes U^\dagger)
        O_{Q_\ell\otimes A}
        (U\otimes U).
    \]

    \item Evaluate
    \(p_0(Q_\ell\otimes A)
    =\Pr_{C_{Q_\ell\otimes A}}(0^{2n})\)
    using the exact bit-string probability algorithm.

    \item Sample \(P_{\ell+1}=A\) according to
    Eq.~\eqref{eq:dfe_probability_from_zero_prob}, and update the partial
    Pauli string as \(Q_{\ell+1}=Q_\ell\otimes P_{\ell+1}\).
\end{enumerate}
After \(n\) steps, the output
\(P=P_1\otimes\cdots\otimes P_n\) is distributed exactly according to
\(\Pr(P)=\chi_\rho(P)^2\). This procedure avoids constructing the full
\(4^n\)-element distribution. Its cost is instead set by repeated exact
probability evaluations of doubled Clifford+\(T\) circuits, with the
non-Clifford cost controlled by the number of \(T\) gates.

\section{Medium-magic validation with cross entropy benchmarking} \label{sec:xeb}

Direct fidelity estimation of quantum states becomes difficult as the state moves far away from stabilizer, since the Pauli expectation values become broadly distributed and typically exponentially small. As a result, two distinct difficulties arise. First, identifying the relevant Pauli observables for importance sampling becomes computationally challenging, since the target state no longer admits a sparse stabilizer description. Second, the number of shots required to resolve small Pauli expectation values grows substantially.  

Even though a direct measurement of state fidelity becomes infeasible beyond a small amount of doping, we perform additional validations with a proxy for fidelity in the regime of medium amounts of doping. For this, we turn to cross-entropy benchmarking (XEB). This is feasible due to two key properties: doping with $O(n)$ $T$ gates is sufficient to make the state anticoncentrate, and strong simulation remains feasible in this regime for our problem size.

\subsection{Linear cross-entropy}

Linear cross-entropy benchmarking (XEB) compares samples from an experimental distribution \(q(x)\) to the ideal output probabilities \(p(x)=|\langle x|C|0\rangle|^2\) of the target circuit. For $N = 2^n$ and $M$ samples, the linear cross-entropy estimator is defined as

\begin{equation}
F_{\mathrm{XEB}}
=
\frac{N}{M}\sum_{i=1}^M p(x_i)-1,
\qquad x_i\sim q ,
\end{equation}
where the \(x_i\) are measured bitstrings from the experiment. Equivalently,
\begin{equation}
\mathbb{E}_{x\sim q}[F_{\mathrm{XEB}}]
=
N\sum_x q(x)p(x)-1 .
\end{equation}

Intuitively, XEB measures whether the experimental device outputs bitstrings that have large probability in the ideal distribution. A uniform sampler (i.e., completely noisy) gives \(F_{\mathrm{XEB}}\approx 0\), while a perfect sampler \(q=p\) gives
\begin{equation}
\mathbb{E}_{x\sim p}[F_{\mathrm{XEB}}]
=
N\sum_x p(x)^2-1
=
N M_2-1 .
\end{equation}

Thus the ideal XEB value is directly tied to the second moment, or collision probability, of the ideal output distribution. In a Haar-random state, the distribution follows a Porter-Thomas signature, where \(N M_2\approx 2\), so a noiseless experiment gives \(F_{\mathrm{XEB}}\approx 1\).  Thus, in order for XEB to be a meaningful benchmark, we require an anticoncentration condition such that at least the second moment follows Haar-like statistics~\cite{Boixo2018}.

\subsection{Anticoncentration of Clifford circuits with $T$-doping}

A useful way to understand the role of \(T\)-doping is to contrast the output distribution of Clifford circuits with those of Haar-random circuits. The Clifford group is highly random at low moments and forms a unitary 3-design~\cite{webb2015clifford}. However, this does not imply that Clifford output distributions are Porter-Thomas distributed. Indeed, the output distributions are flat over an affine subset of bitstrings, where the support is determined by constraints imposed by the stabilizer group.

Injecting non-Clifford resources changes this behavior. In particular, adding magic resources to Clifford circuits drives the overlap distribution towards the Porter-Thomas statistics of Haar-random states, with poly-logarithmically many magic resources already sufficient; see for example~\cite{magni2025anticoncentration, ghosh2023complexity, ghosh2025random}. Further work has shown~\cite{leone2021quantum,leone2026non} that relative-error convergence to Haar-random $k$-designs happens with an amount of doping linear in $n$ and $k$.

We numerically confirm this for our circuit family by computing the second moment of the distribution using a sampler (both classical and quantum) and computing the probability of the observed bitstrings. Figure~\ref{fig:anticoncentration} shows convergence of the scaled second moment of the distribution $2^n M_2 = 2^n \sum_x p(x)^2$ moving from 1 in the case of graph states to 2 in the case of anticoncentrated states. It also shows that for our particular experiment, the output bitstrings after a linear amount of doping show close proximity to a Porter-Thomas distribution.

\begin{figure*}[h]
\centering
\includegraphics[width=\textwidth]{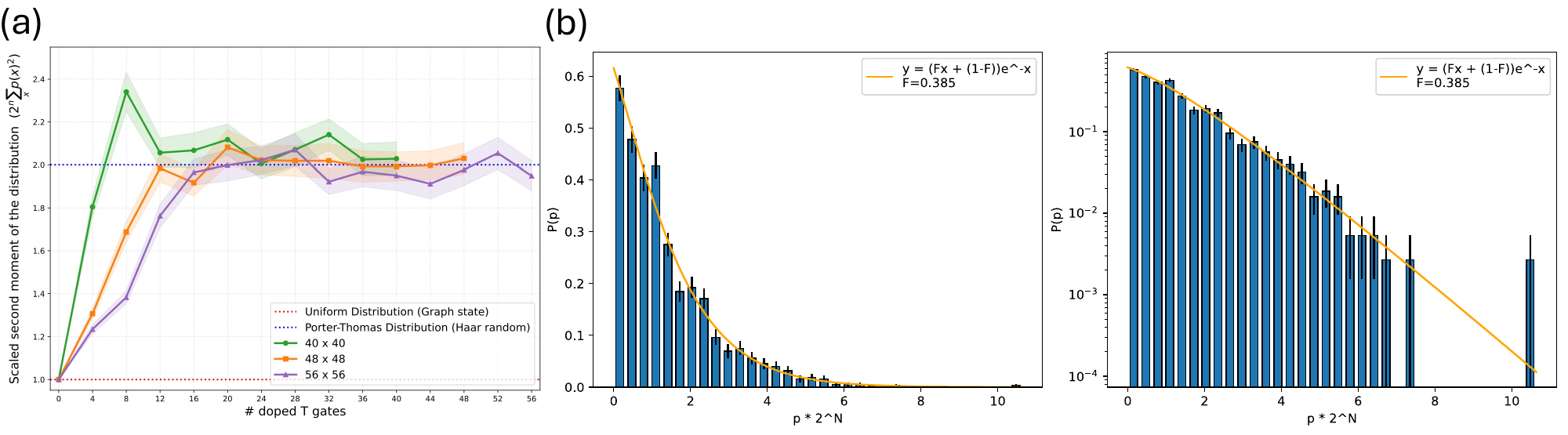} 
\caption{
{\bf Anticoncentration of Clifford circuits with $O(n)$ $T$ doping.} {\bf a.} Convergence of the second moment in simulations of $n \times n$ circuits. {\bf b.} Probability distribution of sample bitstrings experimentally obtained from a $64 \times 73$ circuit with $75$ $T$ gates shows good agreement with Porter-Thomas.
All simulations were performed using QuiZX~\cite{kissinger2022simulating}.} 
\label{fig:anticoncentration}
\end{figure*}

\subsection{Computing cross-entropy from experimental data}

Although exact fidelity estimation remains computationally expensive for generic non-Clifford circuits, the near-Clifford regime considered here still permits strong simulation of individual output amplitudes using state-of-the-art Clifford+\(T\) simulators. We compute the probabilities \(p(x_i)\) using \textbf{QuiZX}, whose ZX-calculus-based simulation methods can evaluate amplitudes for circuits with \(T\)-counts up to approximately \(80\), albeit with runtimes that can reach several days for the largest instances considered here. This makes XEB practically accessible in a regime where direct fidelity estimation would be substantially more challenging.

As noted in prior work~\cite{merkel_clifford_benchmarks,Arute2019}, even without any gate noise, SPAM errors introduce a discrepancy between fidelity and XEB; intuitively readout error can at worst contribute $-1$ to fidelity and $0$ to the XEB. Therefore, XEB can systematically overestimate fidelity in the presence of readout noise.

We mitigate this bias by characterizing readout errors on the qubits. Assuming that readout error is independent across qubits, the outcome probability distributions with and without readout error should be uncorrelated. As the independent correlation is vanishingly small at large system sizes, we can model the fidelity from readout error alone as:

$$F_{\text{readout}} = \prod_{i=1}^{n} (1 - p_i)$$

To compare the XEB score with the fidelities obtained from DFE, we adjust the XEB score for the readout error as shown in Fig.~\ref{fig:fig4}. Specifically, we estimate the effective DFE readout fidelity by averaging the readout fidelities of 10,000 randomly sampled Pauli observables (see Eq.~\eqref{eq:dfe_mitigated_estimator}). We then rescale the measured XEB score by the ratio of this average DFE readout fidelity to the corresponding XEB readout fidelity $F_{\text{readout}}$.

\subsection{Limitations of cross-entropy as a fidelity proxy}

Interpreting XEB as a proxy for state fidelity requires additional assumptions beyond anticoncentration. Under a depolarizing or sufficiently scrambled noise model, XEB tracks circuit fidelity~\cite{Boixo2018}. But this relationship can break down when the circuit does not anticoncentrate, the noise is strong, or spatially or temporally correlated~\cite{ware2023sharp, Morvan2024}. 

Unlike direct fidelity estimation, which probes the overlap of the full quantum state with the target state, XEB only depends on the computational-basis output probabilities. It is fundamentally a second-moment statistic and therefore does not fully characterize the output distribution or the underlying quantum state, and there exist known spoofing strategies capable of reproducing nontrivial XEB values without generating the target quantum state~\cite{GaoKalinowski2024,BarakChouGao2020}.

\section{Estimator consistency: combining post-selection with twirling} \label{sec:psrc}

Pauli twirling regions of the circuit (randomized compiling) convert coherent error processes into stochastic processes, making error mechanisms more well-behaved~\cite{wallman2016randomized,brillant2026mar,winick2022dec}.
In this section we analyze the estimation of observables from shots post-selected on ancilla outcomes under randomized compiling. We verify that pooling shots accepted by post-selection across randomizations before normalizing yields a consistent estimator of an expectation value of the twirl-averaged post-selected quantum state.

Let $\rho_0$ denote the initial quantum state on data and ancilla qubits, and $\mathcal{C}$ the ideal circuit of interest.
Suppose $\mathcal{C}$ contains $R$ locations into which we can compile single-qubit Paulis,
and let \(T \in \mathcal{P}_R\) denote a particular randomization frame drawn 
from the distribution $\Omega$ over $T$ prescribed by the randomization scheme (which will correlate Paulis at different locations, e.g., pre- and post-gate twirls of the same gate).
Of these $R$ locations, let $T_M$ represent $T$ restricted to the locations that directly precede a measurement instruction, i.e., those Paulis participating in measurement randomization.
Write $\mathcal{N}_T$ to represent the noisy implementation of $\mathcal{C}$ for some particular choice $T\sim \Omega$, composed with the ideal channel corresponding to $T_M^\dagger$.
This composition ensures $\mathcal{N}_T = \mathcal{C}$ in the noiseless limit for every $T$, rather than equal to $\mathcal{C}$ up to some Pauli.
Experimentally, it is realized by classically flipping measurement outcomes according to the $X$ / $Y$ component of $T_M$ (i.e., the bits that anticommute with $Z$-basis measurement).
The twirled channel is 
\begin{align}
    \bar{\mathcal{N}} = \int \mathcal{N}_T d\Omega(T).
\end{align}

Let $\Pi$ be the projection onto the syndrome of interest defined on the ancilla qubit subspace. Let $Q = \sum_x q(x) |x\rangle\langle x|$ be a diagonal observable of interest, specified by real coefficients $\{q(x)\}_{x\in\{0,1\}^n}$ and diagonal in the same computational basis in which the qubits are ultimately measured.
We have written $Q$ this way so that it can accommodate both Pauli observables, as used in DFE, and also the observable affinely related to the XEB statistic.
Note that any Pauli is brought into this diagonal form by compiling single-qubit change-of-basis Cliffords into the circuit just prior to $T_M$, after which $q(x) = \prod_i (-1)^{x_i}$ over the qubits $i$ where the rotated Pauli is $Z$. 

Lüders' rule gives the twirled, post-selected state along with the post-selected expectation value of Q as

\begin{align}
\frac{\Pi \bar{\mathcal{N}}(\rho_0) \Pi}{\mathrm{Tr}[\Pi \bar{\mathcal{N}}(\rho_0)]}
\qquad\text{and}\qquad
\langle Q \rangle_{\text{PS}} \;=\;
    \frac{\mathrm{Tr}\!\left[\,Q\, \Pi\, \bar{\mathcal{N}}(\rho_0)\, \Pi\,\right]}
        {\mathrm{Tr}\!\left[\,\Pi\, \bar{\mathcal{N}}(\rho_0)\,\right]},
\end{align}
respectively.

Supposing we draw $J$ circuit randomizations $T_1,\ldots,T_J\sim \Omega$ and perform $K$ shots of each of them, we define an estimator for $\langle Q \rangle_{\text{PS}}$ as 

\begin{align}
    \hat{Q}_{\text{PS}} \;=\; \frac{\sum_{j=1}^J \sum_{k=1}^K q_{j,k}\, n_{j,k}}
                   {\sum_{j=1}^J \sum_{k=1}^K n_{j,k}}
\end{align}
where \(x_{j,k}\in\{0,1\}^n\) is the bitstring observed on the data qubits, \(q_{j,k} := q(x_{j,k})\) is $Q$'s corresponding diagonal value, and
\(n_{j,k}\in\{0, 1\}\) records whether it passes post-selection for shot
\(1\leq k \leq K\) of randomization \(1\leq j \leq J\).

For any particular $j$ and $k$, we can split the indicator $n_{j,k}\in\{0, 1\}$ on its value and see that the post-selection rate $P(\text{pass}\mid T_j)=\mathrm{Tr}[\Pi \mathcal{N}_{T_j}(\rho_0)]$ cancels with the denominator of Lüders' rule, giving

\begin{align}
\mathbb{E}[q_{j,k}\cdot n_{j,k} \mid T_j]
\;=\; 1\cdot \mathbb{E}[q_{j,k}\mid\text{pass}, T_j]\cdot P(\text{pass}\mid T_j) \,+\, 0
\;=\; \mathrm{Tr}[Q\, \Pi\, \mathcal{N}_{T_j}(\rho_0)\, \Pi],
\end{align}
where the last equality is Born's rule for the diagonal observable $Q$ evaluated on the (unnormalized) post-selected, twirled state $\Pi \mathcal{N}_{T_j}(\rho_0)\Pi$. The result is independent of $k$.

Now we can use the law of total expectation to find both the expectation of $q_{j,k}\cdot n_{j,k}$ and the post-selection rate, unconditional on $T_j$,

\begin{align}
\mathbb{E}[q_{j,k}\cdot n_{j,k}] &= \int \mathbb{E}[q_{j,k}\cdot n_{j,k} \mid T] d \Omega(T)
= \mathrm{Tr}\left[Q\, \Pi\, \left(\int \mathcal{N}_T(\rho_0)d\Omega(T)\right) \Pi\right]
= \mathrm{Tr}\left[Q\, \Pi\, \bar{\mathcal{N}}(\rho_0)\, \Pi\right], \text{ and} \nonumber \\
\mathbb{E}[n_{j,k}] &= \mathrm{Tr}[\Pi\, \bar{\mathcal{N}}(\rho_0)].
\end{align}

Since $\hat{Q}_\text{PS}$ is the ratio of the two iid sample means $\sum_k q_{j,k} n_{j,k}$ and $\sum_k n_{j,k}$, we get

\begin{align}
      \hat Q_\text{PS} 
      \xrightarrow{J\to\infty} 
      \frac{\mathbb{E}[\sum_k q_{j,k} n_{j,k}]}{\mathbb{E}[\sum_k n_{j,k}]} 
      = \frac{\mathrm{Tr}[Q \Pi \bar{\mathcal N}(\rho_0)\,\Pi]}{\mathrm{Tr}[\Pi\,\bar{\mathcal N}(\rho_0)]}
      = \langle Q \rangle_{\text{PS}} .
\end{align}

Therefore, $\hat{Q}_\text{PS}$ is a consistent estimator of $\langle Q \rangle_{\text{PS}}$, the expectation of the diagonal observable $Q$ on the twirled, post-selected state, as the number of randomizations $J\rightarrow\infty$. Observe that consistency is independent of the shots per randomization, $K$.

\section{Experimental workflow and data acquisition} \label{sec:experiments}

In this section, we detail our experimental configurations and workflow; see Figure~\ref{fig:exp_workflow} for a visual overview.

\begin{figure*}[h]
\centering
\includegraphics[width=\textwidth]{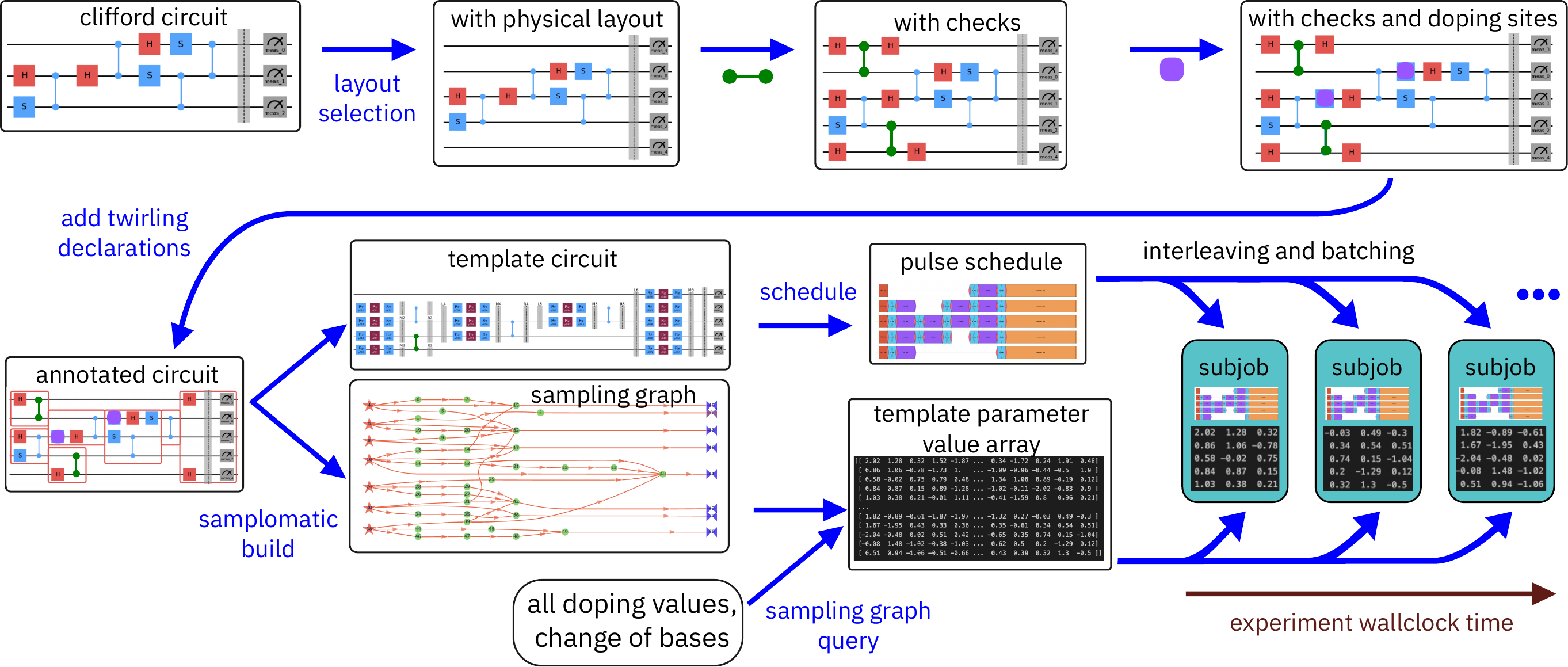} 
\caption{{\bf Experimental workflow.} 
The abstract Clifford circuit is compiled to a physical layout using heuristics based on $CZ$ layer fidelity experiments.
Ancillas and spacetime checks are added, followed by doping sites.
Twirling and change-of-basis intent is declared by enclosing twirling regions of the circuit of interest in Qiskit box instructions annotated with samplomatic directives. This circuit contains parametric $RZ$ gates at valid doping sites. The Samplomatic build mechanism consumes this circuit to produce a template circuit and sampling graph pair. The template circuit is compiled into a parametric low-level circuit schedule by the IBM Quantum Compute Service. The sampling graph is provided with all sets of doping and basis-change configurations, and generates a specified number of randomizations for each configuration in the form of an array of parameter value vectors that are valid input to the template circuit, and hence also to the circuit schedule. These randomizations are carefully interleaved so that all configurations undergo the same drift mechanisms, and batched into subjobs that are executed serially. Subjobs employ fast, late-stage binding of parameter values.}\label{fig:exp_workflow}
\end{figure*}

\subsection{Identifying layouts}

The experiments we present are conducted on IBM Boston, a Heron R3 processor with heavy-hex architecture. Our circuits map data qubits onto a 1D chain with periodic boundary conditions, with ancillas having support on exactly one qubit.

As the operations on data qubits comprise the majority of the $CZ$ gate count, and therefore the error budget, we prioritize identifying an optimal ring of qubits. Note that fixing the data qubit layout constrains the number of available checks and their topology --- in contrast, fixing a layout to a predetermined number of checks and supports can limit the number of available rings and the maximum fidelity. 

To speed up the search for layout candidates, we benchmark the device and exclude qubits with high error rates. Using 1Q and readout benchmarks, nodes with single qubit gate error above 0.5\% and readout error above 2.5\% are discarded. We characterize $CZ$ gate errors with layer fidelity along grids on the device \cite{layerfidelity}, discarding edges with errors above 1\% and durations above $98\,\mathrm{ns}$ (standard being $68\,\mathrm{ns}$).

On the reduced device map, we search for 9 chains with the lowest error metric

$$ \prod_{k=1}^{n} (1 - \epsilon_{readout})^{k} (1 - \epsilon_{1Q})^k \cdot \prod_{l=1}^{(n-1)/2} (1 - \epsilon_{2Q})^{nl/2}$$

that also admit the largest number of ancilla qubits. Ancillas are chosen from qubits neighboring the chain, provided that $CZ$ gate error is below 0.51\%. When potential ancillas have support on multiple data qubits, we choose the one closest to the center. This heuristically increases the likelihood of detecting errors in the circuit.

As layer fidelity errors are sensitive to the crosstalk specific to the layout, we retake these benchmarks on the 9 chains and again rank them using the metric above.

\subsection{Evaluating check performance}

After obtaining the set of ancilla supports from the layout, we can equip our circuits with spacetime Pauli checks. We use the check picking algorithm as described in Section \ref{sec:check_and_doping}. In particular, we use the windowed algorithm with parameter $f = 0.5$, cost function $\Gamma$ of the undetected noise channel, and input a $CZ$ error of $0.2\%$ and coherence $200\,\mathrm{\mu s}$ to the underlying Lindblad noise model.

After carefully selecting a loop layout of length 64, we isolated 19 ancilla candidates able to host Pauli checks (i.e. adjacent to the main loop).
Picking checks for those 19 locations shows that some locations do not provide any new coverage. Figure~\ref{fig:check_picking} shows such a check picking trace. We can easily notice that some checks do not provide any new coverage: the post-selection rate remains unchanged when adding those checks and the resulting fidelity ends up slightly worse (due to hook errors). Overall, for our experiment, we retained 12 locations, discarding 7 ancilla qubits.
\begin{figure*}[h] 
    \centering
\includegraphics[width=0.5\linewidth]{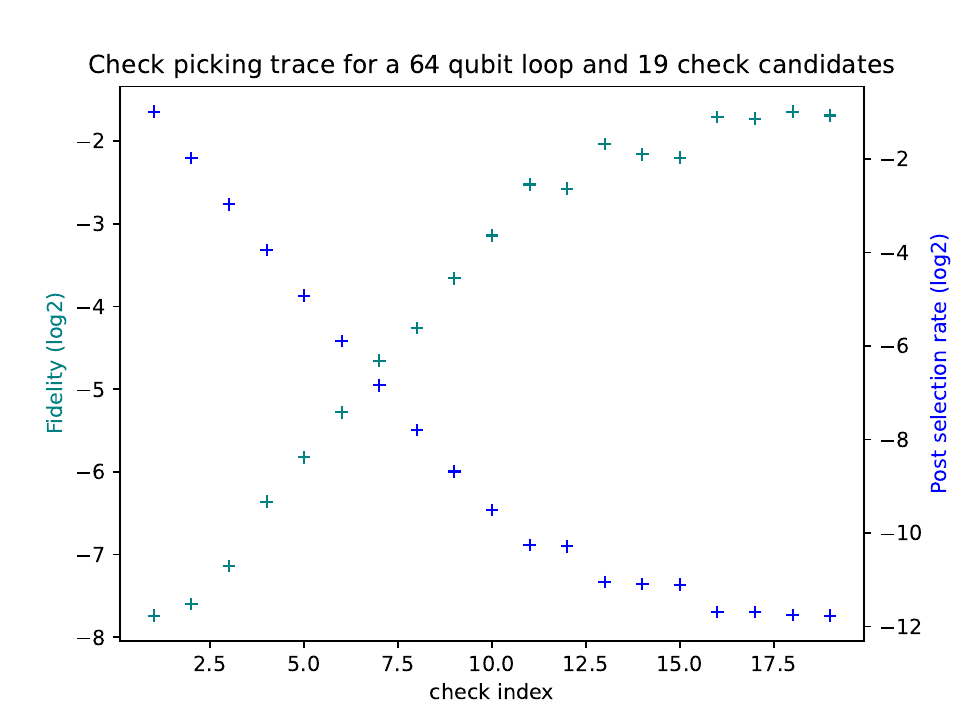}
\caption{  {\bf Check picking trace.} We picked 19 checks with locations guided by the chosen experiment layout. Each point is simulated via Monte Carlo with $10^8$ shots. We can clearly see that some check locations do not bring new coverage: the post-selection rate remains unchanged and the fidelity slightly degrades. Those locations are then removed by hand and new checks are picked until no such situation occurs anymore. For our experiment, we ended up discarding 7 of the 19 available ancilla.}
\label{fig:check_picking}
\end{figure*}

When the checks have been inserted, doping sites are selected as described in Section~\ref{sec:t-doping} and inserted into the circuit as parametric $RZ$ gates. This ensures that all doped circuits are the same as the base circuit, modulo a parametric change in virtual rotation angles. As an example, $T$ doping a particular site corresponds to setting its parameter value to $\pi/4$.

\subsection{Twirling implementation}

We use Qiskit's Samplomatic library to implement randomized compiling, parametric basis changes, and readout twirling~\cite{samplomatic,wallman2016randomized,PhysRevA.105.032620}. To achieve this, individual $CZ$ gates and the single-qubit gates that precede them are wrapped in twirl-annotated boxes, and similarly, measurement layers are wrapped in boxes annotated with change-of-basis and twirling annotations. The annotations are configured to the three-parameter $RZ(a)-RX(b)-RZ(c)$ decomposition.

We emphasize that this method of twirling $CZ$ gates individually departs from typical implementations of twirling which separate the circuit into layers or rounds of ``easy'' (single qubit) and ``hard" (two qubit) gates. This is more amenable for our circuit, which has many unique layers of $CZ$ gates. Twirling in layers significantly increases the number of nontrivial single qubit gates added to the circuit, particularly on the ancillas, and the overall depth. In practice, this risks creating hook errors and generally has lower fidelities and post-selection rates. However, twirling individual $CZ$ gates also introduces long idling times after each gate, which are therefore mitigated with $XY4$ dynamical decoupling sequences, and presents the possibility that long range crosstalk or other out-of-model errors are not suppressed adequately.

Samplomatic's build mechanism produces a tuple containing a parametric quantum circuit, called the template circuit, and a sampling graph, the latter of which can be used to query for arbitrary numbers of random parameter sets under any doping and change-of-basis configuration.
Each random parameter set is expressed as a vector of real numbers that is a valid input to the template circuit.
IBM's low-level software stack supports fast late-stage binding of parameter values to a fixed parametric circuit.
This ensures that every shot executed, independent of twirling, doping or change of basis values, has exactly the same timing schedule, and varies only in $RZ$ phase values and $RX$ gate amplitudes.

\subsection{Experiment settings and results}

\paragraph{Transpilation and Scheduling:} We transpile circuits in Qiskit using the ``light optimization'' setting and do not reroute the circuit.
Circuit instructions are scheduled with the as-late-as-possible (ALAP) option and the standard delay of $250\,\mathrm{\mu s}$ between shot executions (see Figure~\ref{fig:loop_circuit_schedule} for a diagram). Extra pulses are added to implement symmetric $XY4$ dynamical decoupling~\cite{symmetricdd} during idle periods, and post-selection gadgets are used for correlated non-Markovian errors~\cite{olepaper}, which adds measurements on spectator qubits adjacent to the data and ancilla qubits. This extra post-selection step typically improves fidelity by $10\%$ at the cost of a $15\%$ reduction in shot survival. We benchmark the fidelity of the gadget before the experiments and ignore edges with values below 0.9, which can be indicative of poor readout or 1Q gate calibrations. We also report that we use an update to Qiskit primitives that initializes every qubit in a random $Z$ frame before each execution, which has been shown to average out coherent crosstalk in single qubit gates~\cite{phaserandomization}.

\begin{figure*}[h] 
    \centering
\includegraphics[width=0.95\linewidth]{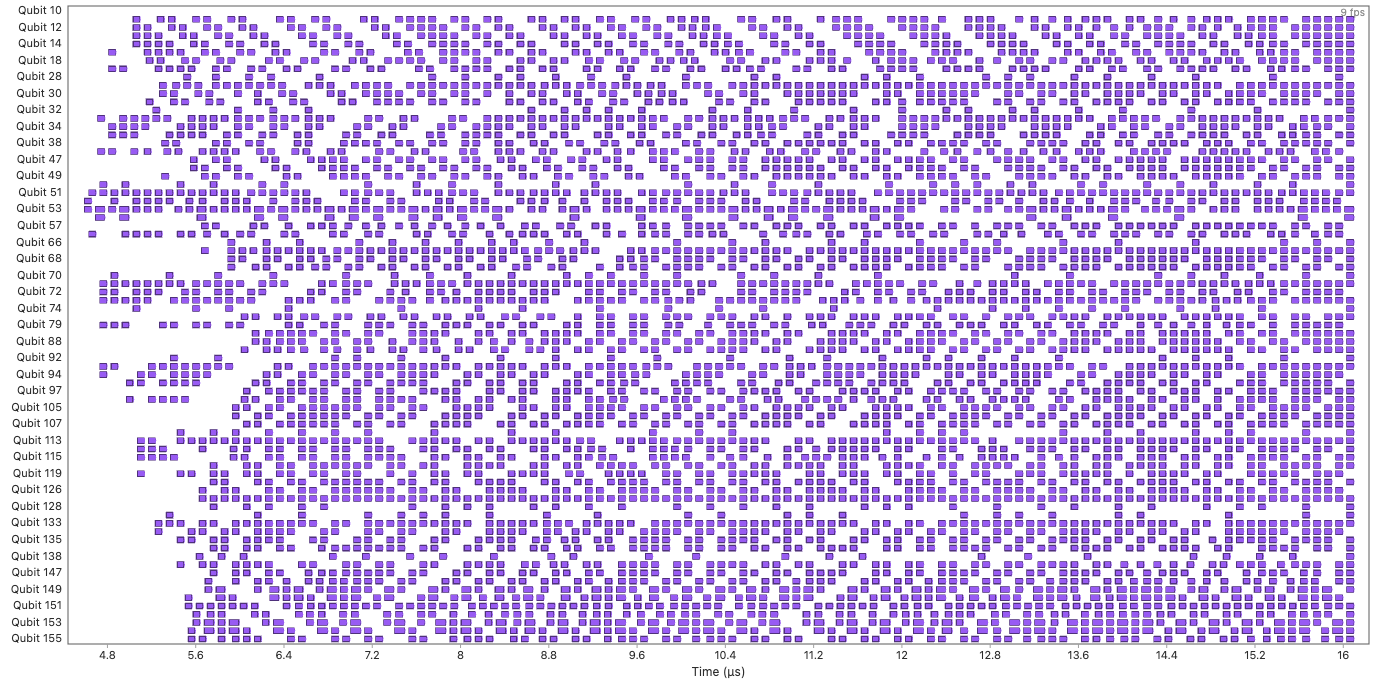}
\caption{  {\bf Circuit schedule.} Some qubit labels are omitted on the y-axis for legibility. With ALAP scheduling, the circuit is unstructured, having many unique layers of CZ gates. }
\label{fig:loop_circuit_schedule}
\end{figure*}

\paragraph{Parameterization:} Each circuit is parameterized with various levels of $I$, $S$, $Z$, or $T$ doping and an appropriate set of Pauli bases for measurements. When doping with Clifford gates ($I$, $S$, $Z$), we measure in random stabilizer bases to perform DFE (see Section \ref{sec:dfe}). Stabilizers here are drawn by multiplying over randomly chosen subsets of the generators. All $T$-doped circuits are sampled in the $Z$ basis to perform cross-entropy benchmarking (see Section \ref{sec:xeb}). For each doping and Pauli basis parameterization input, we fix the seed and draw the same 50 twirl instances.

\paragraph{Subjob construction:} To efficiently schedule the large number of parameterizations, we portion them into separate subjobs. We prioritize measuring all of the doping and basis change parameterizations together to ensure each experiences a similar noise environment. Therefore, each subjob chooses a random subset of 5 twirls and runs the checked circuit and a TREX readout calibration circuit (excluding spectator qubits) parameterized over this subset \cite{PhysRevA.105.032620}. We then randomly order this reduced set, which should reduce systematic biases from drift. Furthermore, the 5 twirls are drawn such that each contiguous block of 10 subjobs covers the entire set of 50 twirls exactly once.

\begin{table*}[!t]
\centering
\small
\setlength{\tabcolsep}{5pt}
\renewcommand{\arraystretch}{1.05}

\begin{tabular}{@{}rcccccrr@{}}
\toprule
& \multicolumn{7}{c}{Experiment 1} \\
\cmidrule(lr){2-8}
& Metric & Fid. & Fid. (RM) & PSR (\%) & Shots & Bases & Time (s) \\
\midrule
0       & DFE         & 0.38(1) & 0.55(2) & 0.0383(7) & 20 M     & 80   & 5,551 \\
314 (T) & Lower bound & 0.349   & ---     & 0.040(1)  & 3.5 M    & 1    & 972   \\
314 (S) & DFE         & 0.38(1) & 0.55(2) & 0.0379(6) & 20 M     & 80   & 5,551 \\
314 (Z) & DFE         & 0.38(1) & 0.54(2) & 0.0377(6) & 20 M     & 80   & 5,551 \\

\midrule
& \multicolumn{7}{c}{Experiment 2} \\
\cmidrule(lr){2-8}
0       & DFE         & 0.38(2) & 0.59(3) & 0.040(1)  & 5 M      & 40   & 1,426 \\
5 (T)   & DFE         & 0.42(2) & 0.64(3) & 0.0405(5) & 35.125 M & 40   & 7,489 \\
5 (S)   & DFE         & 0.39(2) & 0.61(3) & 0.041(1)  & 5 M      & 40   & 1,426 \\
75 (T)  & XEB         & 0.38(4) & 0.51(5) & 0.040(1)  & 3 M      & 1    & 856   \\
75 (S)  & DFE         & 0.42(2) & 0.65(3) & 0.040(1)  & 5 M      & 40   & 1,426 \\
314 (T) & Lower bound & 0.335   & ---     & 0.042(1)  & 3 M      & 1    & 856   \\
314 (S) & DFE         & 0.40(2) & 0.61(3) & 0.042(1)  & 5 M      & 40   & 1,426 \\

\midrule
& \multicolumn{7}{c}{Experiment 3} \\
\cmidrule(lr){2-8}
0       & DFE         & 0.0029(5) & 0.0045(8) & 40.69(2) & 9.4 M & 1000 & 2,905 \\
\bottomrule
\end{tabular}

\caption{\textbf{Summary of experimental results.}
All experiments use 50 twirls and TREX calibrations with 20,000 shots
per twirl, divided among 100, 100, and 50 subjobs, respectively.
Fidelities before and after readout mitigation (RM) are shown with
standard errors. At maximum doping, we report the one-sided normal
95\% lower confidence limit obtained from the zero-doping fidelity and
the maximum fidelity loss from harmless-to-harmful fault conversion.
Post-selection rates (PSRs) include syndrome and correlated
non-Markovian error checks. Shots are split equally among bases, except
in the 5-T experiment, where they are allocated according to the
measured observable's expectation value. Device usage times are obtained
from Qiskit and apportioned according to shot count.}
\label{tab:tab_exp_results}
\end{table*}

\paragraph{Fidelity validation experiments:} As described in the main text, we take two experiments to build confidence that fidelity is consistent across all doping values (up to a numerically-estimated drop of $\approx$1\%). At their core, these consist of collecting samples at maximum $T$ gate doping, which we claim will be difficult to simulate classically, and simultaneous benchmarks with DFE at zero doping, which validate the fidelity of the samples. 

As a note on reporting uncertainties and bounds for fidelity, recall that, with DFE, fidelity is estimated by sampling a finite number of stabilizers. To quantify the uncertainty in these estimates, we use a normal approximation. Error bars correspond to two-sided 95\% confidence intervals, computed as the estimate $\pm 1.96$ standard errors.  Numerical values in the text are reported using the standard uncertainty notation, in which the number in parentheses gives the estimated standard error in the last reported digit of the quoted value. For the lower bound on the fidelity of the maximally doped state (which is not actually measured unlike the values above), we again use the normal approximation and report the difference between the lower 97.5\% confidence limit of the fidelity and the upper 97.5\% confidence limit of the maximum drop from harmless error conversions.

In Experiment 1, we test whether fidelities will show discrepancies when utilizing many stabilizer measurements. This allows us to estimate the fidelity with high precision and to bound the true fidelity with high confidence (as described above). We therefore add only two benchmarks, maximum doping using $S$ and $Z$ gates, to check whether doping is error-free. The $Z$ doped fidelity, in fact, should be precisely the same; up to a phase, we measure the exact same set of stabilizers as in the undoped case. Our results show no evidence of significant differences in the fidelities, as all pairwise differences are consistent within one standard error (see Table~\ref{tab:tab_exp_results}). Furthermore, at 95\% confidence, the smallest distribution-free lower confidence limit from Hoeffding bound is 0.11. Seeing as this is the worst-case assumption, we have evidence that we can prepare states with above-zero fidelity.

In Experiment 2, we relax the number of stabilizer measurements, seeking to compare fidelities over the intermediate range of doping and across different fidelity metrics. We choose random sites to dope both with $T$ and $S$ gates: first at low doping ($\ll n$), where it is still tractable to perform DFE for both, medium doping ($O(n)$), where we can compare XEB for $T$ and DFE for $S$ gates, and maximum doping ($O(n^2)$), where only DFE for $S$ gates can be measured. Although there are larger discrepancies in the fidelities than in Experiment 1, we find that the lowest doped fidelities are consistent with the maximum 1\% drop (see Figure~\ref{fig:fig4}), and the normal 95\% confidence intervals on the pairwise differences all overlap with zero.

As an additional comment on the classical simulation time in Experiment 2, we use the QuiZX package to compute the amplitudes of outcomes and therefore XEB. We ran one job per bitstring resulting in 1197 independent runs. These are scheduled in parallel on a local cluster, with each job using 56 threads. We use nodes either with the AMD EPYC 7763 64-Core Processor (128 CPUs, 1.9 TB memory) or the AMD EPYC 9634 84-Core Processor (168 CPUs, 2.9 TB memory). We report that each sample takes an average of 54,041 seconds to complete - a nontrivial amount that exceeds total sample collection time on quantum hardware.

In Experiment 3, we measure the fidelity of the unencoded circuit without spacetime Pauli checks or doping. Due to the absence of checks, this circuit is shorter and exactly follows the brickwork $CZ$ layers, which requires modifications to the set of calibrations (see Section \ref{sec:hardware}). As reported, the fidelity in Experiment 1 (using syndrome post-selection) is $131 \times$ larger than the bare circuit value.

\paragraph{Additional comments on post-selection:}
Post-selection presents two challenges for our experiments. To properly implement twirling, we require additional overhead for shots. Otherwise, twirls can be left ``empty", i.e., with no surviving shots. As noted in Section \ref{sec:psrc}, it is the number of unique twirls, not the shots per twirl, that reduces the bias in our estimator for expectation values. This motivates our large shot budgets for Experiment 1 and 2, where we respectively observe a maximum of 14 and 27 empty twirls across all measured stabilizers.

Increasing the shots, however, results in a longer experiment. This risks degradations in overall fidelity and post-selection rate from calibration drift and other sources of time-dependent noise. In Figure~\ref{fig:fig_2608_cumul_avg}, which shows post-selection rates per subjob, there is a negative trend in post-selection rates that exceeds the scale expected from shot noise. As described in Section \ref{sec:hardware}, this is partially addressed for Experiment 2 by recalibrating pulse parameters midway through. It remains an open question to determine the set and frequency of calibrations that can best stabilize the fidelity and post-selection rate.

\begin{figure*}[h!] 
    \centering

\includegraphics[width=1.0\linewidth]{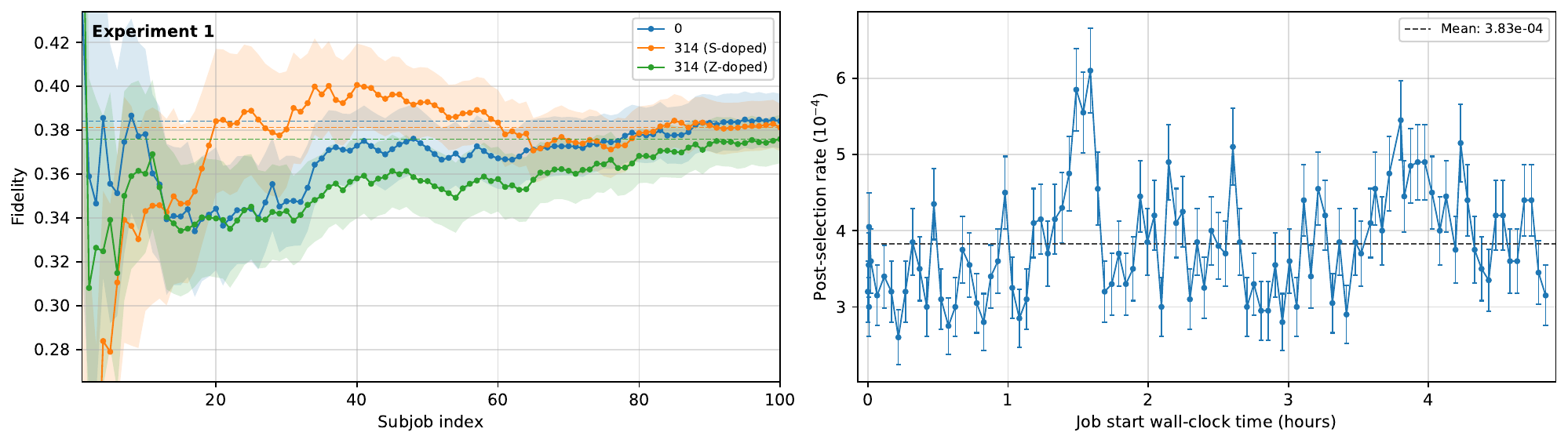}

\includegraphics[width=1.0\linewidth]{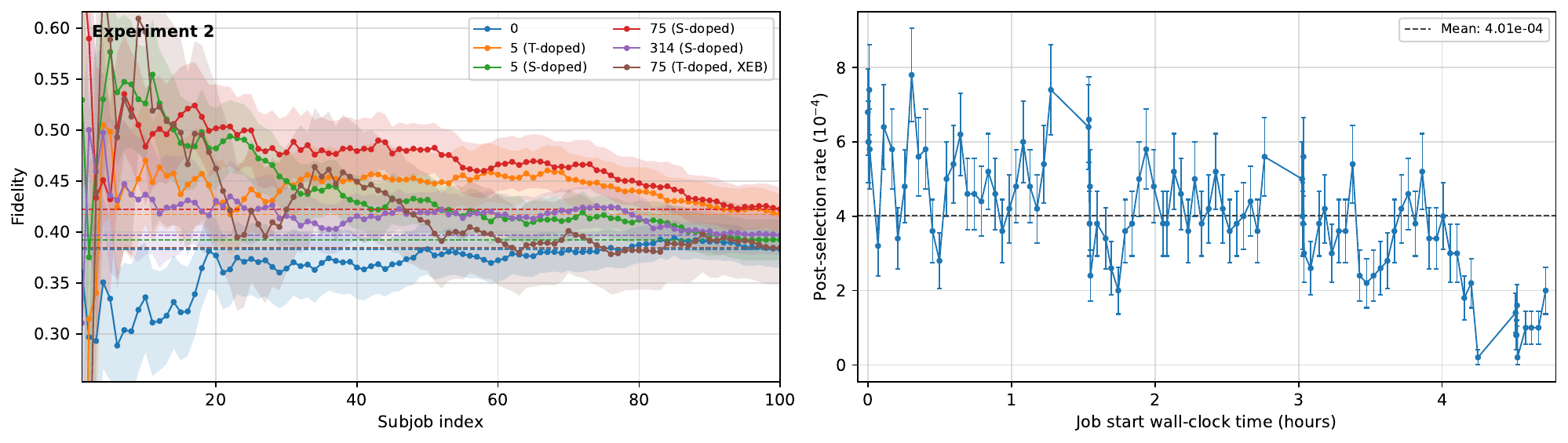}

\caption{{\bf Fidelities and post-selection rates over subjobs}. (Left) Fidelities for Experiment 1 (top) and 2 (bottom), cumulatively averaged over subjobs. Discrepancies appear to converge over time. (Right) Post-selection rate for each subjob with wall-clock time for Experiment 1 (top) and Experiment 2 (bottom). Post-selection rate degrades in both over time, possibly from calibration drift or TLS. Gaps in time for Experiment 2 are due to periodic recalibrations of single-qubit gates and readouts.}
\label{fig:fig_2608_cumul_avg}
\end{figure*}

\section{Hardware properties and calibration} \label{sec:hardware}

\subsection{Calibration procedure}

To achieve the high fidelities in our experiments, we recalibrate pulse parameters for qubits in the layout before each experimental run. Aside from exceptions mentioned later, we follow the typical automatic recalibration schemes for 1Q gates, 2Q gates, and readout (spectator qubits only recalibrate readout). These include rebiasing electrodes on each qubit to minimize coupling to neighboring two-level systems (TLS) and maximize $T1$ \cite{tlsbias}. While this was done to maximize performance, calibrating with the averaged noise strategy may help stabilize $T1$ and therefore the post-selection rate over long periods.

 To note changes we make to the typical calibration scheme:
 \begin{itemize}
  \item To mitigate the effects of crosstalk in $CZ$ gates, we calibrate them simultaneously in batches that match the circuit's layer structure. Recall that the checked circuits are largely unstructured and contain many unique layers of CZ gates. We therefore use a sparse batching scheme with minimum distance 5 (Experiment 1 and 2), while the unchecked circuit batches according to the two underlying brickwork layers (Experiment 3).
  
  \item $CZ$ gates can also have long-lasting transients from signal distortion in the flux control lines \cite{olepaper}. So, when a $CZ$ gate is repeated in quick succession, the latter gates can experience $IZ$ and $ZI$ over-rotations. To correct for this, we characterize the over-rotation between adjacent layers of $CZ$ gates for a few points representative of the delays in our circuit. The circuit is then recompiled with $RZ$ gates applied directly before each $CZ$ gate and interpolated angles (matching the delay) to cancel this over-rotation.
  
  \item Some qubits exhibit large fluctuations in frequency, suspected to be from TLS effects. For these, we calibrate frequencies over the average of ten back-to-back experiments. 
\end{itemize}
The resultant gate errors after these calibrations are reported in Table~\ref{tab:tab_exp_benchmarks}.

For the second experiment described in Section~\ref{sec:experiments}, we recalibrate single-qubit gate parameters, readout, and qubit electrode bias periodically, which leaves a gap in shot execution.  We refrain from recalibrating CZ gates, which are the most time-consuming step.

\begin{table*}
\centering
\begin{tabular}{@{}rcccccrccccc@{}}\toprule
& \multicolumn{5}{c}{Data Qubit Metrics} & \phantom{abc}& \multicolumn{5}{c}{Ancilla Qubit Metrics} \\
\cmidrule{2-6} \cmidrule{8-12}
& Mean & Med. & Min. & Max. & Std. && Mean & Med. & Min. & Max. & Std. \\ \midrule
$RX$ Error (\%) & 0.021 & 0.017 & 0.0079 & 0.096 & 0.015 && 0.021 & 0.018 & 0.012 & 0.047 & 0.0098 \\
$CZ$ Error (\%) & 0.18 & 0.148 & 0.070 & 0.612 & 0.098 && 0.147 & 0.127 & 0.080 & 0.249 & 0.054 \\
RO Error (\%) & 0.432 & 0.342 & 0.073 & 1.95 & 0.321 && 0.458 & 0.305 & 0.171 & 1.27 & 0.345 \\
$T1$ ($\mu$s) & 247.4 & 245.2 & 45.3 & 365.7 & 59.8 && 238.9 & 262.2 & 44.4 & 312.2 & 80.2 \\
$T2$ ($\mu$s) & 293.3 & 298.0 & 68.7 & 487.0 & 100.8 && 273.7 & 305.0 & 88.5 & 379.5 & 82.4 \\
\bottomrule
\end{tabular}
\caption{{\bf Gate and qubit metrics for the 64 data-qubit and 12 ancilla-qubit layout on IBM Boston at the beginning of Experiment 1.} Metrics are taken from the backend-properties snapshot associated with the first main job. $RX$ and $CZ$ errors are backend-reported gate-error estimates. Data-qubit $CZ$ statistics use the 64 specified data--data couplers, while ancilla $CZ$ statistics use the 15 specified couplers incident to an ancilla. RO denotes the backend-reported readout error; $T1$ and $T2$ are the reported coherence times.}
\label{tab:tab_exp_benchmarks}
\end{table*}

\subsection{Noiseless implementation of $T$ gates} \label{sec:noiseless-t}

A key part of our verification scheme relies on $T$ gates being implemented noiselessly. This is natural to superconducting qubit platforms, such as IBM's, in which single-qubit gates are driven by microwave pulses and rotational $Z$ gates can be implemented virtually by frame-tracking in software~\cite{virtualz}. To illustrate this, consider a microwave pulse with amplitude $\Omega$ and phase offset $\gamma$, which implements single-qubit unitary

$$ U = e^{-i \frac{\Omega T}{2} [ cos(\gamma)\hat{X} + sin(\gamma)\hat{Y} ] } $$

The axis of rotation in the $XY$ plane, then, can be modulated through $\gamma$, with amplitude $\Omega$ and time $T$ collectively defining the angle of rotation around said axis.

Any rotational $Z$ gate can therefore be applied in software by modifying this phase offset and adding it to the next rotational $X$/$Y$ pulse. So, $Z$ gates are implemented with no pulse, which means no calibration error, relaxation error, or leakage. The only source of imperfection comes from the stability of the reference oscillator defining the rotating frame, which is estimated to be less than $10^{-8}$~\cite{ball2016role}.\footnote{For open quantum systems, it has also been demonstrated that rotational $Z$ gates can change the trajectory of the gate through the Bloch sphere and therefore change the noise on the gate \cite{symmetricdd}. While this can adversely affect dynamic decoupling sequences, which repeat a sequence of pulses to cancel out idling or other coherent errors, we ignore these effects for our random circuits.}

\newcommand{\CZ}{\mathrm{CZ}}
\newcommand{\CZe}{\CZ_{\mathrm e}}
\newcommand{\CZo}{\CZ_{\mathrm o}}
\newcommand{\HH}{H^{\otimes n}}
\newcommand{\Fe}{F_{\mathrm e}}
\newcommand{\Fo}{F_{\mathrm o}}
\newcommand{\R}{R}
\newcommand{\Xlab}{\mathbf X}
\newcommand{\Zlab}{\mathbf Z}
\newcommand{\poly}{\mathrm{poly}}
\newcommand{\Cauchy}{\operatorname{Cauchy}}
\newtheorem{assumption}{Assumption}
\newtheorem{fact}{Fact}
\newcommand{\C}{\mathbb C}
\newcommand{\Prb}{\mathbb P}
\newcommand{\Unif}{\operatorname{Unif}}
\newcommand{\TV}{d_{\mathrm{TV}}}
\newcommand{\diag}{\operatorname{diag}}
\newcommand{\abs}[1]{\left\lvert #1\right\rvert}
\newcommand{\norm}[1]{\left\lVert #1\right\rVert}

\section{Asymptotic complexity of sampling T-doped Clifford circuits}\label{sec:asymptotic_hardness}

\noindent In this section, assuming appropriate complexity-theoretic conjectures, we will prove worst-case and average-case hardness of classical simulation for our circuit ensemble.

For simplicity and for better presentation, we will first show asymptotic worst-case hardness for open boundary conditions. Then, we will have a concluding part (in \Cref{loops}) laying out how the concepts generalize for closed boundary conditions, which is our actual experimental ensemble. Then, we will have a worst-to-average case reduction. This reduction is agnostic to architecture. 

To this end, define odd and even layers of $CZ$ gates as follows:

\[
    \CZe
    =
    \CZ_{1,2}\CZ_{3,4}\cdots \CZ_{n-1,n},
\]
and
\[
    \CZo
    =
    \CZ_{2,3}\CZ_{4,5}\cdots \CZ_{n-2,n-1}.
\]

The circuit has the gate sequence:
\begin{equation}
\label{sequence1}
    H^{\otimes n},\quad \CZe,\quad \mathcal{R}_Z,\quad H^{\otimes n},\quad \CZo,\quad \mathcal{R}_Z,\quad H^{\otimes n},\quad \CZe,\quad \mathcal{R}_Z,\quad \cdots
\end{equation}
\noindent where $\mathcal{R}_Z$ denotes a layer of single qubit rotation gates.  We will prove the following theorem.

\begin{theorem}
\label{thm: hardness_theorem}
For programmable rotation angles and polynomial depth, the described circuit family is universal.
\end{theorem}

\subsection{Showing universality: Proof of Theorem \ref{thm: hardness_theorem}}
\noindent The proof will consist of two steps.
\begin{itemize}
\item First, we show that if we turn off all the rotation gates, then we get back the identity after sufficiently many repetitions of the Hadamard and odd/even $\CZ$ layers. This follows from elementary Pauli conjugation relations.
\item Then, we will show that if we turn on one or two rotations at specific places of the circuit, we can recover a universal set of gates.
\end{itemize}
\subsubsection{Getting back the identity}
We will use the following facts.
\[
    H X H = Z,
    \qquad
    H Z H = X,
\]
and, for a controlled-$Z$ on qubits $a$ and $b$,
\begin{equation}
\label{rel1}
    \CZ_{a,b}X_a\CZ_{a,b}=X_aZ_b,
    \qquad
    \CZ_{a,b}X_b\CZ_{a,b}=Z_aX_b,
\end{equation}
while
\begin{equation}
\label{rel2}
    \CZ_{a,b}Z_a\CZ_{a,b}=Z_a,
    \qquad
    \CZ_{a,b}Z_b\CZ_{a,b}=Z_b.
\end{equation}
So, a $Z$ passes through $\CZ$ unchanged, but an $X$ picks up a neighboring $Z$. Let:
\begin{equation}
\label{defG}
G = \CZ_o H^{\otimes n} \CZ_eH^{\otimes n}.
\end{equation}
\noindent Let $X_r$ or $Z_r$ be the Pauli $X$ or $Z$ operator on qubit $1 \leq r \leq n$ respectively. Using \eqref{rel1}, \eqref{rel2}, and $\eqref{defG}$, we will track the evolution of these two operators under successive applications of $G$. For odd \(r\),
\begin{equation}
\label{chain1}
X_r
\xrightarrow{\ G\ }
Z_{r-1}X_r
\xrightarrow{\ G\ }
Z_{r-3}X_{r-2}X_r
\qquad\qquad
Z_r
\xrightarrow{\ G\ }
Z_rX_{r+1}Z_{r+2}
\xrightarrow{\ G\ }
Z_rZ_{r+2}X_{r+3}Z_{r+4}.
\end{equation}

\noindent For even \(r\),
\begin{equation}
\label{chain2}
X_r
\xrightarrow{\ G\ }
X_rZ_{r+1}
\xrightarrow{\ G\ }
X_rX_{r+2}Z_{r+3}
\qquad\qquad
Z_r
\xrightarrow{\ G\ }
Z_{r-2}X_{r-1}Z_r
\xrightarrow{\ G\ }
Z_{r-4}X_{r-3}Z_{r-2}Z_r.
\end{equation}

\noindent Observe that once a Pauli ``falls off the chain,'' i.e., the index $r$ is either less than $1$ or greater than $n$, then it is equivalent to applying the identity on all qubits. 

From \eqref{chain1} and \eqref{chain2}, it follows that if we start with any $n$-qubit Pauli operator, then after some $q =\mathcal{O}(n)$ applications of $G$, it is mapped onto itself, as all the extra terms ``fall off the chain.'' Thus, there exists a $q = \mathcal{O}(n)$ such that
\begin{equation}
\label{G_identity}
G^q = \mathbb{I}_n,
\end{equation}
\noindent where $\mathbb{I}_n$ is the identity operator on $n$ qubits.
\subsubsection{Creating new gates}
\noindent Fix $q$ as in \eqref{G_identity}. Now, we will use \eqref{G_identity} to create new gates.  First define ${R_Z}_j(\alpha)$ as applying $e^{i \alpha Z}$ on qubit $j$. Similarly, $R_{Z_i X_j}$ is defined as the two qubit gate $e^{i \alpha Z_i X_j}$.
\paragraph{Getting ${R_Z}_j(\alpha)$ rotation.} If we turn on one ${R_Z}_j(\alpha)$ rotation gate in the very last layer, after the application of $G^q$, we get back the ${R_Z}_j(\alpha)$ rotation, as 
\[
{R_Z}_{j} (\alpha)G^q = {R_Z}_{j} (\alpha).
\]
\paragraph{Getting ${R_X}_j(\alpha)$ rotation.} If we turn on one ${R_Z}_j(\alpha)$ rotation after the first layer of Hadamards, we get an ${R_X}_j(\alpha)$ rotation. This is by
using
\[
R_{Z_j}(\alpha)H^{\otimes n}
=
H^{\otimes n}R_{X_j}(\alpha).
\]
Concretely, we obtain
\[
\begin{aligned}
G^{q-1}CZ_{\mathrm o}H^{\otimes n}CZ_{\mathrm e}
R_{Z_j}(\alpha)H^{\otimes n}
&=
G^{q-1}CZ_{\mathrm o}H^{\otimes n}CZ_{\mathrm e}
H^{\otimes n}R_{X_j}(\alpha)\\
&=
G^qR_{X_j}(\alpha)\\
&=
R_{X_j}(\alpha).
\end{aligned}
\]

\paragraph{Getting \(R_{Z_iX_j}(\alpha)\) and \(R_{X_iZ_j}(\alpha)\).}
Let \((i,j)\) be an even edge, so that \(CZ_{i,j}\) is contained in
\(CZ_{\mathrm e}\). Turn on one \(R_{Z_i}(\alpha)\) rotation after the
second layer of Hadamards, immediately before \(CZ_{\mathrm o}\). This gives an \(R_{Z_iX_j}(\alpha)\) rotation because
\[
\begin{aligned}
R_{Z_i}(\alpha)H^{\otimes n}CZ_{\mathrm e}H^{\otimes n}
&=
H^{\otimes n}R_{X_i}(\alpha)CZ_{\mathrm e}H^{\otimes n}\\
&=
H^{\otimes n}CZ_{\mathrm e}R_{X_iZ_j}(\alpha)H^{\otimes n}\\
&=
H^{\otimes n}CZ_{\mathrm e}H^{\otimes n}
R_{Z_iX_j}(\alpha).
\end{aligned}
\]
Here we used
\[
R_{Z_i}(\alpha)H^{\otimes n}
=
H^{\otimes n}R_{X_i}(\alpha),
\qquad
R_{X_i}(\alpha)CZ_{\mathrm e}
=
CZ_{\mathrm e}R_{X_iZ_j}(\alpha),
\]
and
\[
R_{X_iZ_j}(\alpha)H^{\otimes n}
=
H^{\otimes n}R_{Z_iX_j}(\alpha).
\]
\noindent Similarly, when $(i, j)$ is an odd edge, turning on one \(R_{Z_i}(\alpha)\) rotation immediately
before the final layer of Hadamards and \(CZ_{\mathrm o}\) gives an \(R_{X_iZ_j}(\alpha)\) rotation.

\noindent \paragraph{Showing universality.} Since the angles are freely programmable, the rotations
\(R_X(\alpha)\) and \(R_Z(\alpha)\) generate arbitrary single-qubit
unitaries, up to an irrelevant global phase, by the Euler-angle decomposition
\[
U=e^{i\delta}R_Z(\alpha)R_X(\beta)R_Z(\gamma).
\]
Additionally, we have two-qubit entangling gates on every pair of nearest neighbor qubits. Thus, by standard results like \cite{bremner2002practical}, this gate-set is universal, proving Theorem \ref{thm: hardness_theorem}.

\subsubsection{Generalization to closed loops} 
\label{loops}
Note that the proofs of this section also generalize to closed loops, where the qubit labels are taken cyclically. To close the loop, the extra gate is a $\operatorname{CZ}$ gate between qubits $n$ and $1$. 

Although the Paulis now ``wrap around'' instead of falling off the chain, repeated factors still cancel when they appear twice. Once we account for this periodicity using the same Pauli relations, we can still show that after sufficiently many layers, the Clifford skeleton gives back the identity gate. More concretely, if $n$ is a power of $2$, we can show that $G^{3n} \propto \mathbb{I}_n$. The rest of the steps are unchanged.

\subsection{Worst-case hardness of computing output probabilities and sampling}

\noindent The following lemma holds.

\begin{lemma}
Computing the output probabilities of the ensemble under consideration is $\#\text{P}$-hard. Moreover, an exact classical sampler from the output distribution collapses the polynomial hierarchy.
\end{lemma}
\begin{proof}
Using Theorem \ref{thm: hardness_theorem}, the ensemble is universal. Hence, it can implement circuits whose output probabilities are $\#\text{P}$-hard to compute, like Fourier sampling circuits \cite{fefferman2015power} . 

Moreover, the ensemble remains universal under post-selection and is equal to the complexity class $\text{PostBQP}$, which in turn is equal to $\text{PP}$ \cite{aaronson2004quantumcomputingpostselectionprobabilistic}. Then, if a sampler exists, the polynomial hierarchy collapses using standard reductions like the ones in \cite{Bremner_2010}.
\end{proof}

\noindent So, assuming the polynomial hierarchy does not collapse, a classical sampler from the worst-case output distribution does not exist.

\begin{remark}
\label{remark: additive error}
Using works like \cite{aaronson2010computational, fefferman2015power}, it holds that computing worst-case output probabilities to an additive error of $2^{-\text{poly}(n)}$ is also $\#\text{P}$-hard, for an appropriate choice of polynomial.  This observation will be used in our worst-to-average case reduction.
\end{remark}

\subsection{Worst-case to average-case reduction}
\label{sec:taylor-worst-to-average}

\noindent We will show a worst-to-average case reduction for computing output probabilities, using Remark \ref{remark: additive error}. The proof will involve a straightforward application of techniques in \cite{bouland2019complexity}. 

Let $(i, t)$ be the coordinates for the $i^{\text{th}}$ gate in the $t^{\text{th}}$ rotation layer. Define 
   the corresponding rotation gate as
   \[
   \mathcal{R}_{Z,i,t}(\alpha)
        := e^{i\alpha Z_i}.
\]
\noindent Let $M = \text{poly}(n)$ be the total number of rotation gates in the circuit. For an angle array
$\boldsymbol{\alpha}$, write
\[
    p_x(\boldsymbol{\alpha})
       :=
       \left|
       \langle x|U(\boldsymbol{\alpha})|0^n\rangle
       \right|^2.
\]
\noindent Let
$\boldsymbol{\alpha}^{\star}
       =\{\alpha_{i,t}^{\star}\}_{i,t}$ be the list of worst-case hard angles. For each programmable location $(i,t)$, independently sample $\beta_{i,t}\sim\operatorname{Unif}[0,2\pi)$ and let $h_{i,t}:=\beta_{i,t}Z_i$. For $\theta\in[0,1]$, define 
\begin{align}
\label{eq: truncated}
    \mathcal{R}_{Z,i,t}^{(\boldsymbol{\beta})}(\theta)
       &:=
       \mathcal{R}_{Z,i,t}(\alpha_{i,t}^{\star})
       \mathcal{R}_{Z,i,t}(\beta_{i,t})
       e^{-i\theta h_{i,t}}                                      \\
       &=
       \mathcal{R}_{Z,i,t}\!\left(
          \alpha_{i,t}^{\star}
          +(1-\theta)\beta_{i,t}
       \right).
    \label{eq:exact-taylor-path}
\end{align}
\noindent The angle additions are modulo $2 \pi$.
\paragraph{Recovering the original ensemble at $\theta = 0$ and the worst-case circuit at $\theta = 1$.} Observe that
\[
    \mathcal{R}_{Z,i,t}^{(\boldsymbol{\beta})}(0)
       =
       \mathcal{R}_{Z,i,t}\!\left(
          \alpha_{i,t}^{\star}+\beta_{i,t}
       \right),
\]
which is uniformly random modulo $2\pi$, whereas
\[
    \mathcal{R}_{Z,i,t}^{(\boldsymbol{\beta})}(1)
       =
       \mathcal{R}_{Z,i,t}(\alpha_{i,t}^{\star}),
\]
\noindent which recovers the worst-case hard instance.

\paragraph{Closeness between intermediate distributions and true distribution.} In the next sections, we will pick points close to $\theta = 0$ and define intermediate distributions to help us in interpolating the worst-case hard value. Let $\mathcal D_n$ denote the original average ensemble in which all $M$ 
 angles are chosen independently and uniformly modulo
$2\pi$, and let $\mathcal D_{\boldsymbol{\alpha}^{\star},\theta}$ denote the
distribution of the exact interpolating circuit
$U_{\boldsymbol{\beta}}(\theta)$. The angle
\[
    \alpha_{i,t}^{\star}
      +(1-\theta)\beta_{i,t}
      \pmod{2\pi}
\]
is uniform on an arc occupying a $1-\theta$ fraction of the circle.
Its total variation distance from the uniform distribution on the
entire circle is therefore $\theta$.  Since the $M$ angles are
independent, 
\begin{equation}
\label{theta}
    d_{\mathrm{TV}}\!\left(
       \mathcal D_{\boldsymbol{\alpha}^{\star},\theta},
       \mathcal D_n
    \right)
    \leq M\theta,
\end{equation}
where we have used a data processing inequality to show that constructing the circuit from information about its rotation angles cannot increase total variation distance. Since $M$ is polynomially large, we can choose $\theta$ to be inverse polynomially small such that Equation \ref{theta} is also inverse polynomially small.

\paragraph{Rounding the continuous distribution to finite values.} Without loss of too much generality, we will assume that the angles come from a continuous distribution. In reality, we will have to discretize the angles. However, one can show that the error from rounding scales roughly as $\sim 2^{-b}$, where b is the number of bits used to represent the angles. Thus, by picking $b$ appropriately, this error can be made too small to matter in computing probabilities.

\subsubsection{Truncating the Taylor series} 
\noindent We will now truncate the Taylor series to represent output probabilities as a low-degree polynomial. Define
\[
    T_K(A):=\sum_{k=0}^{K}\frac{A^k}{k!},
\]
and set
\begin{equation}
    \widetilde{\mathcal{R}}_{Z,i,t}^{
       (\boldsymbol{\beta},K)}(\theta)
       :=
       \mathcal{R}_{Z,i,t}(\alpha_{i,t}^{\star})
       \mathcal{R}_{Z,i,t}(\beta_{i,t})
       T_K(-i\theta h_{i,t}).
    \label{eq:truncated-taylor-rotation}
\end{equation}

\noindent This is a truncated expression for the rotation angle defined in Equation \ref{eq: truncated}. Observe that
\[
    \|h_{i,t}\|
       =
       |\beta_{i,t}|
       \leq 2\pi.
\]
Hence,
\begin{align}
\label{eq_phase}
    \left\|
       e^{-i\theta h_{i,t}}
       -
       T_K(-i\theta h_{i,t})
    \right\|
    &\leq
       \sum_{k=K+1}^{\infty}\frac{(2\pi)^k}{k!}                  \\
    &\leq
       e^{2\pi}\frac{(2\pi)^{K+1}}{(K+1)!}
       =:\rho_K.
       \label{rhok}
\end{align}

Furthermore, if $\widetilde U_{\boldsymbol{\beta},K}(\theta)$ is the corresponding circuit expression, with the original rotation gate replaced by the truncated ones, then \begin{equation}
\label{eq: degree}
    q_{\boldsymbol{\beta},K}(\theta)
       :=
       \left|
       \langle x|
       \widetilde U_{\boldsymbol{\beta},K}(\theta)
       |0^n\rangle
       \right|^2.
\end{equation}
\noindent Since there are $M$ rotation gates, and we have truncated the Taylor series at $K$, it follows that the output probability in Equation \ref{eq: degree} is a polynomial with degree at most $2MK.$

\paragraph{Relating exact and truncated output probabilities using Feynman paths.}
After absorbing the fixed Clifford subcircuits into the neighboring
programmable locations, let us write
\[
    U_{\boldsymbol{\beta}}(\theta)
       =G_M(\theta)\cdots G_1(\theta),
    \qquad
    \widetilde U_{\boldsymbol{\beta},K}(\theta)
       =
       \widetilde G_M(\theta)\cdots
       \widetilde G_1(\theta),
\]
where $G_j(\theta)$ and $\widetilde G_j(\theta)$ contain,
respectively, the exact and Taylor-truncated coherent $Z$-rotation
at location $j$.  Using Equation \eqref{eq_phase} and Equation \eqref{rhok},
\[
    \|G_j(\theta)-\widetilde G_j(\theta)\|
       \leq \rho_K,
    \qquad
    \|G_j(\theta)\|=1,
    \qquad
    \|\widetilde G_j(\theta)\|\leq1+\rho_K.
\]

\noindent Set $y_0=0^n$ and $y_M=x$.  The Feynman path expansion and the
telescoping identity for products give
\begin{align}
&\left|
 \langle x|\widetilde U_{\boldsymbol{\beta},K}(\theta)|0^n\rangle
 -
 \langle x|U_{\boldsymbol{\beta}}(\theta)|0^n\rangle
\right|                                                        \nonumber\\
&\quad\leq
\sum_{y_1,\ldots,y_{M-1}\in\{0,1\}^n}
\left|
 \prod_{j=1}^{M}
 \langle y_j|\widetilde G_j(\theta)|y_{j-1}\rangle
 -
 \prod_{j=1}^{M}
 \langle y_j|G_j(\theta)|y_{j-1}\rangle
\right|                                                        \nonumber\\
&\quad\leq
2^{n(M-1)}
M\rho_K(1+\rho_K)^{M-1}
=: \varepsilon_K.
\label{eq:feynman-truncation-bound}
\end{align}
Since $M=\operatorname{poly}(n)$, Stirling's formula allows
$K=\operatorname{poly}(n)$ to be chosen sufficiently large that
\[
    \varepsilon_K=2^{-\operatorname{poly}(n)}.
\]

\noindent Let
\[
    A(\theta)
       :=
       \langle x|U_{\boldsymbol{\beta}}(\theta)|0^n\rangle,
    \qquad
    \widetilde A_K(\theta)
       :=
       \langle x|
       \widetilde U_{\boldsymbol{\beta},K}(\theta)
       |0^n\rangle.
\]
Since $|A(\theta)|\leq1$ and
$|\widetilde A_K(\theta)|\leq1+\varepsilon_K$, we obtain
\begin{align}
\left|
 q_{\boldsymbol{\beta},K}(\theta)
 -
 \left|
   \langle x|U_{\boldsymbol{\beta}}(\theta)|0^n\rangle
 \right|^2
\right|
&=
\left|
 |\widetilde A_K(\theta)|^2-|A(\theta)|^2
\right|                                                        \\
&\leq
\bigl(
 |\widetilde A_K(\theta)|+|A(\theta)|
\bigr)
|\widetilde A_K(\theta)-A(\theta)|                             \\
&\leq
(2+\varepsilon_K)\varepsilon_K
=
2^{-\operatorname{poly}(n)}.
\label{eq:truncated-genuine-closeness}
\end{align}

\paragraph{Choosing interpolation points.} Now, choose $\Delta = 1/\text{poly}(n)$ and let 
\begin{equation}
\label{eq: points}
L:=6d+1,
\end{equation}
for $d := 2MK$ is an appropriately chosen polynomial in $n$. Let
\[
    \Theta=\{\theta_0,\ldots,\theta_{L-1}\}
       \subseteq[0,\Delta]
\]
be any $L$ distinct polynomial-bit rational points. These will be our points for polynomial interpolation. 

Define
$\widetilde{\mathcal D}_{
\boldsymbol{\alpha}^{\star},K,\Theta}$ by sampling
$\boldsymbol{\beta}$ as above, choosing
$\ell\in\{0,\ldots,L-1\}$ uniformly, and outputting the description
of
$\widetilde U_{\boldsymbol{\beta},K}(\theta_\ell)$.

\subsubsection{Average case hardness theorem}
\noindent We are now ready to state our main theorem.

\begin{theorem}
\label{thm:taylor-truncated-average-hardness}
There are choices
\[
    K=\operatorname{poly}(n),
    \qquad
    \Delta=\frac{1}{\operatorname{poly}(n)},
\]
such that exactly computing
$q_{\boldsymbol{\beta},K}(\theta_\ell)$ on at least an $8/9$
fraction of the instances drawn from
$\widetilde{\mathcal D}_{
\boldsymbol{\alpha}^{\star},K,\Theta}$
is $\#\mathrm{P}$-hard under randomized polynomial-time reductions.
\end{theorem}

\begin{proof}
Suppose that a polynomial-time algorithm $\mathcal O$ exactly
computes
$q_{\boldsymbol{\beta},K}(\theta_\ell)$ with probability at least
$8/9$ over
$\widetilde{\mathcal D}_{
\boldsymbol{\alpha}^{\star},K,\Theta}$. Let $r$ be the internal randomness of $\mathcal{O}$. For fixed
$(\boldsymbol{\beta},r)$, define
\[
    g(\boldsymbol{\beta},r)
       :=
       \frac{1}{L}
       \left|
       \left\{
          \ell:
          \mathcal O_r\!\left(
             \widetilde U_{\boldsymbol{\beta},K}(\theta_\ell)
          \right)
          =
          q_{\boldsymbol{\beta},K}(\theta_\ell)
       \right\}
       \right|.
\]
Since the oracle succeeds with probability at least $8/9$ over $(\boldsymbol{\beta},r,\ell)$,
\[
    \mathbb E_{\boldsymbol{\beta},r}
       \bigl[1-g(\boldsymbol{\beta},r)\bigr]
       \leq \frac19.
\]
Therefore, by Markov's inequality,
\[
    \Pr_{\boldsymbol{\beta},r}
       \left[
          g(\boldsymbol{\beta},r)<\frac23
       \right]
    =
    \Pr_{\boldsymbol{\beta},r}
       \left[
          1-g(\boldsymbol{\beta},r)>\frac13
       \right]
    \leq
    \frac{1/9}{1/3}
    =
    \frac13.
\]
Thus, with probability at least $2/3$ over
$(\boldsymbol{\beta},r)$, the oracle is correct on at least
$2L/3$ interpolation points. Fix such a good pair $(\boldsymbol{\beta},r)$.  Since
$q_{\boldsymbol{\beta},K}$ has degree at most $d$ and $L \geq 6d$ (according to Equation \ref{eq: degree}),
\[
    \frac{2L}{3}
       >
       \frac{L+d}{2}.
\]
Thus, this is a sufficient condition for the Berlekamp--Welch algorithm to reconstruct
$q_{\boldsymbol{\beta},K}(\theta)$ exactly from the $L$ interpolation points. 

\paragraph{Relating to worst-case output probability.} At the worst-case endpoint,
\[
    U_{\boldsymbol{\beta}}(1)
       =
       U(\boldsymbol{\alpha}^{\star}).
\]
Hence
\begin{equation}
    \left|
       q_{\boldsymbol{\beta},K}(1)
       -
       p_x(\boldsymbol{\alpha}^{\star})
    \right|
    \leq 2^{-\text{poly}(n)},
    \label{eq:worst-case-endpoint-bound}
\end{equation}
\noindent according to Equation \ref{eq:truncated-genuine-closeness}. According to Remark \ref{remark: additive error}, this would solve a $\#\mathrm{P}$-hard problem. This proves the theorem. 
\end{proof}

\paragraph{A note on the hardness of sampling.} Note that the same techniques as those used in the proof of Theorem \ref{thm:taylor-truncated-average-hardness}, can be used in proving $\#\mathrm{P}$-hardness of computing probabilities up to $2^{-\text{poly}(n)}$ additive error. If we want robustness up to a better additive error, techniques from subsequent follow-ups to \cite{bouland2019complexity}, like \cite{Movassagh2023, Bouland_2022, bouland2025exponential}, are potentially useful. However, if we want to relate the hardness of computing output probabilities to the hardness of classically sampling from the output distribution up to $1/\text{poly}(n)$ total variation distance error (modulo just the polynomial hierarchy not collapsing), the current techniques require us to prove hardness up to an additive error of $2^{-n}/\text{poly}(n)$. This remains an open question, both for random circuit sampling and for the task under our consideration. 

In other words, at present, we can say that assuming \emph{two} conjectures---that computing output probabilities is $\#\mathrm{P}$-hard up to $2^{-n}/\text{poly}(n)$ additive error and that the polynomial hierarchy does not collapse---sampling from our ensemble, up to $1/\text{poly}(n)$ total variation distance in the average case over the random rotations, is classically hard.


\section{Simulation complexity at the experimental scale} \label{sec:simulation_hardness}

In order to assess the simulation hardness of our circuits, we cover here different state-of-the-art simulation methods and their limitations. Typical simulation algorithms either rely on low entanglement or low magic (relative to some non-universal fragment), or even both, with  some recent approaches mixing the two. 

Throughout this section, the typical classical simulation task we consider is the computation of outcome probabilities. For sufficiently high-fidelity circuits, this provides a reasonable alternative to sampling. Furthermore, this is more favorable to classical methods because sampling often involves many calls to amplitude estimation (i.e. computing probabilities).

\subsection{Tensor-network based simulations}

Quantum circuits inducing only low entanglement (such as shallow circuits on limited connectivity) are known to be classically simulable. Therefore, it is important to ensure that our circuits present high entanglement, as this is a necessary condition for simulation hardness.

Standard simulation tools for low entanglement quantum circuits are:
\begin{itemize}
    \item MPS/Tensor trains \cite{Vidal_2003, SCHOLLWOCK201196}. In this approach, the running quantum state is stored into a linear tensor network, effectively compressing it when it admits a factoring for the tensor product. This technique can exploit low bipartite entanglement across cuts of the system along some linear arrangement of the qubits.
    \item General tensor network contraction \cite{Markov_2008, Gray_2021}. A more general approach where the circuit is converted into a general tensor network and some heuristic contraction algorithm tries to find a good contraction ordering. 
\end{itemize}
Overall, these approaches all rely on finding low entropy bipartitions of the system, allowing the simulator to recursively decompose the state into a sum of products of simpler states. Some recent works demonstrated how some particular graph parameter, namely the rank-width, directly ties to contraction complexity, thus providing tighter lower bounds on contraction complexity and practical heuristics for tensor network contraction \cite{cam2023speedingquantumcircuitssimulation, kuyanov2026efficientclassicalsimulationlowrankwidth}.
In general, one can define a parameter, called entanglement-width, that captures how much entanglement the state contains \cite{PhysRevLett.97.150504}.

In our setting, we rely on a Clifford circuit that then gets doped into a non-Clifford regime. For stabilizer states, the entanglement-width happens to exactly match the rank-width of any graph state underlying it. Ideally, one should be able to compute the rank-width of a graph state underlying the stabilizer state used in our experiment to guarantee some lower bound on the contraction time of the corresponding tensor network. This, however, is doomed to fail due to the high complexity cost of computing the rank width of graphs. As an alternative, we rely on the following procedure:
\begin{itemize}
    \item We first compute the adjacency matrix $A$ of some graph state underlying our stabilizer state
    \item We then pick a very large number of random bipartitions of the set of qubits and compute the rank of the bipartite adjacency matrix induced by that cut (corresponding to an off-diagonal block of $A$)
    \item For each bipartite adjacency matrix, we compute its rank over $\mathbb{F}_2$ and record the minimum rank observed.
\end{itemize}
Even though this approach can only give an upper bound on the entanglement-width of our stabilizer state, we claim that this is sufficient to assess the hardness of simulation via tensor network contraction, since any attacker would need to also find low-rank/entropy bipartitions of the system.

In practice, for our 64-qubit experiment, we explored $10^8$ random bipartitions and observed a minimal rank of $27$. If this is the minimal rank, it entails that a tensor network based simulator exploiting this particular minimal bipartition would still need to perform at least $2^{27}$ operations, if not more.

\begin{figure*}[!h]
    \centering
    \includegraphics[width=0.48\linewidth]{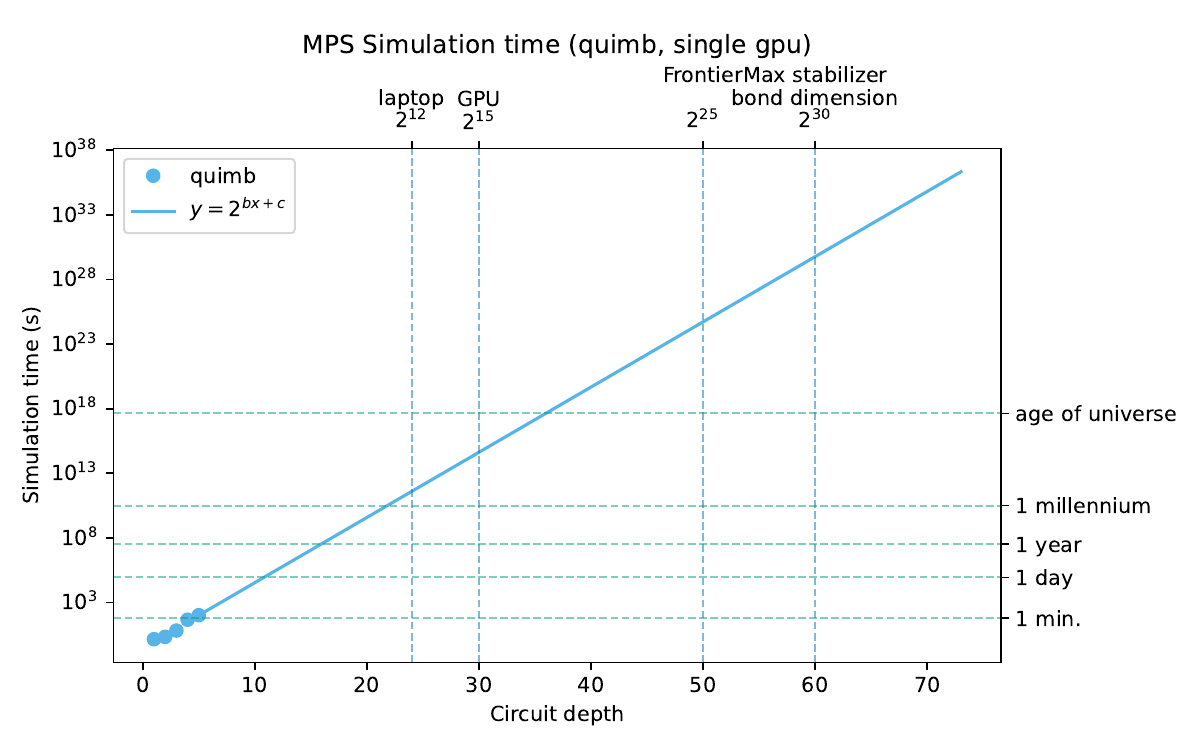}
    \hfill
    \includegraphics[width=0.48\linewidth]{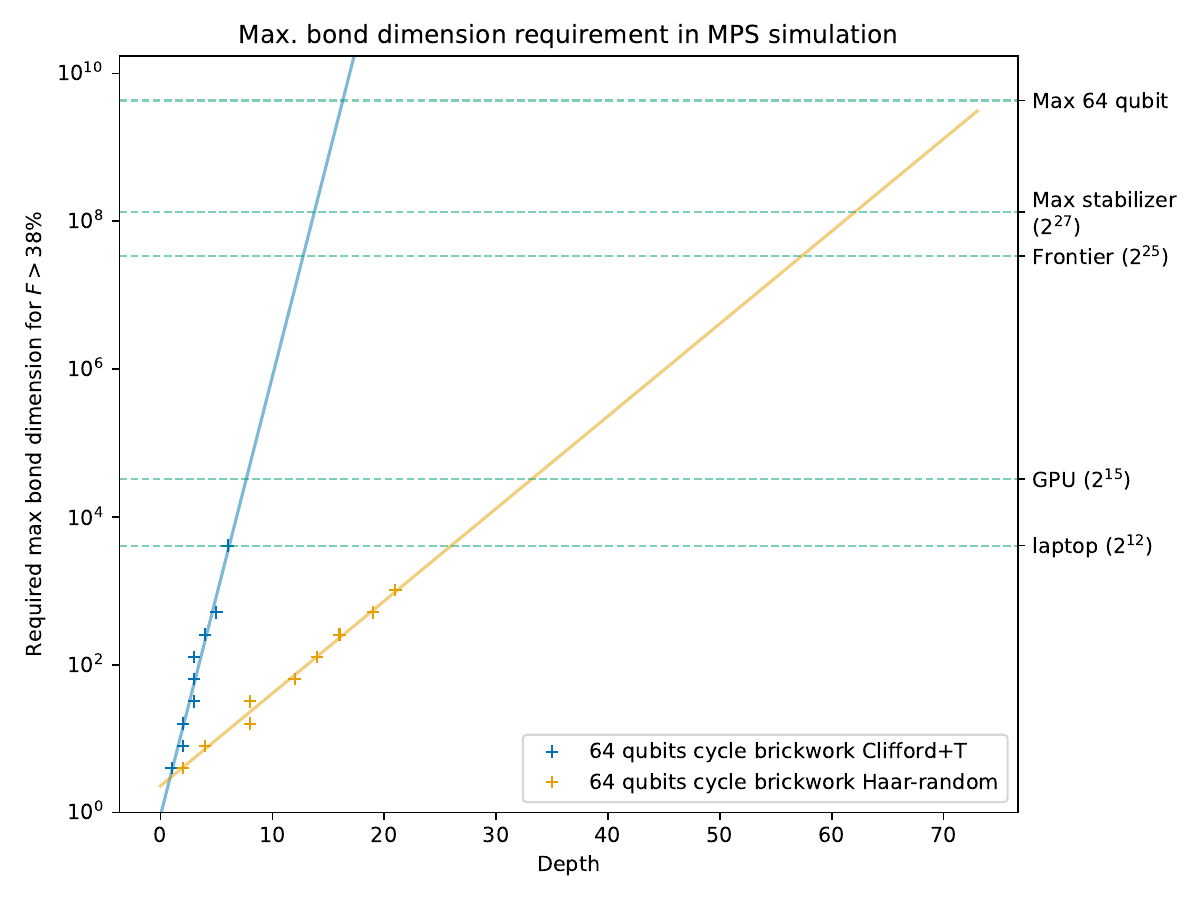}

    \caption{\textbf{MPS simulation scaling and bond-dimension requirements.}
    Left: We simulated our experiment via MPS, tracking the total simulation time at each end of a layer. As expected, due to the bond dimension requirements, the simulation time scales exponentially.
     Dots represent measured simulation times, while the solid curve is a fit to $y=2^{bx+c}$, with $b=1.36$ and $c=-0.94$.
    Right: Each point is obtained by running MPS with a fixed maximum bond dimension and recording the depth at which the truncation error exceeds $1-0.38$. We compare these depths for our $T$-doped circuit and a similar brickwork circuit in which CZ layers are interleaved with layers of Haar-random single-qubit gates. The scaling indicates that our circuits are significantly harder to simulate than their Haar-random counterparts (see Sec.~\ref{paragraph:haar_vs_clifford}).}
    \label{fig:mps}
\end{figure*}

To illustrate this high entanglement entropy scaling, we simulated our circuit with an MPS approach using {\bf quimb}. We ran the MPS algorithm on an H100 GPU (with 80 GB memory) for increasing maximum bond dimension on our circuit and on a Haar-random loop brickwork circuit of the same width. We track the truncation error induced by this maximum bond dimension and stop the simulation when the error goes beyond $1 - 0.38$ (i.e., that the MPS will fail to achieve the fidelity achieved in our experiment). Results are reported in Figure \ref{fig:mps}.

\begin{figure}[!h]
    \centering
    \includegraphics[width=0.95\linewidth]{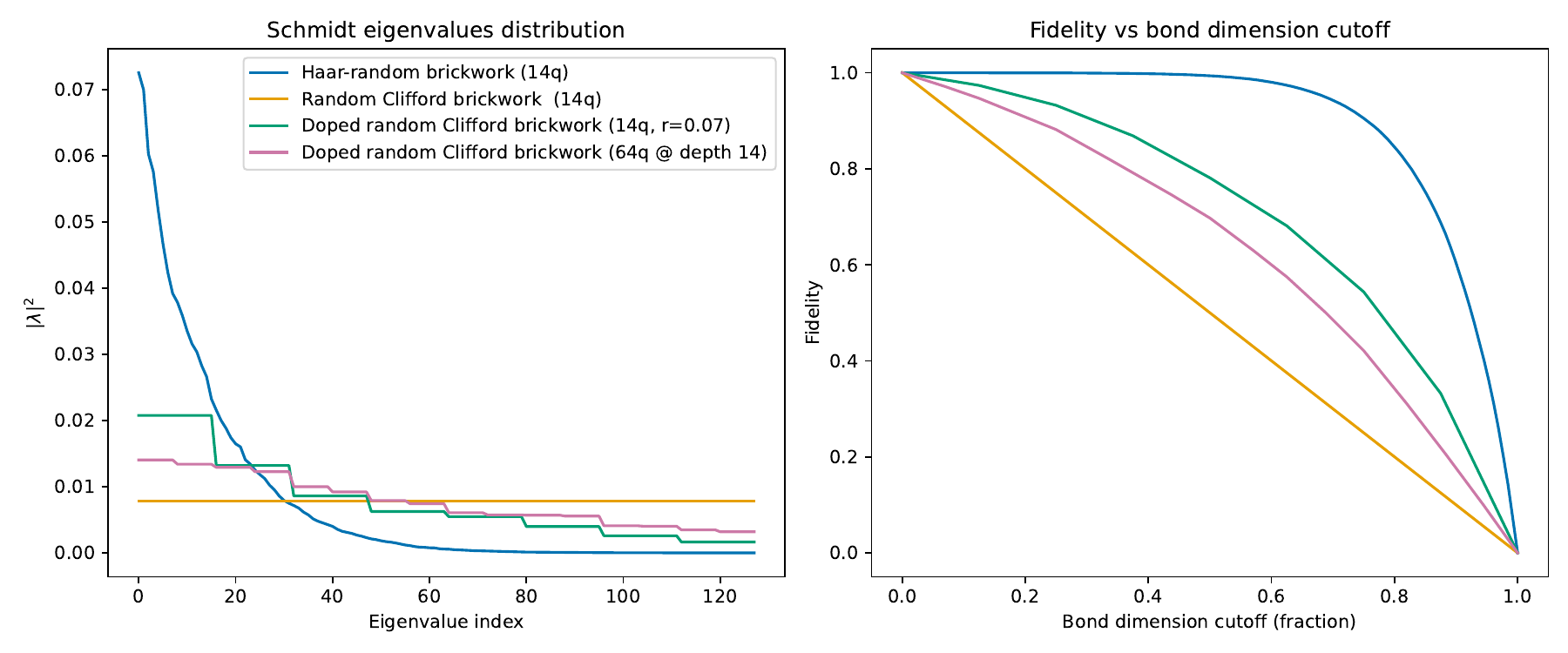}
    \caption{{\bf Singular-value distributions for various circuits.} 
    We considered four different circuits: One Haar-random circuit over 14 qubits (depth 14), a random Clifford circuit over 14 qubits (depth 14), the same random brickwork Clifford circuit, doped at a rate of $r=0.07$ which matches our experiment's doping rate, all of them without periodic boundary conditions, and finally a doped Clifford circuit very similar to our experiment circuit except with no boundary conditions and up until depth 14 at which a maximum bond dimension of $2^7$ fails to faithfully represent the current state.
    Left: Depicts the singular-value distributions for each of those circuits. We can observe that the distribution obtained in our circuit is actually flatter than that of a random doped Clifford circuit.  Right: Depicts the fidelity drop as a function of the truncation aggressiveness. We can see the linear fidelity drop in the Clifford setting. Even though that drop is milder in the doped cases, it is still far from negligible when compared to the Haar-random case.}
    \label{fig:mpsspectrum}
\end{figure}

\newpage

\paragraph{A note on Haar-random vs. Clifford vs. doped Clifford circuits and MPS.}\label{paragraph:haar_vs_clifford}
The ability for a tensor-train/MPS to faithfully represent a given quantum state heavily relies on the structure of the Schmidt rank and the corresponding singular-value distributions. Even though MPS will generally struggle to approximate states with high Schmidt rank to arbitrarily high precision, they can be particularly successful at approximating states with exponentially decaying singular-value distributions \cite{PhysRevX.10.041038}, which includes brickwork circuits with $CZ$ gates interleaved with Haar-random single qubit gates. For maximally entangled stabilizer states however, the distribution of Schmidt eigenvalues is flat, preventing any efficient truncations. In particular, given a state with Schmidt rank $R$ represented by an MPS with bond dimension $\chi<R$, one can show that: $$ F < \frac{\chi}{R}=\frac{\chi}{2^S}$$ where $S$ is the bipartite entanglement entropy.

For doped circuits, we do expect the distribution of singular values to deviate from a flat/Clifford distribution and slowly converge to some exponentially decaying distribution. Our circuit, however, due to the relatively low amount of doping, exhibits a slowly decaying distribution (see Figure \ref{fig:mpsspectrum}).

\begin{figure}[!t]
    \centering
    \includegraphics[width=\linewidth]{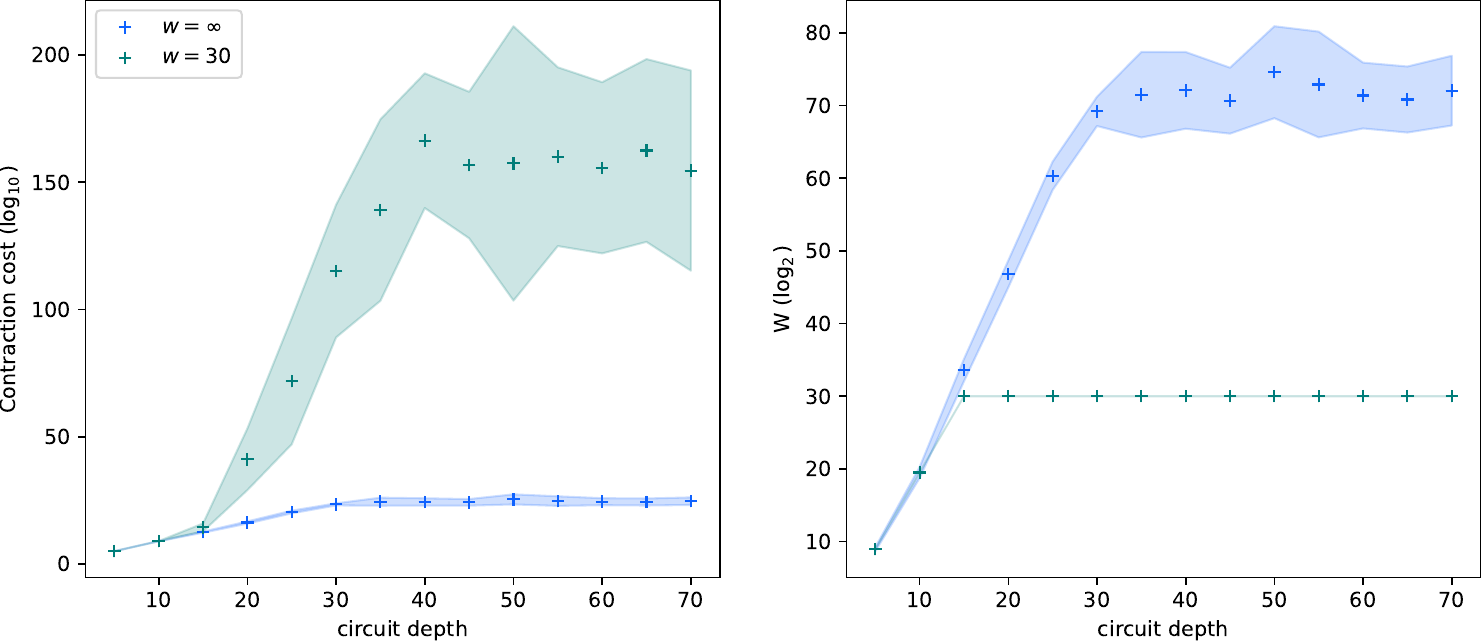}
    \caption{{\bf Tensor-network contraction cost and memory.} 
    We ran {\bf Cotengra}'s hyper optimizer on our circuit truncated at various depths and compared contraction cost (in FLOPs, first panel) and contraction memory (in number of scalars, second panel). In a first run, we let the optimizer find a good contraction order with no memory limitation, leading to a maximum memory of $2^{71}$ scalars ($w=\infty$). In a second run, we limited the maximum memory footprint to $2^{30}$ resulting in higher contraction costs ($w=30$).}
    \label{fig:tn_simulation}
\end{figure}

\newpage

\paragraph{Tensor-network contraction heuristics.} Even though we expect the rank-width and tree-width of our family of circuits to be high, it is important to assess the feasibility of raw tensor network contraction approaches for our 64-qubit realization. To assess this approach, we ran {\bf Cotengra's hyper} optimizer to find a good contraction ordering for our circuit and track the contraction cost and width when increasing depth. We also ran the same experiment, this time bounding contraction width to $2^{30}$, which roughly matches the available memory on a modern GPU. Results are presented in Figure \ref{fig:tn_simulation}.

\subsection{Low stabilizer-rank decomposition}

Another orthogonal approach to quantum circuit simulation is to try to decompose the quantum state of interest into a basis of classically tractable quantum states, usually stabilizer states. This technique, originally introduced in \cite{bravyi2016trading, bravyi2016improved}, presents the advantage of scaling polynomially with the number of qubits, regardless of the amount of entanglement present in the quantum state. The flip side is that the resulting simulation typically scales exponentially with the number of gates lying outside of the tractable gate set (Clifford gates in our case).

This technique has been iteratively improved by reducing the exponential scaling by finding more efficient asymptotic low-rank decompositions \cite{qassim2021improved} or utilizing a tensor-network rewriting strategy to heuristically improve stabilizer decompositions \cite{kissinger_vilmart}.
In practice, we compared different techniques, namely the phase-aware Tableau based implementation of \cite{Bravyi2019simulationofquantum} using the state-of-the-art cat-state-based decomposition of \cite{qassim2021improved}, and the ZX-based heuristic of \cite{kissinger_vilmart}, as well as more recent algorithms~\cite{tsim,clifft}. Overall, the ZX-based approach is by far the most performant heuristic for our circuits. Figure \ref{fig:tsimscaling} summarizes the scaling of this approach as a function of the number of $T$ gates in the simulated circuit.

\begin{figure}[!t]
    \centering
    \includegraphics[width=0.95\linewidth]{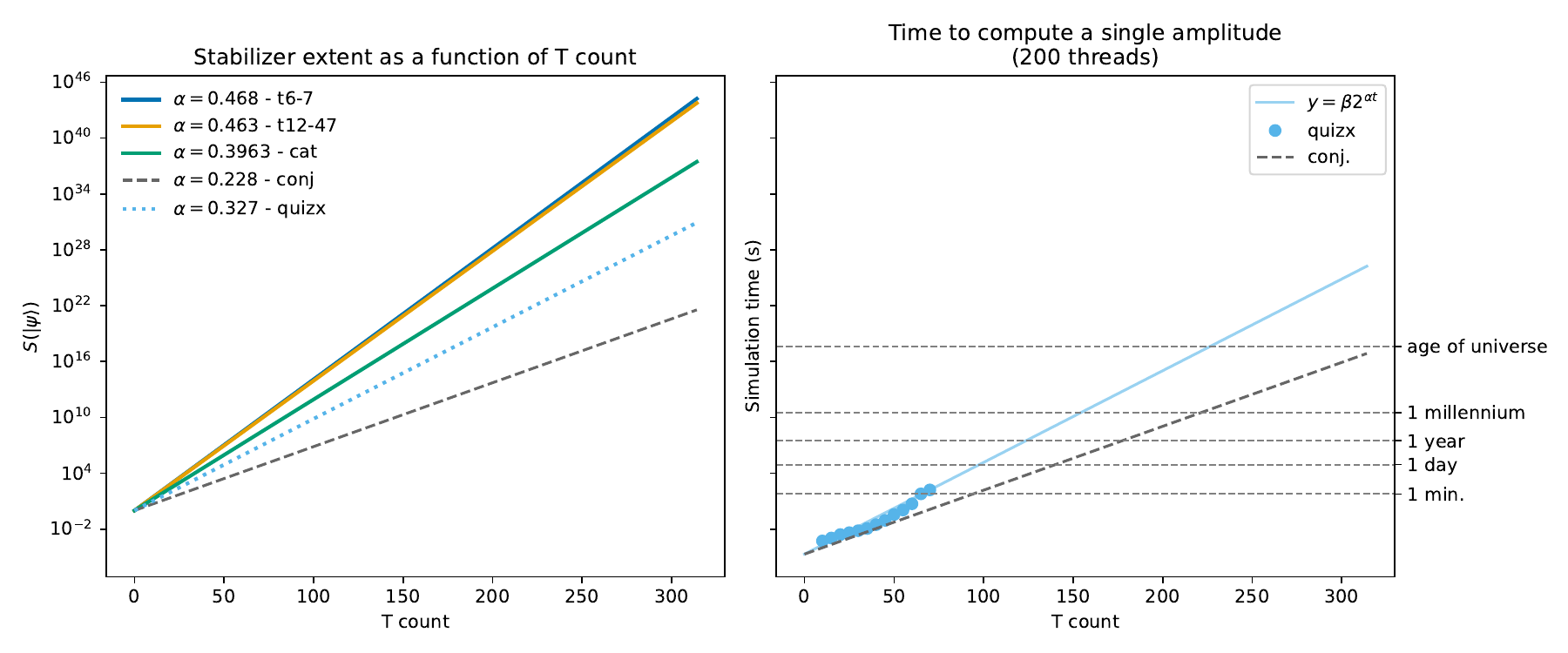}
    \caption{{\bf Near-stabilizer simulation scaling.} (left) Scaling of different known stabilizer decompositions of $\ket{\psi}=\ket{T}^{\otimes t}$. All scalings have form $2^{\alpha t}$ for different $\alpha$. The gray dashed line represents the conjectured lower bound for the stabilizer extent of $\ket{\psi}$. The dashed blue line corresponds to the full blue curve on the right panel. {\bf t6-7} corresponds to \cite{bravyi2016improved} 6 Ts to 7 stabilizers. {\bf t12-47} corresponds to \cite{kocia2022improvedstrongsimulationuniversal} 12 Ts to 47 stabilizers.  {\bf cat} corresponds to \cite{qassim2021improved} 12 Ts to 47 stabilizers. Finally, {\bf quizx} corresponds to the empirical analysis of the performance of the software of the same name that uses dynamical branching and circuit rewriting during the decomposition. (right) Simulation time of our 64-qubit circuit with various amounts of doping using {\bf QuiZX}, up to 70 T gates (dots). In full: fit of $y=\beta 2^{\alpha t}$ on those data points, with $\alpha=0.327$ and $\beta=2.07\cdot 10^{-5}$. The dashed curve corresponds to $y=\beta 2^{0.228x}$: that is using the observed prefactor in {\bf QuiZX} simulations and the conjectured best possible scaling for stabilizer extent.}
    \label{fig:tsimscaling}
\end{figure}

While spacetime checks admit a large number of non-Clifford resources, a potential weakness is that the valid doping sites are no longer uniformly distributed across the circuit. Indeed, when constraining to the locations allowed by spacetime checks, doping sites are sparse in the majority of the layers, but increase in density especially in the last few layers. To quantify how this affects the nonstabilizerness of the resultant state, we measure the 2-stabilizer Réyni entropy (2-SRE) \cite{leone2022stabilizer} exactly for various $20 \times 20$ instances of our ansatz. As seen in Figure~\ref{fig:sresims}, compared to leaving $T$ gate doping unconstrained, constraining by checks does typically result in a lower 2-SRE, also increasing the variance. Nonetheless, there is only a small difference between the largest 2-SRE of 18.00 in the former and the lowest 2-SRE of 17.66 in the latter cases. This indicates that limiting doping to spacetime locations based on checks does not significantly decrease the nonstabilizerness, and our circuits should still be challenging for stabilizer decomposition simulations.

\begin{figure}[!ht]
    \centering
    \includegraphics[width=0.5\linewidth]{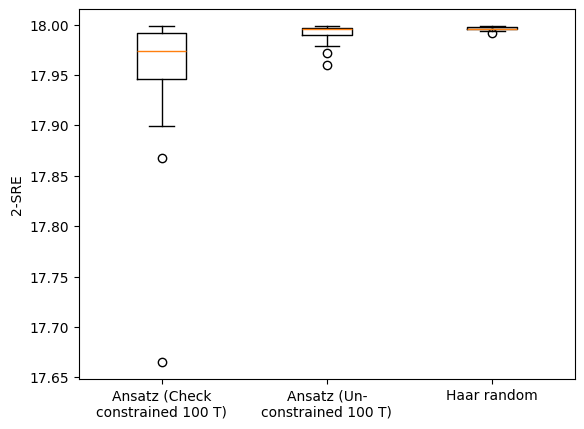}
    \caption{{\bf 2-stabilizer Rényi entropy for various doping schemes.} Using the HadaMAG simulator \cite{sresimulator} we find the exact 2-SRE for 20 instances of our $20 \times 20$ ansatz (with 10 ancillas) doped with 100 $T$ gates, first limiting to sites allowed by spacetime checks and then considering all possible sites, and brickwork Haar-random circuits. As expected, the Haar-random states have the largest 2-SRE, but there is only a small drop-off when using the ansatz with doping in check-constrained sites.}
    \label{fig:sresims}
\end{figure}

\subsection{Hybrid approaches}

Recent work proposes a hybrid approach where a disentangling Clifford frame is computed in which the non-Clifford rotations of the circuit are acting on a small set of qubits, thus generating low entanglement \cite{ybnf-rjw8, clifft}. These approaches work particularly well for circuits where most non-Clifford rotations pairwise commute. In particular, the authors provide a metric for the effectiveness of their simulation technique, directly tied to the anticommutation relations between those rotations. For our circuit, this metric hints at maximal complexity in that all rotations, save for the first $n$, would generate entanglement in the Clifford frame tracked by the simulator.

The metric proposed by \cite{ybnf-rjw8} corresponds to the rank of the kernel of the code generated by the X-components of the back-propagated rotations. In other words, we first compute the back-propagator of each of the $m$ Pauli $Z$ located in place of a $T$ gate in our circuit, compute their symplectic encodings, and record their first $n$ components, giving us a matrix $\mathcal{T}$ in $\mathbb{F}_2^{n \times m}$. We then compute the dimension of the kernel of that matrix. For our circuit with $314$ $T$ gates, we obtain maximal rank of $314-64 = 250$. Consequently, this disentangling approach will successfully disentangle the first $n=64$ $T$ gates and then start building up entanglement for each subsequent $T$ gate. In some sense, this approach will offset the simulation cost of the first $64$ $T$ gates and will then track a MPS with a bond dimension increasing for each subsequent $T$ gate.

To illustrate this point, we ran a similar experiment to the raw MPS above. We used the CAMPS algorithm of \cite{ybnf-rjw8} with increasing maximum bond dimension and tracked the rotation index at which CAMPS failed to represent the running state with sufficient fidelity (here $38\%$). The results are depicted in Figure \ref{fig:camps_scaling}.

Notice that the {\bf Clifft} algorithm proposed in \cite{clifft} is just a full state vector version of the same algorithm. As such, the scaling and memory requirements are actually worse in the sense that from the beginning of the simulation such an approach would need to allocate the $2^{64}$ complex numbers describing a full state vector of $64$ qubits.

\begin{figure}[!t]
    \centering
    \includegraphics[width=0.8\linewidth]{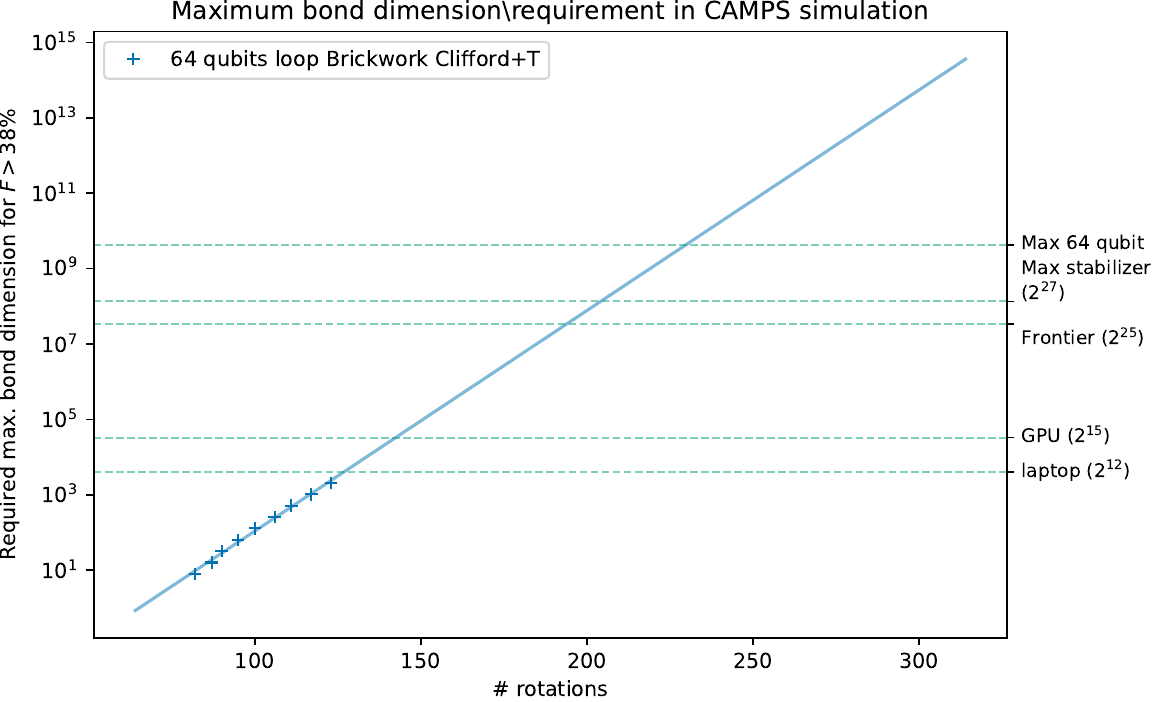}
    \caption{{\bf Bond dimension requirements for CAMPS.} Each point is obtained by running the CAMPS algorithm with fixed bond dimension and tracking the rotation index at which the CAMPS error increases beyond $1-0.38$. We stopped the extrapolation at 230 rotations since by that time the bond dimension would have saturated to the theoretical maximum of $2^{32}$.  Notice how the first $64$ rotations are virtually free and require $\chi=1$ since they were successfully disentangled. The bond dimension then grows exponentially at a rate of $2^{0.19 t}$. Since MPS' running time scales as $\chi^3$, this gives a running time scaling of $2^{0.57 t}$, worse than Clifford+T simulations, with the advantage of a very low prefactor due to the first $64$ $T$ gates being free.}
    \label{fig:camps_scaling}
\end{figure}

\subsection{Polynomial-time noisy simulations} 
Recent techniques~\cite{Aharonov_2023, Noh_2020} utilize the presence of noise in circuit implementations in order to achieve polynomial time simulation. However, the complexity of those approaches heavily depends on the noise rate of the simulated gates. Typically, the degree of their polynomial complexity will rapidly increase as we decrease the simulated error rate.
In order to reproduce the outcome of our experiment, one would need to achieve $38\%$ fidelity across a circuit of roughly $2.5$k gates. This is equivalent to simulating the circuit with a gate error rate of $\varepsilon = 1 - e^{\log{0.38}/2500}\approx 3.9\times10^{-4}$, which falls below reasonable complexities for those approaches (\cite{Noh_2020} reports simulations down to $\varepsilon=6\times10^{-2}$).
In practice, we ran the noisy MPO approach described in \cite{Noh_2020}, tracking MPO entropy after each layer and with a bond dimension limit of 256. The maximum bond dimension saturated after 4 layers, reaching a maximal MPO entropy of 8, which leads us to believe that this approach is not competitive at this noise rate.

\subsection{XEB spoofing through unfaithful simulation}

Another way to produce samples close to our distribution would be to try to simplify the quantum circuit in a simulable regime and return the corresponding samples~\cite{GaoKalinowski2024}.
One approach would be to replace some $T$ gates by some Clifford gates and simulate the resulting circuit. In order to maximize the process fidelity between the $T$ gates and their replacements, one can use $S$ gates, which achieve a relative fidelity of $0.8536$. To test this approach, we extracted slices of widths 12 and 14 in the middle of our full circuit and simulated the following experiment:
\begin{itemize}
    \item for $k\in [0, 40]$, pick $k$ $T$ gates at random and replace them with $S$ gates
    \item simulate the circuit and compute the resulting fidelity w.r.t. the fully doped circuit
\end{itemize}

Figure \ref{fig:spoofing} depicts the results of that experiment. The fidelity drops exponentially as $T$ gates are Cliffordized. As only a maximum of 5 gates can be converted into $S$ gates before the fidelity drops below the experimentally measured value, this will not offer a viable path to spoofing the experiment to obtain high XEB scores.

\begin{figure}[!t]
    \centering
    \includegraphics[width=0.5\linewidth]{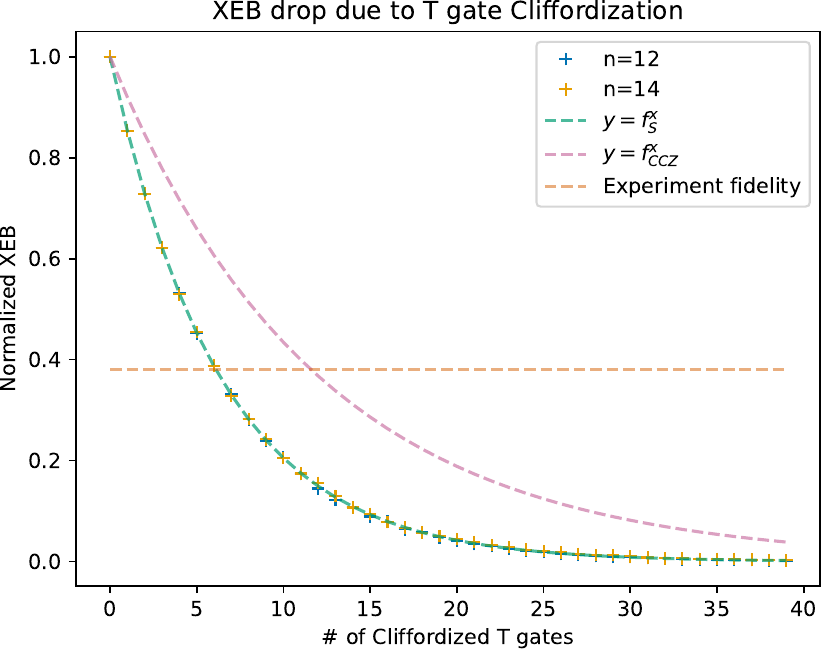}
    \caption{{\bf Attempt at spoofing via $T$ gate Cliffordization. } Each point tracks the average fidelity over 10 runs of replacing a fixed number of $T$ gates by $S$ gates. We also plot for reference the expected fidelity drop due to this replacement given the relative process fidelity between $T$ and $S$, $f_s = 0.8536$ and $CCZ$ to the closest Clifford, normalized to single $T$s, $f_{CCZ}=0.92$. Those scalings show that one can effectively Cliffordize 5 $T$ gates using $S$ gates, or up to 12 $T$ gates using $CCZ$ replacements,  and still achieve $F=0.38$. This brings the simulation requirements from $314$ $T$ gates down to $309$. Notice how the scaling is independent of the number of qubits.}
    \label{fig:spoofing}
\end{figure}

\end{document}